\documentclass{article}

\usepackage[preprint]{neurips_2026}

\usepackage[utf8]{inputenc} % allow utf-8 input
\usepackage[T1]{fontenc}    % use 8-bit T1 fonts
\usepackage{url}            % simple URL typesetting
\usepackage{booktabs}       % professional-quality tables
\usepackage{amsfonts}       % blackboard math symbols
\usepackage{nicefrac}       % compact symbols for 1/2, etc.
\usepackage{microtype}      % microtypography
\usepackage{xcolor}         % colors

\usepackage[pagebackref,breaklinks,colorlinks,linkcolor={blue},citecolor={blue}]{hyperref}

\usepackage{algorithmic}
\usepackage[ruled, vlined, linesnumbered]{algorithm2e}
\SetKw{KwRet}{return}

\usepackage{amsmath}

 \usepackage{amsthm}

\usepackage{graphicx}

\usepackage{multirow}
\usepackage{subcaption}
\usepackage{float}
\usepackage{enumitem}

\theoremstyle{plain}
\newtheorem{theorem}{Theorem}[section]

\theoremstyle{definition}

\theoremstyle{remark}

\title{EpiFormer: Learning Antigen-Antibody Interactions for Epitope Prediction via Geometric Deep Learning}

\author{
  Mansoor Ahmed~\textsuperscript{1,2\thanks{Address correspondence to: \quad \texttt{mahmed76@student.gsu.edu}, \quad \texttt{mpatterson30@gsu.edu} }}  ,
  Huirong Chai~\textsuperscript{1},
  Haoxin Wang\textsuperscript{1},
  Hemanth Venkateswara\textsuperscript{1}, \\
  \hfill  \textbf{Murray Patterson}\textsuperscript{1\tiny{*}} \hfill  \null \vspace{5pt} \\
  \hfill \textsuperscript{1}Georgia State University, Atlanta, GA, USA \hfill \null \\
  \hfill \textsuperscript{2}Georgia Institute of Technology, Atlanta, GA, USA \hfill \null \\
}

\begin{document}

\maketitle

\begin{abstract}

Antibodies neutralize foreign antigens by binding to specific surface regions called epitopes. Computational epitope prediction is critical for understanding immune recognition and guiding antibody engineering. However, existing methods face three fundamental challenges: antibody-aware models encode each chain independently and combine them only at a late stage, failing to capture co-dependent structural features that define binding interfaces, whereas severe class imbalance and scarcity of known antibody-antigen complexes render standard training objectives ineffective. We propose \emph{EpiFormer}, a general encoder-decoder framework that addresses these challenges jointly. Our key design principle is interleaved cross-attention within GNN encoding layers, enabling bidirectional antigen-antibody information flow throughout representation learning rather than only at the output. This early-fusion principle is backbone-agnostic, providing consistent gains across GNN architectures from simple GCNs to equivariant models. We further show that sparsity-aware objectives are effective when paired with early-fusion architectures for the epitope prediction task. \emph{EpiFormer} improves over the previous best method by over 40\% in F1 score on standard benchmarks, demonstrating generalizability and cross-dataset transferability. Notably, EpiFormer discovers known biological principles as emergent behaviors of end-to-end training, where the learned cross-attention gates favor antigen-to-antibody information flow, consistent with the asymmetric roles of the two chains at the binding interface, and the model's preference for geometric over evolutionary features aligns with the established finding that epitope residues are not evolutionarily conserved. The source code is available at: \url{https://github.com/mansoor181/epiformer.git}

\end{abstract}

\section{Introduction}

Antibodies are Y-shaped proteins that recognize and neutralize foreign substances by binding to specific surface regions on antigens called epitopes. Accurate identification of epitopes is critical for therapeutic antibody design, vaccine development, and understanding immune recognition~\citep{norman2020computational,joubbi2024antibody}. While traditional experimental approaches for epitope mapping are time-consuming and expensive, computational methods offer the potential for rapid screening of candidate therapeutics~\citep{hummer2022advances}. 

Computational epitope prediction methods can be categorized as \emph{antibody-agnostic} or \emph{antibody-aware}. Antibody-agnostic methods, including DiscoTope3~\citep{hoie2024discotope}, EpiGraph~\citep{choi2024b}, and GraphBepi~\citep{zeng2023graphbepi}, predict binding sites as intrinsic properties of the antigen without considering the specific antibody. However, epitopes are fundamentally antibody-specific, as different antibodies targeting the same antigen bind to different surface regions~\citep{liu2024asep}. Antibody-aware methods such as PECAN~\citep{pittala2020learning}, MIPE~\citep{wang2024improving}, and WALLE~\citep{liu2024asep} address this by predictions on the antibody structure.

However, antibody-aware epitope prediction faces three fundamental challenges. First, existing methods encode antigen and antibody independently and combine them only at a late stage (\emph{late fusion}), failing to capture co-dependent structural features that define binding interfaces. Antibody CDR loops adapt their conformation to the epitope surface, and epitope residues similarly adjust to accommodate the paratope. Encoding each chain without knowledge of its counterpart misses these interaction-specific geometric signatures. Second, epitope residues constitute fewer than 5\% of antigen surface residues, creating severe class imbalance that renders standard objectives ineffective. Third, the scarcity of known antibody-antigen complexes limits the data available for training.

We propose \emph{EpiFormer}, a general encoder-decoder framework that addresses these challenges jointly. Our key design principle is \emph{interleaved bidirectional cross-attention} within GNN encoding layers, enabling antigen-antibody information flow throughout representation learning rather than only at the output. This early-fusion principle is backbone-agnostic, providing consistent gains across GNN architectures from simple GCNs to equivariant models. We further develop sparsity-aware training objectives, including Dice loss, per-graph count regularization, bipartite edge prediction, and distance supervision, that are effective when paired with early-fusion architectures. Within this framework, we propose \emph{EGNN-R}, a multi-relational E(3)-equivariant GNN where each edge relation type has its own learned transformation, providing additional gains.

Our main contributions are:

\begin{enumerate}
    \item A \textbf{backbone-agnostic early-fusion framework} with interleaved bidirectional cross-attention within GNN encoding layers,  allowing the antigen and antibody representations to inform each other at every layer rather than only in the decoder.

    \item An \textbf{architecture-loss co-design principle}, where sparsity-aware objectives are developed to effectively address class imbalance and data scarcity for the epitope prediction task. 

    \item We show that compact geometric features outperform high-dimensional protein language model embeddings for epitope prediction, and that the learned cross-attention gates expose a consistent antigen-antibody asymmetry that aligns with the distinct roles of the two chains.
\end{enumerate}

\section{Related Work}\label{sc:related-work}

\textbf{Antibody-agnostic methods} predict epitopes without considering which antibody is binding. Early epitope prediction methods relied on sequence features such as hydrophilicity and accessibility~\citep{jespersen2017bepipred}. Structure-based approaches, such as DiscoTope3~\citep{hoie2024discotope}, EpiGraph~\citep{choi2024b}, and GraphBepi~\citep{zeng2023graphbepi},  improved upon these by incorporating 3D information. 

\textbf{Antibody-aware methods} condition predictions on the binding counterpart. PECAN~\citep{pittala2020learning} encodes antibody and antigen with separate graph convolutions and combines them via bilinear attention. MIPE~\citep{wang2024improving} uses multi-modal contrastive learning to align sequence and structure representations, applying multi-head attention for late-stage fusion. WALLE~\citep{liu2024asep} formulates epitope prediction as bipartite link prediction between antibody and antigen residue graphs. EpiScan~\citep{wang2024episcan} incorporates CDR masking into sequence embeddings. All of these methods share a \emph{late fusion} architecture in which antigen and antibody are encoded independently before cross-chain attention is applied.

% \paragraph{Geometric Deep Learning for Proteins.}
\textbf{E(3)-equivariant neural networks} preserve geometric properties under rotations and translations, making them well-suited for molecular modeling. EGNN~\citep{satorras2021egnn} achieves equivariance by updating coordinates along displacement vectors scaled by learned invariant functions. Equiformer~\citep{liao2022equiformer} and EquiformerV2~\citep{liao2023equiformerv2} extend this with SE(3)/SO(2) equivariant attention using spherical harmonics. GearNet~\citep{zhang2022gearnet} introduced multi-relational protein graphs with seven edge types encoding sequential and spatial relationships, but uses a shared message function across relations. Surface-based methods, including MaSIF~\citep{gainza2020deciphering} and AtomSurf~\citep{mallet2023atomsurf}, operate on molecular surface representations with geodesic or spectral convolutions.

% \paragraph{Protein-Ligand and Protein-Protein Methods.}
Several methods from adjacent domains share architectural similarities with \emph{EpiFormer}. \textbf{CheapNet}~\citep{limcheapnet} predicts protein-ligand binding affinity using geometry-informed graph networks (GIGN) for each entity, followed by cross-attention. This is similar to our use of separate encoders with cross-attention, but CheapNet applies cross-attention only at the final layer (late fusion), whereas \emph{EpiFormer} interleaves it at every encoder block. \textbf{EquiPocket}~\citep{zhang2023equipocket} uses E(3)-equivariant message passing for ligand binding site prediction but is antibody-agnostic and operates at atom-level rather than the residue-level. \textbf{DiffDock}~\citep{corso2022diffdock} employs SE(3)-equivariant score networks for molecular docking, demonstrating the value of equivariant architectures for protein-ligand interactions. \textbf{Boltz-1/2}~\citep{wohlwend2025boltz1,passaro2025boltz2} use Pairformer architectures with triangle attention for structure prediction; Boltz-2 adds cross-attention between binding counterparts, but applies it to pairwise representations rather than node embeddings.

\textbf{EpiFormer} differs from these methods by addressing their limitations for epitope prediction. Like CheapNet, we use separate encoders with cross-attention, but apply attention at every layer rather than only at the output. Like GearNet, we use multi-relational graphs, but learn separate transformations per relation type. Similar to EGNN and EquiPocket, we maintain E(3)-equivariance, but extend it to the multi-relational setting. Recently, ATProt~\citep{gao2024towards} applies cross-attention between antibody and antigen chains, but only once after all encoding layers are complete (late fusion). DiffDock-PP~\citep{ketata2023diffdockpp} uses SE(3)-equivariant tensor product convolutions for rigid-body docking but does not perform residue-level binding site prediction. To our knowledge, no prior method interleaves cross-attention within structural GNN encoding layers for protein-protein interaction prediction. While prior antibody-aware methods (PECAN, MIPE, ATProt) encode chains independently before late fusion, \emph{EpiFormer} enables cross-chain information flow throughout the encoding process.

\section{Methods}\label{sc:methods}

\subsection{Preliminaries }

\paragraph{Graph construction.}
We build two independent residue graphs $\mathcal{G}_{\text{ag}} = (\mathcal{V}_{\text{ag}}, \mathcal{E}_{\text{ag}}, \mathcal{R})$ and $\mathcal{G}_{\text{ab}} = (\mathcal{V}_{\text{ab}}, \mathcal{E}_{\text{ab}}, \mathcal{R})$ from the unbound antigen and antibody structures.
Vertex $v_i \in \mathcal{V}$ represents residue $i$, centered on C$_\alpha$ at coordinate $\mathbf{x}_i\in\mathbb{R}^3$ ($|\mathcal{V}_{\mathrm{ag}}| = n$, $|\mathcal{V}_{\mathrm{ab}}| = m$). Each node carries a geometric feature vector $\mathbf{h}_i\in\mathbb{R}^{d_h}$ (encoding residue type, backbone geometry, physicochemical properties) and a coordinate matrix $\mathbf{X}_i\in\mathbb{R}^{3\times 4}$ for four backbone atoms $\xi=\{\text{N}, \text{C}_{\alpha}, \text{C}_{\beta}, \text{O}\}$. Each edge $e_{i,j}$ carries a feature vector $\mathbf{f}_{i,j}\in\mathbb{R}^{d_f}$ (distances, angles) and a relation tuple $\mathbf{r}_{i,j}\subseteq\mathcal{R}$.
The relation set $\mathcal{R} = \{\rho_1, \rho_2, \rho_3, \rho_4\}$ captures sequential bonds ($\rho_1$), short-range coupling ($\rho_2$), $K$-nearest neighbors ($\rho_3$), and medium-range contacts within 8\,\text{\AA}~($\rho_4$). We denote a residue graph as $\mathcal{G} = (\mathbf{H}, \mathbf{X}, \mathbf{F}, \mathbf{R})$; see Appendix~\ref{appendix:graph-construction} for details.

\paragraph{Problem Formulation.}
We formulate epitope prediction as binary node classification. A residue \( v \in \mathcal{V}_{\text{ag}} \) is labeled as an epitope if it is within 4.5\AA\ of any residue in \( \mathcal{V}_{\text{ab}} \), and non-epitope otherwise. The classifier $f: v_{\text{ag}} \rightarrow \{0, 1\}$ predicts $\hat{y}_{{\text{ag}}} = f(v_{\text{ag}}; \mathcal{G}_{\text{ag}}, \mathcal{G}_{\text{ab}})$.

% \mansoor{ The Bipartite graph link prediction seems not to be a task in the paper, but it is formulated in the 3.1. }

% \underline{\emph{Bipartite graph link prediction:}} This task predicts the bipartite adjacency matrix ${\hat{\mathcal{E}}}_{\text{bg}}$ between antibody and antigen in the bipartite graph $\mathcal{G}_{\text{bg}} = (\mathcal{V}_{\text{ag}} \cup \mathcal{V}_{\text{ab}}, \mathcal{E}_{\text{bg}})$, where 
% $\mathcal{V}_{\text{ag}}$ and $\mathcal{V}_{\text{ab}}$ are disjoint vertex sets, and 
% $\mathcal{E}_{\text{bg}} \subseteq \mathcal{V}_{\text{ag}} \times \mathcal{V}_{\text{ab}} \in \{{0,1\}}^{n\times m}$ denotes inter-molecular contacts between antigen and antibody. 
% An edge \( e_{\text{bg}} \in \mathcal{E}_{\text{bg}} \) is a contact (labeled as $1$) if the corresponding residues ($v_{\text{ag}},v_{\text{ab}}$) are within 4.5\AA\ of each other and $0$ otherwise. 
% The edge classifier \( g: e_{\text{bg}} \rightarrow \{0, 1\} \) is defined as:
% \begin{equation}\label{eq:edge-pred}
% {\hat{\mathcal{E}}}_{\text{bg}} \quad = \quad g(e_{\text{bg}}; \mathcal{G}_{\text{bg}}) \quad = \quad
% \begin{cases}
% 1 & \text{if } e_{\text{bg}} \text{ is a contact}, \\
% 0 & \text{otherwise}.
% \end{cases}
% \end{equation}

\paragraph{Equivariance.}
Since molecular properties are unchanged under rigid body transformations, geometric GNNs incorporate E(3)-equivariance as an inductive bias~\citep{jiao2023energy}. For coordinates $\mathbf{X} \in \mathbb{R}^{3 \times m}$ and scalar features $\mathbf{h} \in \mathbb{R}^{d}$, an E(3)-equivariant function satisfies $f(g \cdot \mathbf{X}, \mathbf{h}) = g \cdot f(\mathbf{X}, \mathbf{h})$ for all $g \in \mathrm{E}(3)$, where group actions are translations ($g \cdot \mathbf{X} = \mathbf{X} + \mathbf{b}$) or rotations/reflections ($g \cdot \mathbf{X} = \mathbf{O} \mathbf{X}$, $\mathbf{O} \in \mathrm{O}(3)$). E(3)-invariant functions instead satisfy $f(g \cdot \mathbf{X}, \mathbf{h}) = f(\mathbf{X}, \mathbf{h})$.

\subsection{EpiFormer}

In this section, we present the architecture of \emph{EpiFormer}, an encoder-decoder framework for antibody-antigen binding-site prediction.
The model receives two disjoint multi-relational residue graphs, $\mathcal{G}_{\text{ag}}$ and $\mathcal{G}_{\text{ab}}$, processes them with independent E(3)-equivariant encoders that produce residue-level embeddings, and passes these embeddings to a cross-attention decoder that reconstructs the bipartite adjacency matrix $\hat{\mathcal{E}}_{\text{bg}} \in \{{0,1\}}^{n\times m}$.
The overall workflow is in Figure~\ref{fig:epiformer} and the algorithm in Appendix~\ref{alg:epiformer-pseudocode}.

\begin{figure*}[t]
    \centering
    \includegraphics[width=1\textwidth]{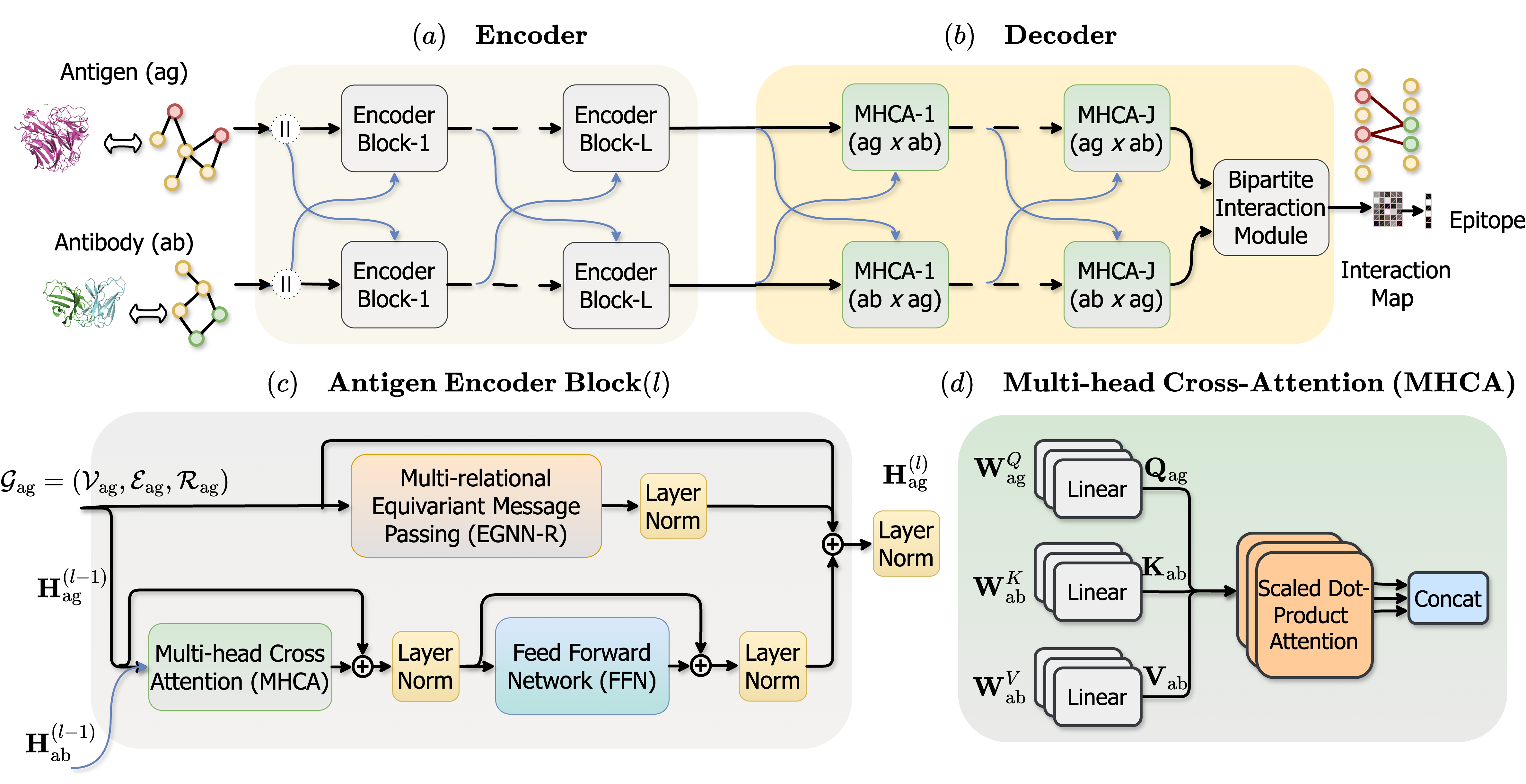}

    \caption{Overview of \emph{EpiFormer}. (a) Parallel EGNN-R encoders with interleaved cross-attention. (b) Bidirectional cross-attention decoder producing the interaction map. (c) Encoder block schematic (``$\oplus$" = addition). (d) Multi-head cross-attention between antigen and antibody residues.}
    \label{fig:epiformer}
\end{figure*}

% \textcolor{purple}{H: Rephrase. Some parts are not clear} The overall schematic of \emph{EpiFormer} for antibody-aware epitope prediction. The model takes an antigen multi-relational graph $\mathcal{G}_{\text{ag}}=(\mathcal{V}_{\text{ag}}, \mathcal{E}_{\text{ag}}, \mathcal{R})$ and an antibody multi-relational graph $\mathcal{G}_{\text{ab}}=(\mathcal{V}_{\text{ab}}, \mathcal{E}_{\text{ab}}, \mathcal{R})$, passes them two independent encoder modules \textbf{(a)}, each with its own weights $\Theta_{\text{ag}}$ and $\Theta_{\text{ab}}$. The input antigen and antibody attend to each other using cross attention in the encoder modules and finally pass through a bi-directional multi-head cross attention decoder \textbf{(b)} $\Theta_{\text{attn}}$ to produce the interaction map as output, represented as a bipartite graph. MHCA in the decoder represents multi-head cross attention, while FFN represents a feed-forward network. \textbf{(c)} The schematic of the architecture of an \emph{EpiFormer} encoder block. In this figure, ``$\oplus$" denotes addition. \textbf{(d)} The schematic of the multi-head cross attention (MHCA) between the antigen and antibody residues (ab \textbf{x} ag or vice versa). The queries and keys are swapped when one attends to the other and are analogous to each other.

% \textcolor{red}{mansoor: remove parameter sets from caption -- modify figure to have input multi-relational graph -- also denote incoming G for layer l-1 and output as l in the superscript -- title of (c) as encoder block (L) -- (also maybe a figure for EGNN-R if we have time)}

\textbf{Encoder.}\label{sec:encoder}
\emph{EpiFormer} contains two parallel encoders with no shared parameters, one dedicated to the antigen chain and the other to the antibody chain (Figure~\ref{fig:epiformer}a). Each node encodes Cartesian coordinates $\mathbf{x}_i\in\mathbb{R}^3$ and geometric descriptors $\mathbf{h}^{\text{geo}}_i \in \mathbb{R}^{d_{\text{geo}}}$, which are projected to the working width $d_h$:
\begin{equation}\label{eq:node-init}
\mathbf{h}_i^{0} = \mathbf{W}_{\text{geo}}\, \mathbf{h}^{\text{geo}}_i \in \mathbb{R}^{d_h}.
\end{equation}
The vector $\mathbf{h}_i^{0}$ serves as the initial node state for the first \emph{EpiFormer} encoder block. The schematic of an \emph{EpiFormer} block is shown in Figure~\ref{fig:epiformer}~\textbf{(c)}. Let $\mathbf{H}_{\text{ag}}^{\ell} \in \mathbb{R}^{n \times d_h}$ and $\mathbf{H}_{\text{ab}}^{\ell} \in \mathbb{R}^{m \times d_h}$ be the current embeddings, which are passed in parallel to their EGNN-R and MHCA layers. 
% The pseudocode of the encoder is presented in Algorithm~\ref{alg:epiformer_encoder}.

% \textcolor{purple}{H: Label all equations}:

\underline{\emph{\textbf{EGNN-R layer}}:}\label{sec:egnnr} We develop a relation-aware variant of EGNN~\citep{satorras2021egnn} to propagate structural and geometric information within each chain. Unlike shared-weight approaches, EGNN-R uses per-relation message functions because sequential edges (covalent constraints) and spatial edges (non-covalent contacts) represent fundamentally different physics. Let $\mathbf{h}_i^{\ell}\in\mathbb{R}^{d_h}$, $\mathbf{x}_i^{\ell}\in\mathbb{R}^3$ be the feature and coordinate of residue $i$ at layer $\ell$, $d_{ij}=\lVert\mathbf{x}_i^{\ell}-\mathbf{x}_j^{\ell}\rVert_2^2$ the squared distance, and $\boldsymbol{\delta}_{ij}=\mathbf{x}_i^{\ell}-\mathbf{x}_j^{\ell}$ the displacement:
% \murray{ could put mij (3) and sij(4) on same line to save space}

% \begin{align}
%     m_{ij}^{\rho}      &= \phi_m^{\rho} \bigl(\mathbf{h}_i^{\ell}, \mathbf{h}_j^{\ell}, \gamma(d_{ij}), \mathbf{f}_{ij}\bigr),  \qquad
%    & \mathbf{h}_i^{(\ell+1)} = \mathbf{h}_i^{\ell} + \phi_h \Bigl(\mathbf{h}_i^{\ell}, \sum_{j\in\mathcal{N}(i)}\;\sum_{\rho\in\mathbf{r}_{ij}} m_{ij}^{\rho}\Bigr), \label{eq:node_update} \\
%     s_{ij}^{\rho}      &= \phi_x^{\rho} \bigl(m_{ij}^{\rho}\bigr),  \qquad
%     & \mathbf{x}_i^{(\ell+1)} = \mathbf{x}_i^{\ell} + \sum_{j\in\mathcal{N}(i)}\;\sum_{\rho\in\mathbf{r}_{ij}}
%                 \frac{\boldsymbol{\delta}_{ij}}{\sqrt{d_{ij}+\varepsilon}}\;s_{ij}^{\rho}. \label{eq:coord_final}
% \end{align}

\begin{align}
    m_{ij}^{\rho}      &= \phi_m^{\rho} \bigl(\mathbf{h}_i^{\ell}, \mathbf{h}_j^{\ell}, \gamma(d_{ij}), \mathbf{f}_{ij}\bigr), \qquad
    s_{ij}^{\rho}      = \phi_x^{\rho} \bigl(m_{ij}^{\rho}\bigr), \label{eq:msg} \\
    \mathbf{h}_i^{(\ell+1)} &= \mathbf{h}_i^{\ell} + \phi_h \Bigl(\mathbf{h}_i^{\ell}, {\textstyle\sum_{j\in\mathcal{N}(i)}\sum_{\rho\in\mathbf{r}_{ij}}} m_{ij}^{\rho}\Bigr), \quad
     \mathbf{x}_i^{(\ell+1)} = \mathbf{x}_i^{\ell} + {\textstyle\sum_{j\in\mathcal{N}(i)}\sum_{\rho\in\mathbf{r}_{ij}}}
                \frac{\boldsymbol{\delta}_{ij}}{\sqrt{d_{ij}+\varepsilon}}\;s_{ij}^{\rho}. \label{eq:updates}
\end{align}

% \murray{is $\varepsilon$ just an arbitrarily small value to avoid divide-by-zero error?  If so, no need to mention, but if not, maybe mention it} 
Here, $\gamma(\cdot)$ is a 16-term radial basis function, $\mathbf{f}_{ij}$ the edge attribute vector, and each $\phi_{\{m,x\}}^{\rho}$ a two-layer MLP with relation-specific parameters ($\varepsilon=10^{-8}$). We have four message MLPs $\phi_m^{\rho}: \mathbb{R}^{2d_h + d_f + 16} \to \mathbb{R}^{d_q}$, four coordinate MLPs $\phi_x^{\rho}: \mathbb{R}^{d_q} \to \mathbb{R}^3$, and a shared node update MLP $\phi_h: \mathbb{R}^{d_h + d_q} \to \mathbb{R}^{d_h}$. With residual connections and layer normalization:
\begin{align}
\mathbf{H}_{\mathrm{ag}}^{\mathrm{intra}}=\{\,W_{\mathrm{ag}}^{\ell}\mathbf{h}_i^{\ell}\mid v_i\in\mathcal{V}_{\mathrm{ag}}\}, \qquad
\mathbf{H}_{\mathrm{ab}}^{\mathrm{intra}}=\{\,W_{\mathrm{ab}}^{\ell}\mathbf{h}_j^{\ell}\mid v_j\in\mathcal{V}_{\mathrm{ab}}\},
\end{align}

% \textcolor{red}{mansoor: the framework is E(3)‑equivariant, yet only the EGNN‑R layers update coordinates; cross‑attention consumes invariant features. The outputs (contact map, node probabilities) should be invariant, not equivariant. EGNN is equivariant is important because we are using coordinates, the origin changes but the model still work, state that we tested it out empirically, and the result reflects the equivariant}

where $W^{\ell}$ are per-layer trainable parameters. The layer is $\mathrm{E}(3)$-equivariant by construction since only the displacement $\boldsymbol{\delta}_{ij}$ enters the coordinate update, while cross-attention operates on invariant features (proof in Appendix~\ref{thm:egnn_equivariance}).

\underline{\emph{\textbf{MHCA layer with feed-forward network}}:} In parallel to geometric message passing, each encoder block applies bidirectional multi-head cross-attention (MHCA)~\citep{vaswani2017attention} to enable inter-chain communication. %The output from the MHCA layer $\widetilde{\mathbf{H}}$ passes through layer normalization, a two-layer feed-forward (FFN) network, and a normalization layer with a residual connection to the MHCA output. %
The MHCA mechanism shown in Figure~\ref{fig:epiformer} \textbf{(d)} produces cross-chain context representations $\widetilde{\mathbf{H}}_{\mathrm{ag}}$ and $\widetilde{\mathbf{H}}_{\mathrm{ab}}$.
A learnable scalar gate $\alpha$ balances intra-chain geometry with cross-chain context:
\begin{align}
    \mathbf{H}_{c}^{(\ell+1)} = \mathbf{H}_{c}^{\ell} + \mathbf{H}_{c}^{\mathrm{intra}} + \alpha_{c}\,\operatorname{FFN}(\widetilde{\mathbf{H}}_{c}), \quad c\in\{\mathrm{ag},\mathrm{ab}\},
\end{align}
where $\alpha_{\mathrm{ag}},\alpha_{\mathrm{ab}}\in\mathbb{R}^+$ are learnable parameters, $\widetilde{\mathbf{H}} = \text{MHCA}(\mathbf{H})$, and FFN is a two-layer Feed Forward Network. The MHCA is detailed in Appendix~\ref{sec:mhca}.
%This repeated stack of parallel EGNN-R and MHCA layers produces residue representations after $L$ layers that simultaneously capture fine-grained geometric detail, rich sequence semantics, and long-range cross-chain interaction patterns, thereby providing the decoder with a comprehensive description of the antibody-antigen complex.%

% \mansoor{The model architecture contains the Equivarant network for monomers and uses cross attention for interaction between monomers. Many papers share a similar idea, such as DiffDock, DynamicBind, and NeuralMD. What do the authors think are the biggest advantages of the proposed model to extract the complex feature compared to other similar architectures?}

\paragraph{Decoder.}\label{sec:decoder}

The decoder refines the residue embeddings produced by the encoder and performs bipartite interaction prediction. It consists of $J$ identical layers, each containing bidirectional MHCA with FFN and residual connections, followed by a bipartite interaction head. The final embeddings $\mathbf{H}^{J}_{\mathrm{ag}}$, $\mathbf{H}^{J}_{\mathrm{ab}}$ are projected into queries and keys of width $d_k$ in both directions:
\begin{align}
    \mathbf{S}_{\mathrm{ag}\to\mathrm{ab}} = \frac{(\mathbf{H}^{J}_{\mathrm{ag}}\mathbf{W}^{\text{out}}_{Q})(\mathbf{H}^{J}_{\mathrm{ab}}\mathbf{W}^{\text{out}}_{K})^{\top}}{\sqrt{d_k}}, \qquad
    \mathbf{S}_{\mathrm{ab}\to\mathrm{ag}} = \frac{(\mathbf{H}^{J}_{\mathrm{ab}}\mathbf{W}^{\prime\,\text{out}}_{Q})(\mathbf{H}^{J}_{\mathrm{ag}}\mathbf{W}^{\prime\,\text{out}}_{K})^{\top}}{\sqrt{d_k}}.
\end{align}
The two score maps are fused via a learnable mixing vector $\mathbf{w}\in\mathbb{R}^{2}$ and bias $b\in\mathbb{R}$ to produce logits $\mathbf{Z}=\mathbf{w}^{\top}[\,\mathbf{S}_{\mathrm{ag}\to\mathrm{ab}}\ (\mathbf{S}_{\mathrm{ab}\to\mathrm{ag}})^{\top}\,]+b$, and the interaction probabilities are $\hat{\mathcal{E}}_{\text{bg}}=\sigma(\mathbf{Z})\in\mathbb{R}^{n\times m}$.
% The per-residue epitope probabilities are then read from $\hat{\mathcal{E}}_{\text{bg}}$ by pooling the top-$k$ probabilities from each column of the bipartite adjacency matrix, and thresholding these probabilities produces binary labels for the downstream epitope prediction task. Here, $k$ represents the average number of bipartite edges for the binding residues derived as an empirical prior and is fixed during evaluation.

% and is set to 2 as a prior computed from the dataset, representing that, on average, an antigen binding residue has two edges with its corresponding antibody binding residue. \textcolor{teal}{Huirong: mention k=2 is fixed during evaluation and is derived from prior calculation and is an empirical prior.}

% \textcolor{red}{todo: also add the distance heads to the decoder block}

\subsection{Joint objective}

% \mansoor{ For the losses, the core task is actually Node Classification Loss, but it puts the Edge Prediction Loss, Dice Loss for Graph Segmentation, and Per-Graph Sparsity Regularization together, which seems to the model needs to address those problems. Maybe it is better to merge them into the Auxiliary losses. }

% \mansoor{Some formulas and superscripts, and subscripts are too cumbersome; please consider making them succinct.}

\emph{EpiFormer} is trained with a joint objective that combines the primary epitope node classification loss with auxiliary terms for bipartite edge reconstruction and inter-chain geometric classification.
The overall training objective is a weighted sum of these components:
\begin{align}
\mathcal{L} =  
\lambda_{\mathrm{node}}\,\mathcal{L}_{\mathrm{node}} + 
\lambda_{\mathrm{edge}}\,\mathcal{L}_{\mathrm{edge}} + \lambda_{\mathrm{geo}}\,\mathcal{L}_{\mathrm{geo}}.
\end{align}

\paragraph{Node Classification Loss ($\mathcal{L}_{\mathrm{node}}$).}
The node classification loss supervises epitope nodes only and combines three complementary objectives to handle class imbalance and enforce structural priors:
\begin{align}
\mathcal{L}_{\mathrm{node}} = \beta_{\mathrm{BCE}}\,\mathcal{L}_{\mathrm{BCE}}^{\mathrm{epi}} + \beta_{\mathrm{Dice}}\,\mathcal{L}_{\mathrm{Dice}}^{\mathrm{epi}} + \beta_{\mathrm{sparsity}}\,\mathcal{L}_{\mathrm{sparsity}}^{\mathrm{epi}},
\end{align}
where $\beta_{\{\cdot\}}$ weight the different terms.
The probability that node $v_{\text{ag}}$ is an epitope
is derived from the bipartite interaction matrix by aggregating across the antibody dimension:
\begin{align}
(\hat{y}_{\text{ag}})_{i} = \frac{1}{m}\sum_{j=1}^{m}(\hat{\mathcal{E}}_{\text{bg}})_{ij},
\end{align}
where $m = |\mathcal{V}_{\text{ab}}|$ is the number of antibody residues.

\underline{\emph{\textbf{Class-Reweighted Binary Cross-Entropy}:}}
The primary classification loss applies positive class reweighting ( $\pi_{\mathrm{epi}} > 1$) to address the severe class imbalance in epitope prediction:
\begin{align}
    \mathcal{L}_{\mathrm{BCE}}^{\mathrm{epi}} = -\,\frac{1}{n}\sum_{i=1}^{n}   \bigl[\pi_{\mathrm{epi}}\,(y_{\text{ag}})_i\log (\hat{y}_{\text{ag}})_i + (1-(y_{\text{ag}})_i)\log (1-(\hat{y}_{\text{ag}})_i)\bigr],
\end{align}
%where $({y}_{\text{ab}})_{i} \in \{0,1\}$ are ground-truth epitope labels, $(\hat{y}_{\text{ab}})_{i} \in [0,1]$ are predicted probabilities, and $\pi_{\mathrm{epi}} > 1$ is a positive class weight that compensates for the rarity of epitope residues (typically 5-15\% of antigen residues).
%
\underline{\emph{\textbf{Dice Loss:}}}
The Dice loss treats epitope prediction as a segmentation problem~\citep{sudre2017dice}, directly measuring the overlap between predicted and true epitope regions:
\begin{align}
\mathcal{L}_{\mathrm{Dice}}^{\mathrm{epi}} = 1 - \frac{2\sum_{i} (\hat{y}_{\text{ag}})_i\,(y_{\text{ag}})_i + \epsilon}{\sum_{i} (\hat{y}_{\text{ag}})_i + \sum_{i} (y_{\text{ag}})_i + \epsilon},
\end{align}
where $\epsilon > 0$ is a smoothing constant.

\underline{\emph{\textbf{Sparsity Regularization:}}}
The sparsity term enforces cardinality matching between predicted and true epitope counts: $\mathcal{L}_{\mathrm{sparsity}}^{\mathrm{epi}} = \| \hat{y}_{\text{ag}} - y_{\text{ag}} \|_1$.

% \textcolor{purple}{H: Need to elaborate on these terms - not clear.}

% density-plots figure moved to Appendix
% \textcolor{red}{mansoor: need to provide epitope label definition }

\paragraph{Edge Prediction Loss ($\mathcal{L}_{\mathrm{edge}}$).}

This loss applies positive-class-reweighted binary cross-entropy over all antigen-antibody residue pairs:
\begin{align}
    \mathcal{L}_{\mathrm{edge}} = -\,\frac{1}{nm}\sum_{i,j} \bigl[\pi_{\mathrm{edge}}\,E_{ij}\log \hat{E}_{ij} + (1-E_{ij})\log(1-\hat{E}_{ij})\bigr],
\end{align}
where $E_{ij}=(\mathcal{E}_{\text{bg}})_{ij}\in\{0,1\}$ are ground-truth contacts (1 if residues $v_{\text{ag}},v_{\text{ab}}$ are within 4.5\AA), $\hat{E}_{ij}=(\hat{\mathcal{E}}_{\text{bg}})_{ij}$ are predicted probabilities, and $\pi_{\mathrm{edge}}$ compensates for the extreme sparsity of positives.

\paragraph{Auxiliary Distance Classification Loss ($\mathcal{L}_{\mathrm{geo}}$).}
The auxiliary geometric term classifies inter-chain distances into discrete bins ($\{0\text{--}4, 4\text{--}8, 8\text{--}16, 16\text{--}32\}\,\text{\AA}$), helping the model learn distance-aware representations. The loss applies class-balanced cross-entropy over near-contact antigen-antibody residue pairs (details in Appendix~\ref{appendix:geo_loss}).

\section{Experiments}\label{sc:experiments}

\subsection{Settings}

\paragraph{Dataset.}
We use the AsEP dataset~\citep{liu2024asep}, the largest benchmark of antibody-antigen complexes for epitope prediction, retaining 1,721 complexes after preprocessing (details in Appendix~\ref{appendix:preprocessing}). Each residue is represented by a geometric feature vector encoding amino-acid type, secondary structure, solvent accessibility, and local spatial geometry (Appendix~\ref{appendix:graph-construction}).

% \paragraph{Splits.} 
We adopt two pre-defined splitting strategies from AsEP~\citep{liu2024asep}: the \textbf{epitope-ratio split} stratifies complexes by the fraction of epitope residues, balancing task difficulty across splits; the \textbf{epitope-group split} clusters complexes by epitope identity, completely excluding test epitopes from training to evaluate generalization to novel binding sites. 

% Both use an 80/10/10 allocation (1,381 train / 170 val / 170 test).

% \FloatBarrier
\begin{table}[h!]
\centering

\small
\caption{Epitope-ratio split results (mean $\pm$ std over 3 seeds). Epitope-group results are in Table~\ref{tab:baseline_epigroup}.}
\label{tab:baseline_comparison}
\scriptsize
\begin{tabular}{lcccccc}
\toprule
\textbf{Method} & \textbf{AUC}$\uparrow$ & \textbf{AUPRC}$\uparrow$ & \textbf{F1}$\uparrow$ & \textbf{MCC}$\uparrow$ & \textbf{Prec.}$\uparrow$ & \textbf{Rec.}$\uparrow$ \\
\midrule
\multicolumn{7}{l}{\textit{\textbf{Epitope and paratope prediction}}} \\
EpiGraph~\citep{choi2024b} & .790{\scriptsize$\pm$.004} & .222{\scriptsize$\pm$.011} & .064{\scriptsize$\pm$.032} & .089{\scriptsize$\pm$.035} & .339{\scriptsize$\pm$.052} & .035{\scriptsize$\pm$.018} \\
EpiScan~\citep{wang2024episcan} & .669{\scriptsize$\pm$.024} & .135{\scriptsize$\pm$.019} & .203{\scriptsize$\pm$.052} & .124{\scriptsize$\pm$.043} & .148{\scriptsize$\pm$.072} & .327{\scriptsize$\pm$.041} \\
MIPE~\citep{wang2024improving} & .827{\scriptsize$\pm$.035} & \underline{.409{\scriptsize$\pm$.024}} & \underline{.337{\scriptsize$\pm$.071}} & \underline{.356{\scriptsize$\pm$.061}} & \textbf{.637{\scriptsize$\pm$.146}} & .229{\scriptsize$\pm$.047} \\
WALLE~\citep{liu2024asep} & .808{\scriptsize$\pm$.005} & .206{\scriptsize$\pm$.011} & .203{\scriptsize$\pm$.002} & .208{\scriptsize$\pm$.003} & .114{\scriptsize$\pm$.002} & \textbf{.926{\scriptsize$\pm$.006}} \\
DiscoTope3~\citep{hoie2024discotope} & .821{\scriptsize$\pm$.003} & .278{\scriptsize$\pm$.010} & .251{\scriptsize$\pm$.002} & .247{\scriptsize$\pm$.003} & .149{\scriptsize$\pm$.001} & .797{\scriptsize$\pm$.006} \\
GraphBepi~\citep{zeng2023graphbepi} & .783{\scriptsize$\pm$.021} & .232{\scriptsize$\pm$.014} & .079{\scriptsize$\pm$.005} & .118{\scriptsize$\pm$.024} & \underline{.430{\scriptsize$\pm$.031}} & .044{\scriptsize$\pm$.076} \\
PECAN~\citep{pittala2020learning} & .740{\scriptsize$\pm$.022} & .156{\scriptsize$\pm$.019} & .229{\scriptsize$\pm$.041} & .182{\scriptsize$\pm$.027} & .150{\scriptsize$\pm$.052} & .488{\scriptsize$\pm$.039} \\
\midrule
\multicolumn{7}{l}{\textit{\textbf{Protein binding-site prediction}}} \\
EquiPocket~\citep{zhang2023equipocket} & .761{\scriptsize$\pm$.017} & .203{\scriptsize$\pm$.027} & .202{\scriptsize$\pm$.010} & .183{\scriptsize$\pm$.033} & .116{\scriptsize$\pm$.023} & .799{\scriptsize$\pm$.029} \\
AtomSurf~\citep{mallet2023atomsurf} & .606{\scriptsize$\pm$.022} & .094{\scriptsize$\pm$.048} & .139{\scriptsize$\pm$.004} & .062{\scriptsize$\pm$.007} & .076{\scriptsize$\pm$.042} & .817{\scriptsize$\pm$.072} \\
ESMBind~\citep{schreiber2023esmbind} & .768{\scriptsize$\pm$.037} & .192{\scriptsize$\pm$.006} & .208{\scriptsize$\pm$.042} & .192{\scriptsize$\pm$.018} & .120{\scriptsize$\pm$.052} & .798{\scriptsize$\pm$.026} \\
\midrule
\multicolumn{7}{l}{\textit{\textbf{Protein interaction and docking}}} \\
ATProt~\citep{gao2024towards} & .798{\scriptsize$\pm$.001} & .236{\scriptsize$\pm$.003} & .252{\scriptsize$\pm$.006} & .237{\scriptsize$\pm$.002} & .152{\scriptsize$\pm$.005} & .736{\scriptsize$\pm$.027} \\
DiffDock-PP~\citep{ketata2023diffdockpp} & .735{\scriptsize$\pm$.037} & .148{\scriptsize$\pm$.024} & .192{\scriptsize$\pm$.019} & .169{\scriptsize$\pm$.032} & .109{\scriptsize$\pm$.012} & .812{\scriptsize$\pm$.036} \\
DiffDock~\citep{corso2022diffdock} & .677{\scriptsize$\pm$.011} & .123{\scriptsize$\pm$.035} & .154{\scriptsize$\pm$.016} & .089{\scriptsize$\pm$.005} & .143{\scriptsize$\pm$.027} & .168{\scriptsize$\pm$.004} \\
\midrule
\multicolumn{7}{l}{\textit{\textbf{Structure and affinity prediction}}} \\
Boltz-1~\citep{wohlwend2025boltz1} & .771{\scriptsize$\pm$.021} & .242{\scriptsize$\pm$.032} & .257{\scriptsize$\pm$.018} & .223{\scriptsize$\pm$.003} & .155{\scriptsize$\pm$.064} & .750{\scriptsize$\pm$.017} \\
Boltz-2~\citep{passaro2025boltz2} & .798{\scriptsize$\pm$.014} & .316{\scriptsize$\pm$.038} & .255{\scriptsize$\pm$.016} & .232{\scriptsize$\pm$.019} & .152{\scriptsize$\pm$.011} & .815{\scriptsize$\pm$.007} \\
AlphaFold3~\citep{abramson2024alphafold3} & 0.562 & 0.077 & 0.153 & 0.071 & 0.097 & 0.361 \\
CheapNet~\citep{limcheapnet} & .719{\scriptsize$\pm$.022} & .194{\scriptsize$\pm$.029} & .184{\scriptsize$\pm$.009} & .149{\scriptsize$\pm$.017} & .105{\scriptsize$\pm$.006} & .765{\scriptsize$\pm$.017} \\
GearBind~\citep{cai2024gearbind} & .799{\scriptsize$\pm$.012} & .215{\scriptsize$\pm$.031} & .244{\scriptsize$\pm$.027} & .236{\scriptsize$\pm$.002} & .144{\scriptsize$\pm$.040} & .787{\scriptsize$\pm$.037} \\
\midrule
\multicolumn{7}{l}{\textit{\textbf{Molecular property prediction}}} \\
EquiformerV2~\citep{liao2023equiformerv2} & \underline{.827{\scriptsize$\pm$.011}} & .264{\scriptsize$\pm$.031} & .260{\scriptsize$\pm$.016} & .262{\scriptsize$\pm$.022} & .154{\scriptsize$\pm$.013} & \underline{.823{\scriptsize$\pm$.022}} \\
\midrule
\textbf{\emph{EpiFormer} (ours)} & \textbf{.924\scriptsize{$\pm$.003}} & \textbf{.493\scriptsize{$\pm$.012}} & \textbf{.482\scriptsize{$\pm$.011}} & \textbf{.464\scriptsize{$\pm$.009}} & \underline{.363\scriptsize{$\pm$.014}} & .720\scriptsize{$\pm$.012} \\
\bottomrule
\end{tabular}
\end{table}

\textbf{Baseline Methods.} We compare against 20 methods spanning five categories. The first group is dedicated epitope and paratope predictors (EpiGraph, EpiScan, MIPE, WALLE, DiscoTope3, GraphBepi, PECAN). Because labeled antibody-antigen complexes are scarce, we also adapt models from related structural-biology tasks: general protein binding-site prediction (EquiPocket, AtomSurf, ESMBind), protein interaction and docking (ATProt, DiffDock-PP, DiffDock), structure and affinity prediction (Boltz-1, Boltz-2, AlphaFold3, CheapNet, GearBind), and molecular property prediction (EquiformerV2). All baselines were retrained on AsEP with their published configurations; training details are in Appendix~\ref{appendix:baselines}.

\subsection{Main Results}

\emph{EpiFormer} leads on all primary metrics across both evaluation splits. On the epitope-ratio split (Table~\ref{tab:baseline_comparison}), it achieves 0.924 AUC and 0.482 F1, surpassing the next-best method MIPE (0.337 F1) by a wide margin. On the harder epitope-group split (Table~\ref{tab:baseline_epigroup}), \emph{EpiFormer} reaches 0.826 AUC and 0.305 F1, maintaining a clear margin over all baselines.

Several patterns emerge from these results. Antigen-only methods such as EpiGraph, GraphBepi, and DiscoTope3 reach reasonable AUC but low F1, because they cannot condition on the specific antibody and therefore predict generic binding-prone surface rather than the true epitope. Among antibody-aware methods, MIPE is the strongest baseline. Its edge comes largely from protein language model embeddings (ESM and AbLang), though its high variance ($\pm$0.071 F1) points to sensitivity to initialization. Methods from adjacent domains, such as Boltz-1/2, GearBind, and EquiformerV2, are competitive on AUC but have no mechanism for the severe class imbalance of epitope prediction, which keeps their F1 low. The docking model DiffDock falls behind much simpler baselines, and ATProt, which fuses the two chains only after encoding, trails \emph{EpiFormer} by a wide margin.

These differences are clearest in the balance between precision and recall. Several baselines reach high recall only by labeling most residues as positive; WALLE, for instance, recovers almost every epitope residue (0.926 recall) at the cost of very low precision (0.114). The Dice loss and sparsity regularizer in \emph{EpiFormer} instead produce balanced predictions (0.363 precision, 0.720 recall), and this balance is what drives its higher F1 and MCC.

\subsection{Analysis}

\textbf{Backbone-Agnostic Framework.}
The gains come more from the interleaved cross-attention framework rather than from the choice of GNN backbone. Replacing EGNN-R with a plain GCN~\citep{kipf2016semi}, while keeping the decoder, cross-attention, and losses fixed, still reaches 0.447 F1 and already surpasses all 19 baselines (Table~\ref{tab:gnn_backbone}). GAT~\citep{velivckovic2017gat} and RGCN~\cite{schlichtkrull2018rgcn} behave similarly, and every backbone we tried falls within a narrow range. The encoder choice does add a smaller but consistent increment on top, with equivariance contributing $+$0.018 F1 (EGNN over GAT), multi-relational encoding $+$0.012, and learnable coordinate updates $+$0.017.

\begin{table}[h!]
\centering
\scriptsize
\caption{GNN backbone ablation on the epitope-ratio split. All variants share the same cross-attention decoder and losses; only the encoder differs. Eq.=equivariant message passing, MR=multi-relational edges, Coord=learned coordinate updates. }
\label{tab:gnn_backbone}
\begin{tabular}{lccccccc}
\toprule
\textbf{GNN} & \textbf{Eq.} & \textbf{MR} & \textbf{Coord} & \textbf{AUC} & \textbf{AUPRC} & \textbf{F1} & \textbf{MCC} \\
\midrule
GCN & No & No & No & .905 & .397 & .447 & .424 \\
GIN & No & No & No & .890 & .349 & .408 & .369 \\
GAT & No & No & No & .908 & .399 & .452 & .424 \\
RGCN & No & Yes & No & .910 & .405 & .454 & .422 \\
EGNN & Yes & No & Yes & .912 & .421 & .470 & .455 \\
EGNN-R (frozen) & Part. & Yes & No & .908 & .412 & .465 & .441 \\
\textbf{EGNN-R (ours)} & \textbf{Yes} & \textbf{Yes} & \textbf{Yes} & \textbf{.924} & \textbf{.493} & \textbf{.482} & \textbf{.464} \\
\bottomrule
\end{tabular}
\end{table}
\vspace{-5pt}

\textbf{Architecture-Loss Co-Design.}
We note that the sparsity-aware losses help only when the architecture already fuses the two chains early. We added \emph{EpiFormer}'s auxiliary losses (Dice, count regularizer, and edge prediction) to six baselines and tuned the loss weights for each by grid search (Table~\ref{tab:loss_fairness}). The losses raise F1 for three of the six, WALLE ($+$0.129), EquiformerV2 ($+$0.041), and DiscoTope3 ($+$0.031), and lower it for the remaining three, most sharply for CheapNet ($-$0.104). The same losses are more beneficial inside \emph{EpiFormer}, where they lift F1 from 0.336 with plain BCE to 0.482, which indicates that the architecture is the primary driver and the losses act as a complement to it.

\begin{table}[h!]
\centering
\scriptsize
\caption{Architecture and loss co-design: baselines augmented with \emph{EpiFormer}'s auxiliary losses. Tier~1 (antigen-only) receives Dice and count regularization; Tier~2 (antibody-aware) additionally receives edge prediction; Tier~3 is \emph{EpiFormer} itself, which uses the full joint objective. }
\label{tab:loss_fairness}
\begin{tabular}{llcccc}
\toprule
\textbf{Method} & \textbf{Tier} & \textbf{Orig.\ F1} & \textbf{F1} & \textbf{MCC} & \textbf{$\Delta$ F1} \\
\midrule
DiscoTope3 & 1 & 0.251 & 0.282 & 0.263 & $+0.031$ \\
EquiformerV2 & 1 & 0.260 & \textbf{0.301} & \textbf{0.293} & $+0.041$ \\
\midrule
WALLE & 2 & 0.203 & \textbf{0.332} & 0.280 & $\mathbf{+0.129}$ \\
PECAN & 2 & 0.229 & 0.218 & 0.157 & $-0.011$ \\
CheapNet & 2 & 0.184 & 0.080 & 0.154 & $-0.104$ \\
MIPE & 2 & 0.337 & 0.261 & 0.230 & $-0.076$ \\
\midrule
\textbf{EpiFormer} & \textbf{3} & 0.336 & \textbf{0.482} & \textbf{0.464} & $\mathbf{+0.146}$ \\
\bottomrule
\end{tabular}
\end{table}
\vspace{-5pt}

\textbf{Geometric Features Outperform PLMs.}
To probe our feature choice, we swapped the geometric features for protein language model (PLM) embeddings in \emph{EpiFormer} and, conversely, removed PLM embeddings from four baselines that rely on them. \emph{EpiFormer} is slightly worse with PLM features than with geometric ones (0.889 vs.\ 0.924 AUC), whereas the baselines drop sharply once their PLM features are removed, by 0.133 AUC for DiscoTope3, 0.122 for WALLE, and 0.088 for MIPE. This contrast suggests that geometric features are sufficient for \emph{EpiFormer}, while these baselines lean heavily on the PLM signal. A plausible reason is that epitope residues are not evolutionarily conserved~\citep{ponomarenko2007conformational}, so the conservation signal that PLMs encode is of limited use here. Epitope prediction is better described as a relational task defined by spatial proximity to a specific antibody~\citep{sela2015antibody}, which the geometric features capture directly. The full analysis is in Appendix~\ref{appendix:plm_analysis}.

\textbf{Cross-Dataset Generalization.}
To test out-of-distribution robustness, we evaluated \emph{EpiFormer} on three external benchmarks with no overlap with AsEP, SAbDab~\citep{dunbar2014sabdab}, CoV-AbDab~\citep{raybould2021cov}, and ANABAG~\citep{grandguillaume2025anabag} (Table~\ref{tab:crossdataset_main}). We compare against MIPE, which is the second-best method and the strongest antibody-aware baseline. Under zero-shot transfer, \emph{EpiFormer} achieves a mean AUC of 0.786, outperforming MIPE (0.688) despite substantial distribution shift (Figure~\ref{fig:distshift}). Leave-one-dataset-out (LODO) fine-tuning further improves the AUC to 0.890, with particularly strong gains on CoV-AbDab. Full results are in Appendix~\ref{appendix:crossdataset}.

\begin{table}[h!]
\centering
\small
\caption{Cross-dataset generalization (AUC) on three external benchmarks, SAbDab, CoV-AbDab, and ANABAG~\citep{grandguillaume2025anabag}. Zero-shot: AsEP-trained only. LODO (leave-one-dataset-out): fine-tuned on two external sets, tested on the third. MIPE is the second-best AsEP method and serves as the antibody-aware reference; per-dataset best in \textbf{bold}.}
\label{tab:crossdataset_main}
\begin{tabular}{lcccc}
\toprule
& \multicolumn{2}{c}{\textbf{Zero-shot AUC}} & \multicolumn{2}{c}{\textbf{LODO AUC}} \\
\cmidrule(lr){2-3} \cmidrule(lr){4-5}
\textbf{Dataset} & \textbf{EpiFormer} & \textbf{MIPE} & \textbf{EpiFormer} & \textbf{MIPE} \\
\midrule
SAbDab & \textbf{0.782} & 0.614 & \textbf{0.858} & 0.746 \\
CoV-AbDab & \textbf{0.819} & 0.730 & \textbf{0.890} & 0.813 \\
ANABAG & \textbf{0.758} & 0.719 & \textbf{0.837} & 0.772 \\
\bottomrule
\end{tabular}
\end{table}

\subsection{Ablations}

We isolate the contribution of each component; full details are provided in Appendix~\ref{appendix:ablations}.

\textbf{Joint Objective.} The joint loss is essential. With BCE alone, \emph{EpiFormer} achieves only 0.336 F1. Adding Dice and edge prediction yields marginal gains, whereas the sparsity regularizer is a more impactful term, improving F1 by $+$0.132 (Table~\ref{tab:loss_ablation}). This indicates that threshold degeneration, where the model predicts all residues as positive, is a primary failure mode that the count regularizer directly corrects. We also note that adding InfoNCE on top degrades performance ($-$0.029 F1) due to conflicting gradients on boundary residues~\citep{ji2024regcl}.

\textbf{Interleaved Cross-Attention.} Removing the encoder MHCA entirely ($\alpha{=}0.0$) drops F1 by 0.042 and MCC by 0.051 (Table~\ref{tab:mhca}). Even when $\alpha$ is initialized at zero, which forces the model to discover cross-attention utility from scratch, it learns non-zero gate values (0.014 to 0.033 for the antibody stream), confirming that cross-modal information provides a useful learning signal. See Appendix~\ref{appendix:mhca_ablation} for the full ablation.

\begin{figure}[h!]
    \centering
    \includegraphics[width=1\linewidth]{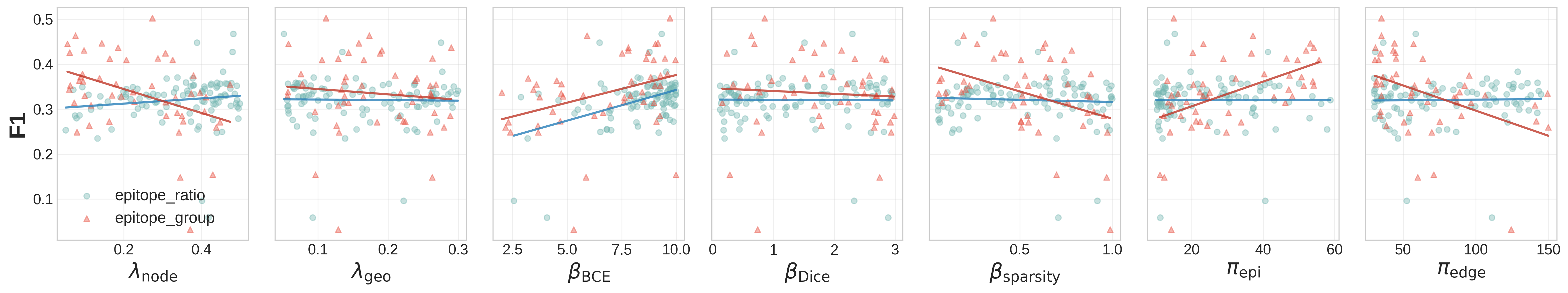}
    \caption{F1 sensitivity to seven loss hyperparameters across both splits. Most weights show flat trends, indicating robustness.}
    \label{fig:hparam-sensitivity}
\end{figure}

\textbf{Hyperparameter Sensitivity.}\label{appendix:hparam_sensitivity}
Figure~\ref{fig:hparam-sensitivity} confirms that F1 is largely insensitive to the exact loss weights, with most hyperparameters showing flat trends across both splits.

\subsection{Qualitative Analysis}

\textbf{Interaction Map Visualization.}
Figure~\ref{fig:attn+antigen-size-shift}(a) shows the interaction maps for a representative complex (\textsc{8df5\_3P}). The ground-truth contact map has a sparse binding interface concentrated in a narrow antigen region, and the distance map confirms that these contacts sit at close spatial proximity between epitope and paratope residues. The interaction matrix predicted by \emph{EpiFormer} closely follows this pattern, with high scores aligned to the true epitope band and non-binding regions suppressed. 

\begin{figure}[h!]
    \centering
    \includegraphics[width=0.65\linewidth]{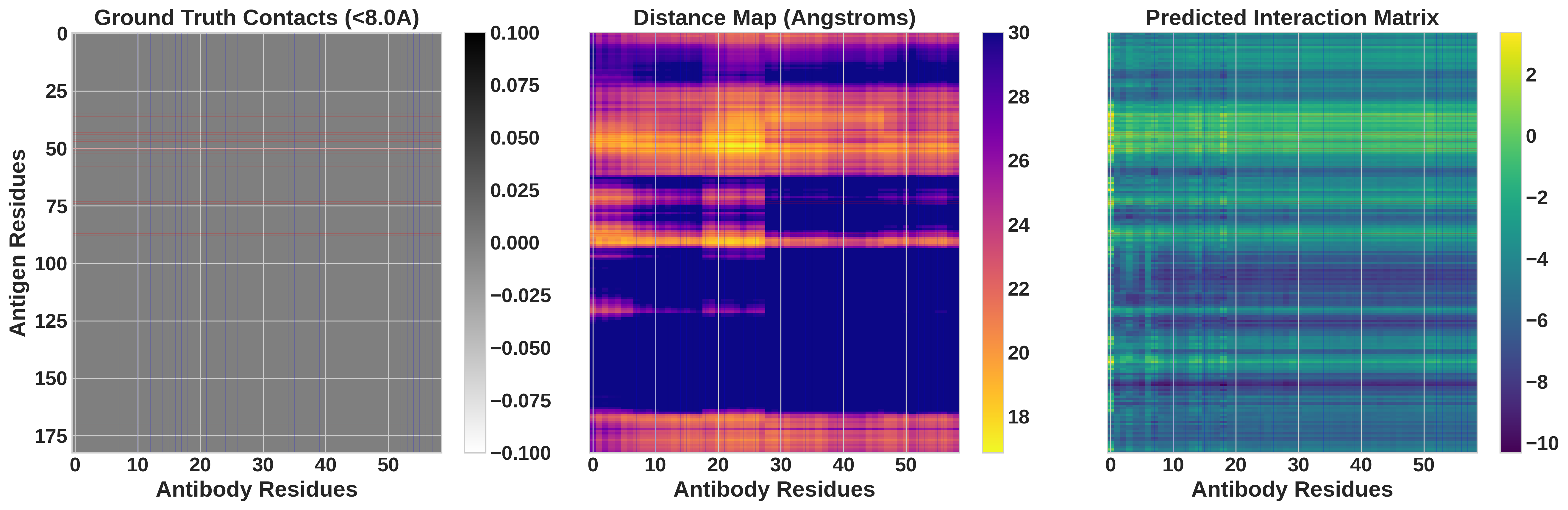} %
    \includegraphics[width=0.32\linewidth]{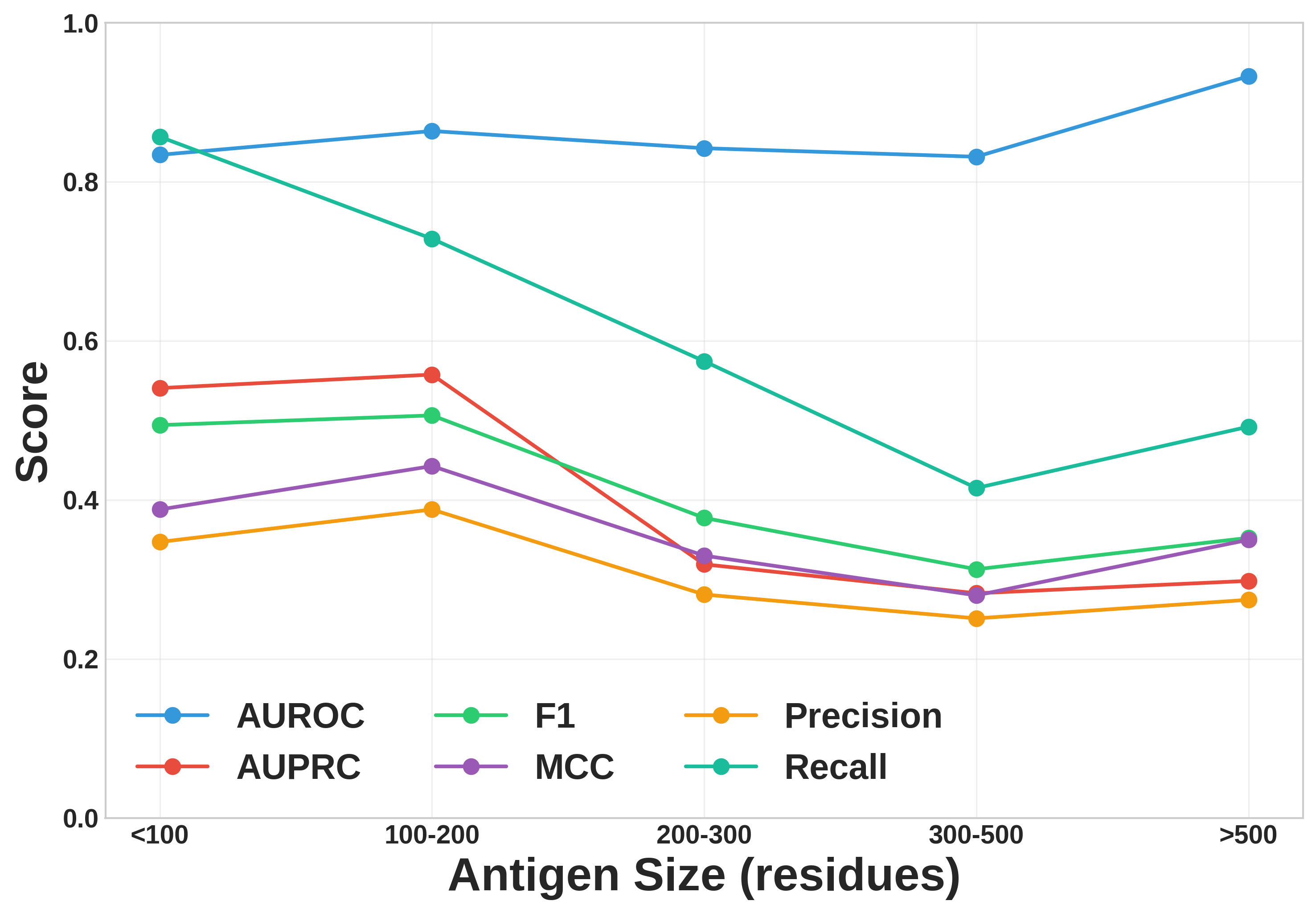}
    \caption{(a) Ground-truth contact map, pairwise distance map, and predicted interaction matrix for complex \textsc{8df5\_3P}. (b) Performance across antigen size bins.}
    \label{fig:attn+antigen-size-shift}
\end{figure}

% This indicates that the model captures the genuine structure of the antigen-antibody interaction rather than a diffuse surface prior.

\textbf{Performance by Antigen Size.}
Figure~\ref{fig:attn+antigen-size-shift}(b) stratifies performance by the antigen size. Ranking metrics (AUC, MCC, precision) remain stable across all size bins, but threshold-dependent metrics (AUPRC, F1, recall) decline for the largest antigens ($>$500 residues). This is expected, since as antigens grow the epitope makes up a smaller fraction of residues, which makes the classification task harder.

\textbf{Asymmetry in Learned Cross-Attention.}
Figure~\ref{fig:learnedalphas} shows a consistent asymmetry in the learned gates. Across every encoder block and every initialization value, the antibody stream draws more on the antigen than the antigen draws on the antibody. This direction of information flow is consistent with the asymmetric roles of the two chains, since the antibody CDR loops are the variable part that conforms to a given epitope~\citep{cagiada2025uncovering}. 

% The asymmetry emerges from the prediction objective alone, as the gate values are never directly supervised.

\begin{figure}[h!]
    \centering
    \includegraphics[width=0.95\linewidth]{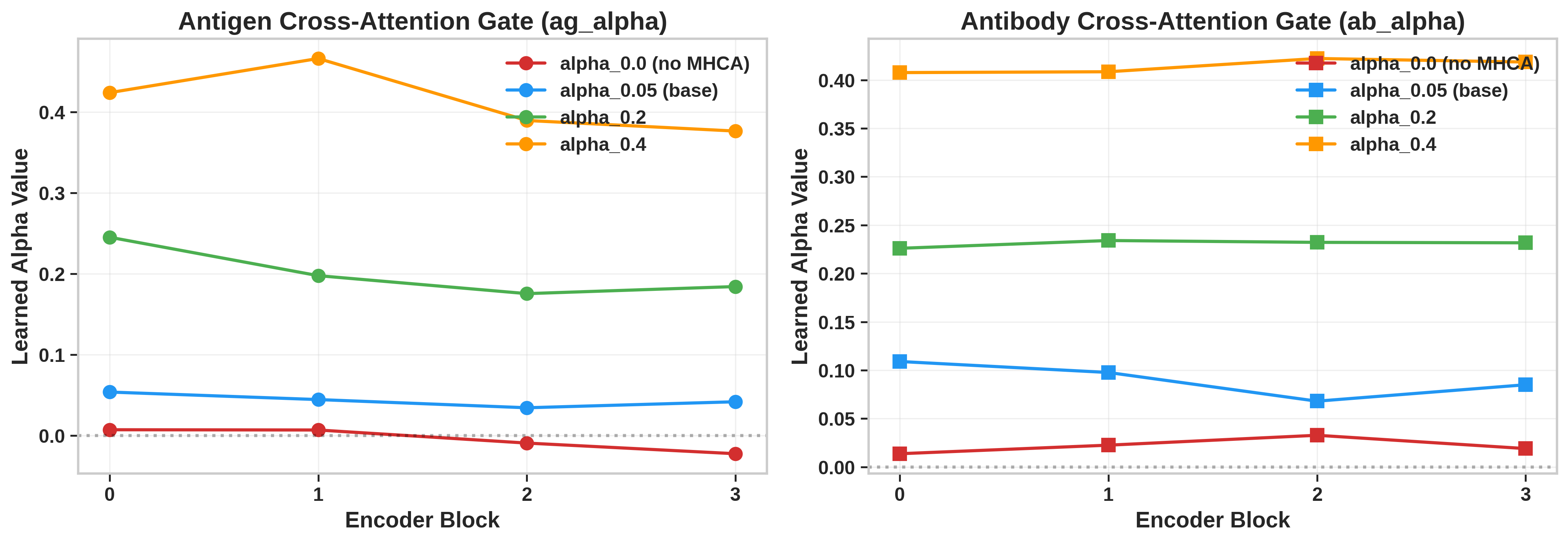}
    \caption{Learned cross-attention gate values ($\alpha$) across encoder blocks for four initialization values ($\alpha_{\text{init}} \in \{0.0, 0.05, 0.2, 0.4\}$). (a) Antigen gate. (b) Antibody gate. The antibody gate consistently exceeds the antigen gate, reflecting the asymmetric roles of the two chains.}
    \label{fig:learnedalphas}
\end{figure}

\section{Conclusion}\label{sc:conclusion}

We presented \emph{EpiFormer}, a general encoder-decoder framework for antibody-aware epitope prediction that addresses late fusion, class imbalance, and data scarcity jointly, and achieves state-of-the-art performance. Our key finding is that interleaved bidirectional cross-attention within GNN encoding layers provides consistent gains that are backbone-agnostic. We further showed that sparsity-aware training objectives must be co-designed with the fusion strategy, as the same losses degrade late-fusion models while improving early-fusion architectures. 

% \section*{Reproducibility Statement} 
% We will make the code publicly available on GitHub and provide installation scripts to address the complex dependency issue of the libraries/packages. We hope that this will support and accelerate future research and development.


\newpage

\bibliography{bibliography}
\bibliographystyle{plainnat}

%%%%%%%%%%%%%%%%%%%%%%%%%%%%%%%%%%%%%%%%%%%%%%%%%%%%%%%%%%%%%%%%%%%%%%%%%%%%%%%
%%%%%%%%%%%%%%%%%%%%%%%%%%%%%%%%%%%%%%%%%%%%%%%%%%%%%%%%%%%%%%%%%%%%%%%%%%%%%%%
% APPENDIX
%%%%%%%%%%%%%%%%%%%%%%%%%%%%%%%%%%%%%%%%%%%%%%%%%%%%%%%%%%%%%%%%%%%%%%%%%%%%%%%
%%%%%%%%%%%%%%%%%%%%%%%%%%%%%%%%%%%%%%%%%%%%%%%%%%%%%%%%%%%%%%%%%%%%%%%%%%%%%%%
\newpage
\appendix
\onecolumn

% You can have as much text here as you want. The main body must be at most $8$
% pages long. For the final version, one more page can be added. If you want, you
% can use an appendix like this one.

% The $\mathtt{\backslash onecolumn}$ command above can be kept in place if you
% prefer a one-column appendix, or can be removed if you prefer a two-column
% appendix.  Apart from this possible change, the style (font size, spacing,
% margins, page numbering, etc.) should be kept the same as the main body.
%%%%%%%%%%%%%%%%%%%%%%%%%%%%%%%%%%%%%%%%%%%%%%%%%%%%%%%%%%%%%%%%%%%%%%%%%%%%%%%
%%%%%%%%%%%%%%%%%%%%%%%%%%%%%%%%%%%%%%%%%%%%%%%%%%%%%%%%%%%%%%%%%%%%%%%%%%%%%%%

\section{Appendix}

\subsection{Dataset Statistics and Preprocessing}

Figure~\ref{fig:density-plots} summarizes the dataset. Antigen surface graphs average 288 residues (range 20--1500+), antibody CDR graphs average 61 residues, and the mean epitope size is 19 residues with a mean of 44 bipartite edges per complex.

\begin{figure}[h!]
    \centering
    \includegraphics[width=0.8\textwidth]{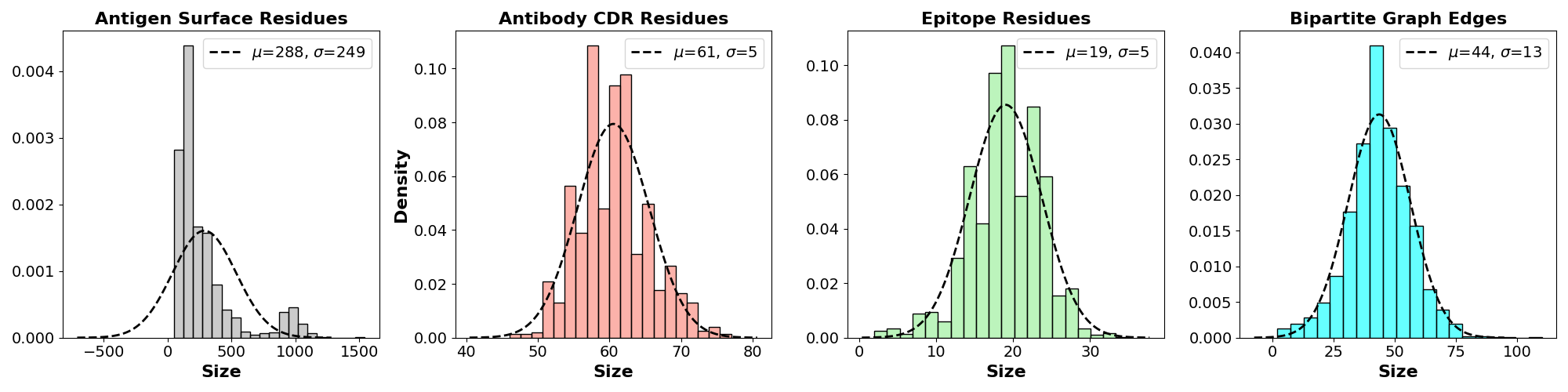}
    \caption{Size distributions of antigen surface residues, antibody CDR residues, epitope residues, and bipartite edges in the AsEP dataset.}
    \label{fig:density-plots}
\end{figure}

\subsubsection{Preprocessing}\label{appendix:preprocessing}

For each complex, we first separated the paired antigen and antibody chains into individual structure files. We then performed sequence-structure alignment using Clustal Omega~\citep{sievers2011clustal} to establish correspondence between SEQRES (complete sequence) and ATOMSEQ (resolved atoms) records. This alignment generated binary masks that enable reliable mapping of sequences to structural residues while preserving the native crystallographic ordering.

For antibody chains, we applied the alignment masks to reindex heavy (H) and light (L) chains by removing insertion codes to enforce consecutive 1-based residue numbering required for graph construction. Antigen chains underwent similar processing to maintain parity between sequences and structures. This step ensures that each residue in the protein sequence corresponds exactly to its structural counterpart during the graph representation.
Then, we applied solvent-accessibility filters to retain only antigen surface residues, using the original AsEP seqres2surf masks to define the node set for antigen residue graphs. The binary epitope labels were projected onto the surface ATOMSEQ via alignment masks, while paratope labels were preserved for antibody residue nodes. This surface filtering step prevents non-surface residues from confounding epitope supervision while maintaining all necessary information for cross-chain interaction modeling.

We excluded two complexes (5nj6\_0P and 5ies\_0P) from the AsEP dataset due to sequence alignment inconsistencies and unresolved residues, with the final dataset containing 1,721 complexes. 
The contact distribution between residues in the bipartite graph had a mean of 43.7 contacts with a standard deviation of 12.8. Additionally, the dataset includes 641 unique antigens and 973 epitope groups, highlighting the diversity and complexity of the antibody-antigen interactions captured in the AsEP dataset.

Finally, we used these preprocessed structures to generate HeteroData objects for the multi-relational graphs using PyTorch Geometric~\citep{fey2019pyg}. Each residue node is attributed a geometric feature vector encoding amino acid type, backbone geometry, distances, and local coordinate frames (detailed in Section~\ref{appendix:graph-construction}).

\subsection{Baseline Comparison }~\label{appendix:baselines}

Table~\ref{table:benchmark_features_comparison} compares the architectural components and loss functions used by each baseline against \emph{EpiFormer}.

\begin{table}[h!]
\centering
\scriptsize
\caption{Architectural and loss component comparison across baselines and EpiFormer. Column abbreviations are defined in the header rows.}
\label{table:benchmark_features_comparison}
\resizebox{\textwidth}{!}{
\begin{tabular}{l|cccccccccc|ccccc}
\toprule
& \multicolumn{10}{c|}{\textbf{Architecture}} & \multicolumn{5}{c}{\textbf{Loss Components}} \\
\textbf{Method} & \textbf{Ab} & \textbf{St} & \textbf{PLM} & \textbf{Graph} & \textbf{Geom} & \textbf{M-rel} & \textbf{E(3)} & \textbf{X-Att} & \textbf{Inter} & \textbf{Gate} & \textbf{BCE} & \textbf{Dice} & \textbf{Contr} & \textbf{Count} & \textbf{Edge} \\
\midrule
\multicolumn{16}{l}{\textit{\textbf{Epitope and paratope prediction}}} \\
EpiGraph & $\times$ & \checkmark & \checkmark & \checkmark & $\times$ & $\times$ & $\times$ & $\times$ & $\times$ & $\times$ & \checkmark & $\times$ & $\times$ & $\times$ & $\times$ \\
EpiScan & \checkmark & $\times$ & $\times$ & $\times$ & $\times$ & $\times$ & $\times$ & $\times$ & $\times$ & $\times$ & \checkmark & \checkmark & $\times$ & $\times$ & $\times$ \\
MIPE & \checkmark & \checkmark & \checkmark & \checkmark & \checkmark & \checkmark & $\times$ & \checkmark & $\times$ & $\times$ & \checkmark & $\times$ & \checkmark & $\times$ & $\times$ \\
WALLE & \checkmark & \checkmark & \checkmark & \checkmark & $\times$ & $\times$ & $\times$ & $\times$ & $\times$ & $\times$ & \checkmark & $\times$ & $\times$ & \checkmark & \checkmark \\
DiscoTope3 & $\times$ & \checkmark & \checkmark & $\times$ & \checkmark & $\times$ & $\times$ & $\times$ & $\times$ & $\times$ & \checkmark & $\times$ & $\times$ & $\times$ & $\times$ \\
GraphBepi & $\times$ & \checkmark & \checkmark & \checkmark & \checkmark & $\times$ & $\times$ & $\times$ & $\times$ & $\times$ & \checkmark & $\times$ & $\times$ & $\times$ & $\times$ \\
PECAN & \checkmark & \checkmark & $\times$ & \checkmark & \checkmark & $\times$ & $\times$ & \checkmark & $\times$ & $\times$ & \checkmark & $\times$ & $\times$ & $\times$ & $\times$ \\
\midrule
\multicolumn{16}{l}{\textit{\textbf{Protein binding-site prediction}}} \\
EquiPocket & $\times$ & \checkmark & $\times$ & \checkmark & \checkmark & $\times$ & \checkmark & $\times$ & $\times$ & $\times$ & \checkmark & \checkmark & $\times$ & $\times$ & $\times$ \\
AtomSurf & $\times$ & \checkmark & \checkmark & \checkmark & \checkmark & $\times$ & $\times$ & $\times$ & $\times$ & $\times$ & \checkmark & $\times$ & $\times$ & $\times$ & $\times$ \\
ESMBind & $\times$ & $\times$ & \checkmark & $\times$ & $\times$ & $\times$ & $\times$ & $\times$ & $\times$ & $\times$ & \checkmark & $\times$ & $\times$ & $\times$ & $\times$ \\
\midrule
\multicolumn{16}{l}{\textit{\textbf{Protein interaction and docking}}} \\
ATProt & \checkmark & \checkmark & $\times$ & \checkmark & \checkmark & $\times$ & $\times$ & \checkmark & $\times$ & $\times$ & \checkmark & $\times$ & $\times$ & $\times$ & $\times$ \\
DiffDock-PP & \checkmark & \checkmark & \checkmark & \checkmark & \checkmark & $\times$ & \checkmark & \checkmark & $\times$ & $\times$ & \checkmark & $\times$ & $\times$ & $\times$ & $\times$ \\
DiffDock & $\times$ & \checkmark & $\times$ & \checkmark & \checkmark & $\times$ & \checkmark & $\times$ & $\times$ & $\times$ & $\times$ & $\times$ & $\times$ & $\times$ & $\times$ \\
\midrule
\multicolumn{16}{l}{\textit{\textbf{Structure and affinity prediction}}} \\
Boltz-1 & $\times$ & \checkmark & \checkmark & $\times$ & \checkmark & $\times$ & $\times$ & $\times$ & $\times$ & $\times$ & \checkmark & $\times$ & $\times$ & $\times$ & $\times$ \\
Boltz-2 & \checkmark & \checkmark & \checkmark & $\times$ & \checkmark & $\times$ & $\times$ & \checkmark & $\times$ & $\times$ & \checkmark & $\times$ & $\times$ & $\times$ & $\times$ \\
AlphaFold3 & \checkmark & \checkmark & $\times$ & $\times$ & \checkmark & $\times$ & \checkmark & $\times$ & $\times$ & $\times$ & $\times$ & $\times$ & $\times$ & $\times$ & $\times$ \\
CheapNet & \checkmark & \checkmark & $\times$ & \checkmark & \checkmark & $\times$ & $\times$ & \checkmark & $\times$ & $\times$ & \checkmark & $\times$ & $\times$ & $\times$ & $\times$ \\
GearBind & $\times$ & \checkmark & $\times$ & \checkmark & \checkmark & \checkmark & $\times$ & $\times$ & $\times$ & $\times$ & \checkmark & $\times$ & $\times$ & $\times$ & $\times$ \\
\midrule
\multicolumn{16}{l}{\textit{\textbf{Molecular property prediction}}} \\
EquiformerV2 & $\times$ & \checkmark & \checkmark & \checkmark & \checkmark & $\times$ & \checkmark & $\times$ & $\times$ & $\times$ & \checkmark & $\times$ & $\times$ & $\times$ & $\times$ \\
\midrule
\textbf{\emph{EpiFormer}} & \textbf{\checkmark} & \textbf{\checkmark} & $\times$ & \textbf{\checkmark} & \textbf{\checkmark} & \textbf{\checkmark} & \textbf{\checkmark} & \textbf{\checkmark} & \textbf{\checkmark} & \textbf{\checkmark} & \textbf{\checkmark} & \textbf{\checkmark} & $\times$ & \textbf{\checkmark} & \textbf{\checkmark} \\
\bottomrule
\end{tabular}
}
\vspace{0.5em}

\raggedright
\footnotesize
Architecture columns: Ab = antibody-aware; St = uses 3D structure; PLM = protein language model features; Graph = graph neural network; Geom = explicit geometric or surface features; M-rel = multi-relational edges; E(3) = E(3)/SE(3)-equivariant; X-Att = cross-attention; Inter = interleaved (per-layer) cross-attention; Gate = gated cross-attention. \\

\end{table}

All baselines were evaluated on the 1,721-complex AsEP benchmark under identical settings: official splits (epitope-ratio and epitope-group), threshold 0.5, and the same six metrics (AUROC, AUPRC, F1, MCC, Precision, Recall). Below we describe each baseline's architecture and its adaptation to epitope prediction.

\subsubsection{Protein Binding-Site Prediction Methods}

\paragraph{EquiPocket~\citep{zhang2023equipocket}} operates at atom-level with 6D features, computing molecular surfaces via MSMS with 7D local geometric descriptors. The Surface-EGNN module performs E(3)-equivariant multi-channel convolution with DenseNet-style connections. We adapted to epitope prediction by adding residue aggregation (CA atom selection or mean pooling) to convert per-atom outputs to per-residue predictions.

\paragraph{AtomSurf~\citep{mallet2023atomsurf}} jointly encodes molecular surfaces using DiffusionNet (spectral convolution with 22D geometric features including HKS) and residue graphs using GCN with ProNet features, connected through bidirectional KNN-based message passing. We preserve the joint surface-graph encoding for antigen-only processing with mean-pooled surface features aggregated to the residue level.

\paragraph{ESMBind~\citep{schreiber2023esmbind}} adapts the ESM-2 protein language model (35M parameters by default) with LoRA (Low-Rank Adaptation) for parameter-efficient fine-tuning, training only \~0.5\% of parameters. The model processes antigen amino acid sequences through ESM-2 with LoRA adapters (rank $r=4$, scaling $\alpha=8$, dropout 0.1) applied to attention layers, followed by a per-residue classification head for epitope prediction. We fine-tune ESMBind on the epitope sequences using class-weighted BCE loss to handle label imbalance, and support multiple ESM-2 variants (8M to 650M parameters).

% \paragraph{MaSIF-site~\citep{gainza2020deciphering}} applies geodesic convolution on molecular surfaces with polar coordinate patches (nthetas=4, nrhos=3, maxdistance=9.0\textup{\AA}). We substitute original chemical features with 22 geometric features (curvatures, HKS, normals) and ESM-2 embeddings (1280D), using Euclidean KNN approximation for geodesic patches and vertex-to-residue aggregation for per-residue output.

\subsubsection{Protein Interaction and Docking Methods}

\paragraph{ATProt~\citep{gao2024towards}} uses SEGCN (Structure-Enhanced Graph Convolutional Network) for general protein-protein interaction site prediction. We ported the architecture from DGL to PyG, added cross-attention between antibody and antigen chains for antibody awareness, and trained from scratch on AsEP with the same evaluation protocol. ATProt applies cross-attention only once after all encoding layers (late fusion), in contrast to EpiFormer's interleaved design.

\paragraph{DiffDock-PP~\citep{ketata2023diffdockpp}} uses SE(3)-equivariant tensor product convolutions for rigid-body protein-protein docking. We adapted the architecture for residue-level epitope classification by replacing the docking score head with per-residue binary classification, adding cross-attention between antibody and antigen representations. Both antibody and antigen are encoded using the same TensorProductConvLayer backbone with ESM-2 embeddings.

\paragraph{DiffDock~\citep{corso2022diffdock}} generates ligand poses through diffusion over SE(3) using E(3)-equivariant score networks. We use inference-only mode with DiffDock-L to dock antibody CDR regions onto antigen structures, labeling antigen residues within contact distance of predicted CDR poses as epitopes, weighted by confidence scores.

\subsubsection{Structure and Affinity Prediction Methods}

\paragraph{Boltz-1~\citep{wohlwend2025boltz1}} uses a Pairformer trunk with triangle multiplication (\( z_{ij} += \sum_k a_{ik} \cdot b_{jk} \)), triangle attention, and attention with pair bias for structure-aware representations. We adapt by projecting ESM-2 embeddings to single (384D) and pairwise (128D) representations, reducing to 8 layers, and adding an epitope prediction head on concatenated single and mean-pooled pairwise features.

\paragraph{AlphaFold3~\citep{abramson2024alphafold3}} uses a diffusion-based architecture with Pairformer trunk processing single and pairwise representations through adaptive LayerNorm, gated linear units, and triangle operations. We use inference-only mode with pretrained weights~\footnote{The AF3 model weights were obtained from DeepMind. We used AF3 without MSA due to computational resource constraints.}, converting inter-chain contact predictions from predicted complexes to epitope labels using distance thresholds.
Due to imprecise docking without evolutionary information, we used a relaxed 15Å contact threshold.
AlphaFold3 without MSA achieved AUC=0.56 with the relaxed 15Å threshold, only marginally better than random, confirming that structure prediction alone without co-evolutionary signals cannot reliably identify epitopes.

% \paragraph{ESMFold~\citep{lin2023esm2}} performs single-sequence structure prediction using ESM-2 embeddings without MSA. We use inference-only mode to predict antibody-antigen complex structures, deriving epitope labels from inter-chain contacts in predicted structures with Clustal Omega for sequence alignment between predicted and reference structures.

\paragraph{CheapNet~\citep{limcheapnet}} uses Geometry-Informed Graph Neural Network blocks for local structure encoding with intra/inter-molecular edges and cross-attention between entity representations. We adapted from graph-level affinity regression to node-level classification by replacing ligand-protein atom pairs with antibody-antigen residue pairs, projecting 35D residue features from RAAD~\citep{wu2025raad}/ESM-2, and adding a per-residue MLP classifier with BCE loss.

\paragraph{GearBind~\citep{cai2024gearbind}} is a pretrainable geometric GNN based on GearNet, pretrained on CATH using contrastive learning and fine-tuned on SKEMPI for \(\Delta\Delta G_{\text{bind}}\) prediction. We extract encoder representations for antigen residues, replace the affinity regression head with a binary classifier, and initialize from CATH-pretrained weights.

\paragraph{Boltz-2~\citep{passaro2025boltz2}} extends Boltz-1 with cross-attention between binding partners. We additionally process antibody embeddings (AntiBERTy, 512D) alongside antigen (ESM-2, 1280D), using 2-layer CrossPairAttention for antibody-aware epitope prediction with the same triangle-based Pairformer backbone.

\subsubsection{Molecular Property Prediction Methods}

\paragraph{EquiformerV2~\citep{liao2023equiformerv2}} is an improved equivariant transformer using SO(2) convolutions for efficiency, processing 3D graphs with type-\(L\) features and attention in irreducible representation space. We adapt from energy/force regression to residue-level classification by constructing antigen graphs with CA coordinates, combining RAAD features with ESM-2 embeddings, and replacing the output head with per-residue binary classification.

\subsection{Epitope-Group Split Results}\label{appendix:epigroup}

Table~\ref{tab:baseline_epigroup} reports results on the more challenging epitope-group split, where test epitopes are completely unseen during training. EpiFormer achieves 0.826 AUC and 0.305 F1, outperforming all baselines.

\begin{table}[h!]
\centering
\scriptsize
\caption{Performance comparison on the AsEP epitope-group split. Best in \textbf{bold}, second-best \underline{underlined}.}
\label{tab:baseline_epigroup}
\begin{tabular}{lcccccc}
\toprule
\footnotesize
\textbf{Method} & \textbf{AUC}$\uparrow$ & \textbf{AUPRC}$\uparrow$ & \textbf{F1}$\uparrow$ & \textbf{MCC}$\uparrow$ & \textbf{Prec.}$\uparrow$ & \textbf{Rec.}$\uparrow$ \\
\midrule
\multicolumn{7}{l}{\textit{\textbf{Epitope and paratope prediction}}} \\
EpiGraph~\citep{choi2024b} & 0.779 & 0.194 & 0.056 & 0.096 & \underline{0.401} & 0.030 \\
EpiScan~\citep{wang2024episcan} & 0.443 & 0.080 & 0.127 & 0.048 & 0.089 & 0.221 \\
MIPE~\citep{wang2024improving} & 0.740 & \underline{0.228} & 0.172 & \underline{0.206} & \textbf{0.495} & 0.104 \\
WALLE~\citep{liu2024asep} & 0.713 & 0.137 & 0.170 & 0.143 & 0.095 & \textbf{0.830} \\
DiscoTope3~\citep{hoie2024discotope} & 0.763 & 0.208 & 0.210 & 0.189 & 0.123 & 0.729 \\
GraphBepi~\citep{zeng2023graphbepi} & \underline{0.781} & \underline{0.220} & 0.144 & 0.161 & 0.386 & 0.088 \\
PECAN~\citep{pittala2020learning} & 0.729 & 0.166 & \underline{0.254} & 0.195 & 0.169 & 0.513 \\
\midrule
\multicolumn{7}{l}{\textit{\textbf{Protein binding-site prediction}}} \\
EquiPocket~\citep{zhang2023equipocket} & 0.711 & 0.132 & 0.194 & 0.159 & 0.113 & 0.676 \\
AtomSurf~\citep{mallet2023atomsurf} & 0.699 & 0.144 & 0.202 & 0.151 & 0.127 & 0.500 \\
ESMBind~\citep{schreiber2023esmbind} & 0.685 & 0.130 & 0.174 & 0.126 & 0.100 & 0.663 \\
\midrule
\multicolumn{7}{l}{\textit{\textbf{Protein interaction and docking}}} \\
ATProt~\citep{gao2024towards} & 0.710 & 0.155 & 0.200 & 0.160 & 0.120 & 0.600 \\
DiffDock-PP~\citep{ketata2023diffdockpp} & 0.650 & 0.110 & 0.165 & 0.120 & 0.095 & 0.680 \\
DiffDock~\citep{corso2022diffdock} & 0.443 & 0.095 & 0.146 & 0.081 & 0.131 & 0.164 \\
\midrule
\multicolumn{7}{l}{\textit{\textbf{Structure and affinity prediction}}} \\
Boltz-1~\citep{wohlwend2025boltz1} & 0.702 & 0.174 & 0.229 & 0.169 & 0.146 & 0.536 \\
Boltz-2~\citep{passaro2025boltz2} & 0.702 & 0.168 & 0.220 & 0.156 & 0.140 & 0.518 \\
AlphaFold3~\citep{abramson2024alphafold3} & 0.570 & 0.076 & 0.155 & 0.079 & 0.098 & 0.375 \\
CheapNet~\citep{limcheapnet} & 0.641 & 0.095 & 0.152 & 0.100 & 0.084 & \underline{0.776} \\
GearBind~\citep{cai2024gearbind} & 0.734 & 0.154 & 0.206 & 0.169 & 0.123 & 0.627 \\
\midrule
\multicolumn{7}{l}{\textit{\textbf{Molecular property prediction}}} \\
EquiformerV2~\citep{liao2023equiformerv2} & 0.729 & 0.160 & 0.212 & 0.169 & 0.131 & 0.561 \\
\midrule
% \textbf{\emph{EpiFormer} (ours)} & \textbf{0.826\scriptsize{$\pm$.070}} & 0.290\scriptsize{$\pm$.100} & \textbf{0.305\scriptsize{$\pm$.066}} & \textbf{0.290\scriptsize{$\pm$.085}} & 0.196\scriptsize{$\pm$.047} & 0.698\scriptsize{$\pm$.089} \\
\textbf{\emph{EpiFormer} (ours)} & \textbf{0.826} & \textbf{0.290} & \textbf{0.305} & \textbf{0.290} & 0.196 & 0.698 \\
\bottomrule
\end{tabular}
\end{table}

\subsection{E(3)-equivariance of the EGNN-R layer}\label{thm:egnn_equivariance}

\begin{theorem}
Consider the EGNN-R layer in \S\ref{sec:egnnr} with updates

\begin{align}
m_{ij}^{\rho}      &= \phi_m^{\rho} \bigl(\mathbf{h}_i^{\ell},  
\label{equation19}
\mathbf{h}_j^{\ell}, \gamma(d_{ij}), \mathbf{f}_{ij}\bigr), \\
s_{ij}^{\rho}      &= \phi_x^{\rho} \bigl(m_{ij}^{\rho}\bigr), \\[2pt]
\mathbf{h}_i^{(\ell+1)} &= \mathbf{h}_i^{\ell} + \phi_h \Bigl(\mathbf{h}_i^{\ell}, \sum_{j\in\mathcal{N}(i)}\;\sum_{\rho\in\mathbf{r}_{ij}} m_{ij}^{\rho}\Bigr), \\
\mathbf{x}_i^{(\ell+1)} &= \mathbf{x}_i^{\ell} + \sum_{j\in\mathcal{N}(i)}\;\sum_{\rho\in\mathbf{r}_{ij}}
            \frac{\boldsymbol{\delta}_{ij}}{\sqrt{d_{ij}+\varepsilon}}\;s_{ij}^{\rho},
\end{align}

where $\boldsymbol{\delta}_{ij}=\mathbf{x}_i^{\ell}-\mathbf{x}_j^{\ell}$, $d_{ij}=\|\boldsymbol{\delta}_{ij}\|_2^2$, and $\varepsilon>0$. Assume:
(i) node features $\mathbf{h}_i^{\ell}\in\mathbb{R}^{d_h}$ are scalar channels,
(ii) $\mathbf{h}_{ij}$ and $\mathbf{r}_{ij}$ are categorical and independent of coordinates,
(iii) $\gamma$ is any scalar function of $d_{ij}$,
(iv) each $\phi_{\{m,x\}}^{\rho}$ is an MLP from scalars to scalars.
Let the $E(3)$ action be $g=(R,t)$ with $R\in O(3)$ and $t\in\mathbb{R}^3$, acting as
$\mathbf{x}_i^{\ell}\mapsto R\mathbf{x}_i^{\ell}+t$ and $\mathbf{h}_i^{\ell}\mapsto \mathbf{h}_i^{\ell}$.
Then the layer is $E(3)$-equivariant:
\[
\bigl\{\,\mathbf{x}_i^{\ell},\mathbf{h}_i^{\ell}\,\bigr\}_{i=1}^n \mapsto
\bigl\{\,R\mathbf{x}_i^{\ell}+t,\ \mathbf{h}_i^{\ell}\,\bigr\}_{i=1}^n
\ \Longrightarrow\
\bigl\{\,\mathbf{x}_i^{(\ell+1)},\mathbf{h}_i^{(\ell+1)}\,\bigr\}_{i=1}^n \mapsto
\bigl\{\,R\mathbf{x}_i^{(\ell+1)}+t,\ \mathbf{h}_i^{(\ell+1)}\,\bigr\}_{i=1}^n.
\]
Consequently, any stack of such layers is $E(3)$-equivariant by composition.
\end{theorem}

\begin{proof}
Let $g=(R,t)\in E(3)$ act as stated. Edge data $\mathbf{f}_{ij}$ and $\mathbf{r}_{ij}$ are unchanged.

\emph{Invariants.}
Relative displacement and distance transform as
\begin{align}
\boldsymbol{\delta}_{ij}\mapsto R\boldsymbol{\delta}_{ij},
\qquad
d_{ij}=\|\boldsymbol{\delta}_{ij}\|^2 \mapsto \|R\boldsymbol{\delta}_{ij}\|^2=d_{ij}.
\end{align}

Hence $d_{ij}$, $\gamma(d_{ij})$, and $(d_{ij}+\varepsilon)^{-1/2}$ are invariant scalars.

\emph{Scalar messages and coefficients.}
Each message $m_{ij}^{\rho}=\phi_m^{\rho}(\mathbf{h}_i^{\ell},\mathbf{h}_j^{\ell},\gamma(d_{ij}),\mathbf{f}_{ij})$ depends only on scalars that are invariant under $g$, so $m_{ij}^{\rho}$ is invariant. Then $s_{ij}^{\rho}=\phi_x^{\rho}(m_{ij}^{\rho})$ is also invariant.

\emph{Feature update.}
The update
\begin{align}
\mathbf{h}_i^{(\ell+1)} &= \mathbf{h}_i^{\ell} + \phi_h \Bigl(\mathbf{h}_i^{\ell}, \sum_{j\in\mathcal{N}(i)}\;\sum_{\rho\in\mathbf{r}_{ij}} m_{ij}^{\rho}\Bigr)
\end{align}
uses only invariant scalars, so $\mathbf{h}_i^{(\ell+1)}$ is invariant. This matches the scalar action on features.

\emph{Coordinate update.}
The increment
\begin{align}
\Delta\mathbf{x}_i=\sum_{j\in\mathcal{N}(i)}\;\sum_{\rho\in\mathbf{r}_{ij}}\frac{\boldsymbol{\delta}_{ij}}{\sqrt{d_{ij}+\varepsilon}}\,s_{ij}^{\rho}
\end{align}
is a sum of relative vectors scaled by invariant scalars. Under $g$ each term becomes
\begin{align}
\frac{\boldsymbol{\delta}_{ij}}{\sqrt{d_{ij}+\varepsilon}}\,s_{ij}^{\rho}
\ \mapsto\
\frac{R\boldsymbol{\delta}_{ij}}{\sqrt{d_{ij}+\varepsilon}}\,s_{ij}^{\rho}
= R\Bigl(\frac{\boldsymbol{\delta}_{ij}}{\sqrt{d_{ij}+\varepsilon}}\,s_{ij}^{\rho}\Bigr),
\end{align}
so $\Delta\mathbf{x}_i\mapsto R\,\Delta\mathbf{x}_i$. Therefore
\begin{align}
\mathbf{x}_i^{(\ell+1)}=\mathbf{x}_i^{\ell}+\Delta\mathbf{x}_i
\ \mapsto\
R\mathbf{x}_i^{\ell}+t+R\Delta\mathbf{x}_i
=R\bigl(\mathbf{x}_i^{\ell}+\Delta\mathbf{x}_i\bigr)+t
=R\mathbf{x}_i^{(\ell+1)}+t.
\end{align}

\emph{Composition.}
The composition of equivariant maps is equivariant. Hence, any stack of EGNN-R layers is $E(3)$-equivariant.
\end{proof}

\subsection{Multi-head cross-attention with feed-forward network (MHCA)}\label{sec:mhca}
The bidirectional multi-head cross-attention mechanism enables information exchange between antigen and antibody chains. Let $n_{\mathrm{head}}$ be the number of heads with per-head width $d_a=d_h/n_{\mathrm{head}}$. For layer $\ell$, independent linear projections produce queries, keys, and values:
\begin{align}
    \mathbf{Q}^{\ell}_{\mathrm{ag}} &= \mathbf{H}^{(\ell-1)}_{\mathrm{ag}}\mathbf{W}^{Q(\ell)}_{\mathrm{ag}},\\
    \mathbf{K}^{\ell}_{\mathrm{ab}} &= \mathbf{H}^{(\ell-1)}_{\mathrm{ab}}\mathbf{W}^{K(\ell)}_{\mathrm{ab}},\\
    \mathbf{V}^{\ell}_{\mathrm{ab}} &= \mathbf{H}^{(\ell-1)}_{\mathrm{ab}}\mathbf{W}^{V(\ell)}_{\mathrm{ab}},
\end{align}
with analogous expressions for the reverse direction. After reshaping to $n_{\mathrm{head}}$ heads of width $d_a$, scaled dot-product attention computes the affinity matrices:
\begin{align}
    \mathbf{A}^{\ell}_{\mathrm{ag}\leftarrow\mathrm{ab}} &= \operatorname{softmax} \Bigl(\tfrac{1}{\sqrt{d_h}}\,\mathbf{Q}^{\ell}_{\mathrm{ag}}{\mathbf{K}^{\ell}_{\mathrm{ab}}}^{ \top} + \mathbf{M}\Bigr),
\end{align}
where $\mathbf{M}$ is a batch mask (applied only in decoder) that assigns $-\infty$ to residue pairs from different complexes. The resulting context vectors are:
\begin{align}
    \widetilde{\mathbf{H}}^{\ell}_{\mathrm{ag}} &= \bigl[\mathbf{A}^{\ell}_{\mathrm{ag}\leftarrow\mathrm{ab}}\mathbf{V}^{\ell}_{\mathrm{ab}}\bigr]\mathbf{W}^{\ell}_{O,\mathrm{ag}},\\
    \widetilde{\mathbf{H}}^{\ell}_{\mathrm{ab}} &= \bigl[\mathbf{A}^{\ell}_{\mathrm{ab}\leftarrow\mathrm{ag}}\mathbf{V}^{\ell}_{\mathrm{ag}}\bigr]\mathbf{W}^{\ell}_{O,\mathrm{ab}}.
\end{align}
Each direction then applies a feed-forward network $\operatorname{FFN}(\mathbf{x})=\mathbf{W}_2\,\sigma(\mathbf{W}_1\mathbf{x}+\mathbf{b}_1)+\mathbf{b}_2$ with dropout, residual connections, and layer normalization.

\begin{algorithm}[h!]
\scriptsize
\SetAlgoLined
\KwIn{%
  Antigen graph $\mathcal{G}_{\mathrm{ag}}$ and antibody graph $\mathcal{G}_{\mathrm{ab}}$ with coordinates $\mathbf{X}$, geometric features $\mathbf{h}^{\text{geo}}$}
\KwOut{Bipartite interaction matrix $\hat{\mathcal{E}}_{\text{bg}} \in [0,1]^{n \times m}$}

\BlankLine
\tcp{\textbf{Feature Initialization}}
\ForEach{chain $\in \{\text{ag}, \text{ab}\}$}{
  Project geometric features to hidden dimension\;
  $\mathbf{h}_i^{0} \leftarrow \mathbf{W}_{\text{geo}}\,\mathbf{h}^{\text{geo}}_i$ for each residue $i$\;
}

\BlankLine
\tcp{\textbf{Encoder: Parallel Processing}}
\For{layer $\ell = 1$ \KwTo $L$}{
  \tcp{Intra-chain geometric message passing}
  $(\mathbf{H}_{\mathrm{ag}}^{\mathrm{intra}}, \mathbf{X}^{\mathrm{ag}}) \leftarrow \mathrm{EGNN\text{-}R}(\mathcal{G}_{\mathrm{ag}}, \mathbf{H}_{\mathrm{ag}}^{(\ell-1)}, \mathbf{X}^{\mathrm{ag}})$\;
  $(\mathbf{H}_{\mathrm{ab}}^{\mathrm{intra}}, \mathbf{X}^{\mathrm{ab}}) \leftarrow \mathrm{EGNN\text{-}R}(\mathcal{G}_{\mathrm{ab}}, \mathbf{H}_{\mathrm{ab}}^{(\ell-1)}, \mathbf{X}^{\mathrm{ab}})$\;
  
  \tcp{Inter-chain cross-attention}
  $\widetilde{\mathbf{H}}_{\mathrm{ag}} \leftarrow \mathrm{MHCA}(\mathbf{H}_{\mathrm{ag}}^{\mathrm{intra}}, \mathbf{H}_{\mathrm{ab}}^{\mathrm{intra}}, \mathbf{H}_{\mathrm{ab}}^{\mathrm{intra}})$\;
  $\widetilde{\mathbf{H}}_{\mathrm{ab}} \leftarrow \mathrm{MHCA}(\mathbf{H}_{\mathrm{ab}}^{\mathrm{intra}}, \mathbf{H}_{\mathrm{ag}}^{\mathrm{intra}}, \mathbf{H}_{\mathrm{ag}}^{\mathrm{intra}})$\;
  
  \tcp{Combine intra-chain and cross-chain information}
  $\mathbf{H}_{\mathrm{ag}}^{\ell} \leftarrow \mathbf{H}_{\mathrm{ag}}^{(\ell-1)} + \mathbf{H}_{\mathrm{ag}}^{\mathrm{intra}} + \alpha_{\mathrm{ag}} \operatorname{FFN}(\widetilde{\mathbf{H}}_{\mathrm{ag}})$\;
  $\mathbf{H}_{\mathrm{ab}}^{\ell} \leftarrow \mathbf{H}_{\mathrm{ab}}^{(\ell-1)} + \mathbf{H}_{\mathrm{ab}}^{\mathrm{intra}} + \alpha_{\mathrm{ab}} \operatorname{FFN}(\widetilde{\mathbf{H}}_{\mathrm{ab}})$\;
}

\BlankLine
\tcp{\textbf{Decoder: Cross-Attention Refinement}}
Initialize decoder embeddings: $\mathbf{H}_{\mathrm{ag}}^{\mathrm{dec}} \leftarrow \mathbf{H}_{\mathrm{ag}}^{L}$, $\mathbf{H}_{\mathrm{ab}}^{\mathrm{dec}} \leftarrow \mathbf{H}_{\mathrm{ab}}^{L}$\;
\For{layer $\ell = 1$ \KwTo $L$}{

  \tcp{Inter-chain cross-attention}
  $\widetilde{\mathbf{H}}_{\mathrm{ag}}^{\mathrm{dec}} \leftarrow \mathrm{MHCA}(\mathbf{H}_{\mathrm{ag}}^{\mathrm{dec}}, \mathbf{H}_{\mathrm{ab}}^{\mathrm{dec}}, \mathbf{H}_{\mathrm{ab}}^{\mathrm{dec}})$\;
  $\widetilde{\mathbf{H}}_{\mathrm{ab}}^{\mathrm{dec}} \leftarrow \mathrm{MHCA}(\mathbf{H}_{\mathrm{ab}}^{\mathrm{dec}}, \mathbf{H}_{\mathrm{ag}}^{\mathrm{dec}}, \mathbf{H}_{\mathrm{ag}}^{\mathrm{dec}})$\;
  
  \tcp{Combine intra-chain and cross-chain information}
  $\mathbf{H}_{\mathrm{ag}}^{{\mathrm{dec}}(\ell)} \leftarrow \mathbf{H}_{\mathrm{ag}}^{{\mathrm{dec}}(\ell-1)} +  \operatorname{FFN}(\widetilde{\mathbf{H}}_{\mathrm{ag}}^{\mathrm{dec}})$\;
  $\mathbf{H}_{\mathrm{ab}}^{{\mathrm{dec}}(\ell)} \leftarrow \mathbf{H}_{\mathrm{ab}}^{{\mathrm{dec}}(\ell-1)} +  \operatorname{FFN}(\widetilde{\mathbf{H}}_{\mathrm{ab}}^{\mathrm{dec}})$\;
  
  % Apply bidirectional cross-attention between antigen and antibody\;
  % Apply position-wise feed-forward networks\;
  % Update $\mathbf{H}_{\mathrm{ag}}^{\mathrm{dec}}$ and $\mathbf{H}_{\mathrm{ab}}^{\mathrm{dec}}$ with residual connections\;
}

\BlankLine
\tcp{\textbf{Bipartite Interaction Prediction}}
Compute bidirectional attention scores:\;
$\mathbf{S}_{\mathrm{ag}\to\mathrm{ab}} \leftarrow \frac{(\mathbf{H}_{\mathrm{ag}}^{\mathrm{dec}}\mathbf{W}^{\text{out}}_{Q})(\mathbf{H}_{\mathrm{ab}}^{\mathrm{dec}}\mathbf{W}^{\text{out}}_{K})^{\top}}{\sqrt{d_k}}$\;
$\mathbf{S}_{\mathrm{ab}\to\mathrm{ag}} \leftarrow \frac{(\mathbf{H}_{\mathrm{ab}}^{\mathrm{dec}}\mathbf{W}^{\prime\,\text{out}}_{Q})(\mathbf{H}_{\mathrm{ag}}^{\mathrm{dec}}\mathbf{W}^{\prime\,\text{out}}_{K})^{\top}}{\sqrt{d_k}}$\;

Fuse scores and apply sigmoid:\;
$\mathbf{Z} \leftarrow \mathbf{w}^{\top}[\,\mathbf{S}_{\mathrm{ag}\to\mathrm{ab}}\ (\mathbf{S}_{\mathrm{ab}\to\mathrm{ag}})^{\top}\,] + b$\;
$\hat{\mathcal{E}}_{\text{bg}} \leftarrow \sigma(\mathbf{Z})$\;

\BlankLine
\tcp{\textbf{Epitope Extraction }}
Extract per-residue epitope probabilities via mean pooling:\;
$(\hat{y}_{\text{ag}})_{i} = \frac{1}{m}\sum_{j=1}^{m}(\hat{\mathcal{E}}_{\text{bg}})_{ij}$\;

\SetKwFunction{MHCA}{MHCA}
\SetKwProg{Fn}{Function}{:}{end}
\Fn{\MHCA{$Q, K, V$}}{
    $Q_h \gets QW_Q^h$, $K_h \gets KW_K^h$, $V_h \gets VW_V^h$ \tcp*{Project per head $h$}
    $\alpha_{ij}^h \gets \text{Softmax}_j\left(\frac{Q_{h,i} \cdot K_{h,j}^\top}{\sqrt{d_h}}\right)$ \tcp*{Attention scores}
    $C_i^h \gets \sum_j \alpha_{ij}^h V_{h,j}$ \tcp*{Context vector}
    \KwResult{$\text{Concat}(C^1,\dots,C^H)W_O$} \tcp*{Combine heads}
}
\SetKwFunction{FFN}{FFN}
\Fn{\FFN{$X$}}{
    $\hat{\mathcal{E}}_{\text{bg}} \gets \mathrm{SiLU}(XW_1 + b_1)W_2 + b_2$ \tcp*{$W_1 \in \mathbb{R}^{d \times d_{ff}}$, $W_2 \in \mathbb{R}^{d_{ff} \times d}$}
    \KwResult{$\hat{\mathcal{E}}_{\text{bg}}$}
}
\caption{\emph{EpiFormer}: High-Level Architecture}
\label{alg:epiformer-pseudocode}
\end{algorithm}

\subsection{Graph Construction} \label{appendix:graph-construction}

\subsubsection{Node Features}
Each residue node in our protein graph incorporates two complementary information sources that together provide a rich representation of both local structural properties and evolutionary context:

% Each residue node stores two information sources \textcolor{purple}{H: This is very wordy. Can we reduce it. Refer to previous paper which introduced it and add to Appendix?}:

% \textcolor{red}{mansoor: add references to other papers such as alphafold that use/define such descriptors}

\paragraph{Local geometry \& physicochemistry.}

Each residue $v_i \in \mathcal{V}$ is annotated with a 105-dimensional geometric and biochemical feature vector $\mathbf{h}_i^{\text{geo}} \in \mathbb{R}^{d_{\text{geo}}}$ that encodes the type, position, distance, direction, angle, and orientation of each residue. 
Such residue-level descriptors are widely employed in diverse protein-related studies in structural bioinformatics~\citep{wu2025raad,jing2020gvp_gnn,jumper2021alphafold}.
This vector is constructed as follows:
\begin{equation}\label{equ:node-feat-res-graph}
\mathbf{h}_i^{\text{geo}} = \Big[ E_{\text{type}}(v_i), \ E_{\text{pos}}(i), \ \sin(\eta_i), \ \cos(\eta_i), \ \operatorname{RBF}(\|\mathbf{x}_{i,C_\alpha} - \mathbf{x}_{i,\xi}\|), \ Q_i^\top \frac{\mathbf{x}_{i,\xi} - \mathbf{x}_{i,C_\alpha}}{\|\mathbf{x}_{i,\xi} - \mathbf{x}_{i,C_\alpha}\|} \Big],
\end{equation}
\[
% \mathbf{h}_i^{\text{geo}} \in \mathbb{R}^{d_{\text{geo}}}, \quad d_n = \underbrace{d_{\text{type}}}_{\text{20}} + \underbrace{d_{\text{pos}}}_{\text{16}} + \underbrace{12}_{\sin/\cos \eta_i} + \underbrace{3\,d_{\text{rbf}}}_{\text{48}} + \underbrace{9}_{Q_i^{\top}\mathbf{u}_i},
\]
where:
\begin{itemize}
    \item \textbf{$E_{\text{type}}$}: Embedding for amino acid residue type (e.g., arginine, glycine).
    \item \textbf{$E_{\text{pos}}$}: Positional encoding of residue index in the sequence, enabling the model to distinguish between identical amino acids based on their sequence context. This positional information is crucial for understanding long-range dependencies and structural motifs, as amino acids at different sequence positions (N-terminus vs. C-terminus, loop regions vs. secondary structures) often play different functional roles even if they are the same amino acid type.
    \item \textbf{$\eta_i$}: Local backbone geometry encoded through six fundamental angles that determine how the protein chain folds at each residue $v_i$ and are encoded by their sine and cosine (12 scalars). 
    Bond angles ($\alpha_i$, $\beta_i$, $\gamma_i$) describe the geometric constraints of covalent bonds, while dihedral angles ($\psi_i$, $\phi_i$, $\omega_i$) capture the rotational freedom that gives rise to secondary structures like helices and sheets.

    \item \textbf{$\operatorname{RBF}(\cdot)$}: Radial basis function encoding distances between $\text{C}_\alpha$ and other backbone atoms ($\xi \in \{\text{C}_\beta, \text{N}, \text{O}\}$), with each distance represented by 16 Gaussian basis functions.
    % \item \textbf{$Q_i^{\top}\mathbf{u}_i \in \mathbb{R}^{3 \times 3}$}: Local coordinate frame derived from $\text{C}_\alpha$, $\text{C}_\beta$, and $\text{N}$ atoms, ensuring E(3)-equivariance. \textcolor{teal}{Huirong: Issue: No definition of $\mathbf{u}_i$ . }
    
    \item \textbf{$Q_i^{\top}\mathbf{u}_i$}: Here, $Q_i \in \mathbb{R}^{3 \times 3}$ is the orthonormal rotation matrix defining the local coordinate system constructed from the $\text{C}_\alpha$, $\text{C}_\beta$, and $\text{N}$ atoms of residue $i$, and $\mathbf{u}_i = [\mathbf{u}_i^{1}, \mathbf{u}_i^{2}, \mathbf{u}_i^{3}] \in \mathbb{R}^{3 \times 3}$ contains the normalized direction vectors between these atoms (e.g.,
    $\mathbf{u}_i^{1} = \frac{\mathbf{x}_{i,\text{C}_\beta} - \mathbf{x}_{i,\text{C}_\alpha}}{\|\mathbf{x}_{i,\text{C}_\beta} - \mathbf{x}_{i,\text{C}_\alpha}\|}$).
    The matrix product $Q_i^{\top}\mathbf{u}_i$ transforms these direction vectors into the local coordinate frame and is flattened to yield a 9-dimensional feature vector. Note that the oxygen atom is stored in the coordinate matrix for other calculations (like the RBF distance features), but isn't used for the local coordinate frame construction.
    
    \item The coordinates are held in a $3 \times 4$ matrix, which is used in the calculation of node and edge features.  
    % \textcolor{teal}{Huirong: Issue: I thought we want to use $\mathbf{x}_i$ for $\mathbf{x}_{i,C_\alpha}$? Because it might be an issue that these two are the same but we use both notations.}         % \todo{bold x for the node features}
        \[
        \mathbf{X}_i =
        \begin{bmatrix}
        \mathbf{x}_{i,{\text{N}}} & \mathbf{x}_{i,\text{C}_\alpha} & \mathbf{x}_{i,\text{C}_{\beta}} & \mathbf{x}_{i,\text{O}} 
        \end{bmatrix} \quad \in \mathbb{R}^{3 \times 4}, \quad \text{where} \quad
        \mathbf{x}_{i,\xi} \in \mathbb{R}^{3}
        \]
\end{itemize}

\subsubsection{Edge features}

% \textcolor{purple}{H: Quite wordy, but not easy to follow} 
We compute a 100-dimensional edge feature vector $\mathbf{f}_{i,j} \in \mathbb{R}^{d_f}$ that describes the spatial and sequential relationship between two residues $v_i$ and $v_j$.
This vector integrates multiple complementary descriptors to provide a rich representation of inter-residue interactions~\citep{jing2020gvp_gnn} and is defined as follows: 
% \textcolor{teal}{Huirong: Issue: same notation issue $\mathbf{x}_i$ for $\mathbf{x}_{i,C_\alpha}$?} \murray{I think using both interchangeably is okay, whichever is more convenient --- I left a note about this convenience in the statement of $\mathbf{x}_i$ being short for $\mathbf{x}_{i,C_\alpha}$ above}
\begin{equation}
\hspace{-1.5em}
\begin{small}
\begin{aligned}
    \mathbf{f}_{i,j} = \Big\{E_{\text{type}}&(e_{i,j}), E_{\text{pos}}(i-j), \operatorname{RBF}(\|\mathbf{x}_{i,\text{C}_\alpha} - \mathbf{x}_{j,\xi}\|), Q_i^{\top}\frac{\mathbf{x}_{j,\xi} - \mathbf{x}_{i,\text{C}_\alpha}}{\left\|\mathbf{x}_{j,\xi}-\mathbf{x}_{i,\text{C}_\alpha}\right\|}, q\left(Q_i^{\top} Q_j\right) \ \big| \ \xi\Big\},
\end{aligned}
\end{small}
\end{equation}
% \[
% \mathbf{e}_{ij} \in \mathbb{R}^{d_e}, \quad d_e = d_{\text{type}}^{\text{edge}} + d_{\text{pos}} + 4\,d_{\text{rbf}} + 12 + 4,
% \]

% \todo{arrows b/w atoms -- bold X instead of }

where $E_{\text{type}}(e_{i,j})$ is the one-hot encoding of relations $\mathbf{r}_{i,j}$ of length 4 between two residues, and the positional encoding $E_{\text{pos}}(i-j)$ encodes the relative sequential position sinusoidally to 16 scalars. The third and fourth terms are distance and direction encodings of four backbone atoms $\xi$ in residue $v_j$ in the local coordinate frame $Q_i$. These four inter-residue distances $\{d(\text{C}_\alpha,\text{C}_\beta)$, $d(\text{C}_\alpha,\text{N})$, $d(\text{C}_\alpha,\text{O})$, $d(\text{C}_\alpha,\text{C}_\alpha)\}$) are each represented by 16 Gaussian basis functions.
% \murray{now I am thinking that $\rightarrow$ is not good because it might indicate chemical reaction, e.g., $CH_4 + 2O_2 \rightarrow CO_2 + 2H_2O$.  the dash is no good either, because it might indicate a single bond, e.g., $H_2O = H-O-H$.  Maybe just $d(C_\alpha,C_\beta)$?}. 
The last term $q\left(Q_i^{\top} Q_j\right)$ is the quaternion representation $q(\cdot)$ of $Q_i^{\top} Q_j$. 
By integrating sequence position, local geometry, and orientation, the model understands the residue identity from global pose and enables robust generalization across structures. 
These node and edge features are visualized in Figure~\ref{fig:raad-features}(a).

\subsubsection{Edge Relations}
Since spatial proximity between residues alone cannot capture hydrogen bonding’s directional specificity or electrostatic complementarity’s charge-based selectivity, we use multi-relational edges to capture distinct interaction types~\citep{zhang2022gearnet}. By treating each relation separately, the model learns complex interaction patterns within the protein. 
Hence, to expand the contexts of these interactions, we divide the edges into four different types of relations $\mathcal{R}=\{\rho_1, \rho_2, \rho_3,\rho_4\}$, including \textit{(i)} \textbf{sequential relations} $\rho_1$ and $\rho_2$ between two residues with relative sequential distance equal to 1 (peptide bond) and 2 (short-range torsion coupling); \textit{(ii)} \textbf{ spatial relations} between residues that are from the same component and spatially connected due to $K$-nearest neighbors (relation $\rho_3$ that captures local packing shell) or with a Euclidean distance less than 8\AA\ (relation $\rho_4$) capturing medium-range contact between residues within the protein structure~\citep{wu2025raad}.

To illustrate the importance of edge relations, consider a discontinuous epitope that spans two antigen loops. The sequential edges ($\rho_1,\rho_2$) maintain the structural integrity of each loop, while the spatial edges ($\rho_3,\rho_4$) capture the three-dimensional proximity between residues from different loops, allowing the model to represent how distant sequence regions come together to form a cohesive binding interface.
We provide a schematic of edge relations in Fig.~\ref{fig:raad-features} (b), where each edge $e_{i,j}\in\mathcal{E}$ is associated with a set of relations $\mathbf{r}_{i,j}\in\mathcal{R}$. Besides, two relations $\rho_1$ (with sequence distance equal to 1) can derive a relation $\rho_2$ (with sequence distance equal to 2), while an edge may connect two nodes (residues) due to both relations $\rho_3$ and $\rho_4$.

% To illustrate the importance of edge relations, consider a discontinuous epitope spanning two antigen loops and sequential edges (\murray{$r_1,r_2$}) maintain loop integrity, while hydrogen bonds ($r_3$) and spatial edges (\murray{$r_4$}) stabilize interactions with the antibody's paratope. 

\begin{figure}[h!]
    \centering
    \includegraphics[width=0.45\linewidth]{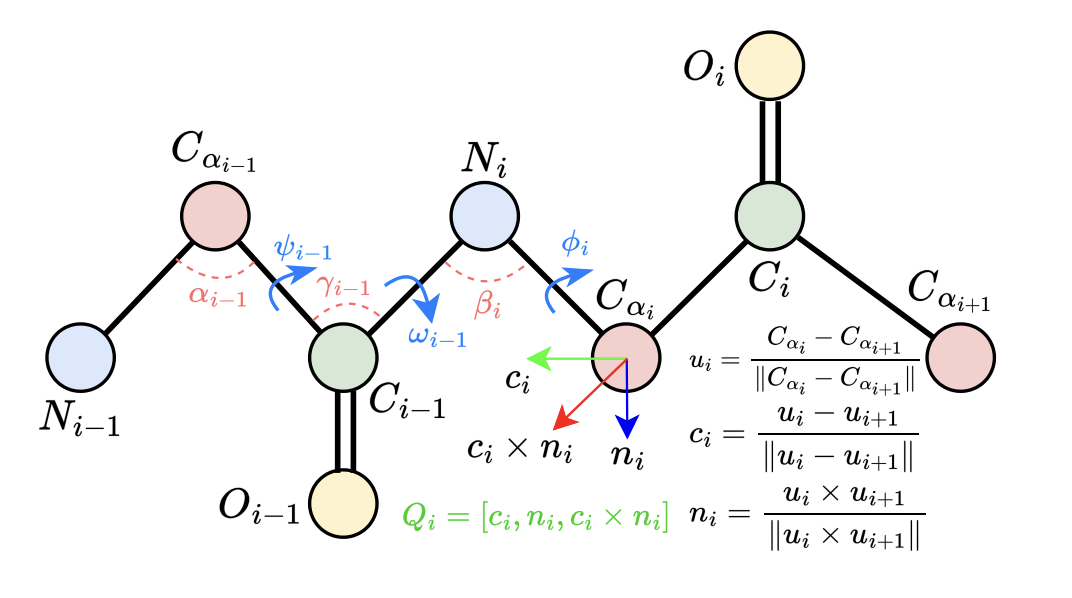}%
    \includegraphics[width=0.45\linewidth]{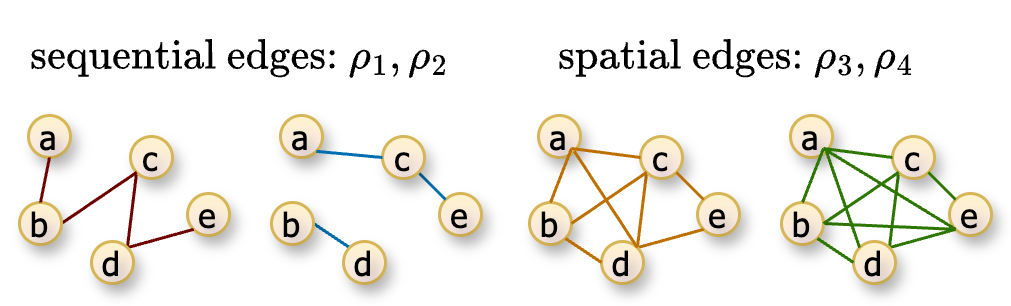}
    \caption{(a) Node and edge features encoding position, distance, direction, angle, and orientation (Figure credit: ~\citep{wu2025raad}). (b) Four edge relations (sequential $\rho_1$, $\rho_2$; spatial $\rho_3$, $\rho_4$ ). To avoid complexity, we visualize only some edges. }
    \label{fig:raad-features}
\end{figure}

% \textcolor{red}{mansoor: add numbering to the nodes to show the multi-relational edges}

\subsection{Auxiliary Distance Classification Loss}\label{appendix:geo_loss}

Let $\mathcal{M}=\{(i,j)\,:\,d_{ij}\le D_{\max}\}$ be the set of antigen-antibody residue pairs within the maximum distance cutoff, where $d_{ij}$ is the Euclidean distance between residues $i$ and $j$. The bins are defined by distances $\{d_0,d_1,d_2,d_3,d_4\}=\{0,4,8,16,32\}\,\text{\AA}$, creating $B=4$ bins:
\begin{align}
b(i,j) = \arg\max_{b\in\{1,\dots,4\}}\mathbf{1}[\,d_{b-1}\le d_{ij}<d_b\,].
\end{align}
The network predicts per-pair distance logits $\boldsymbol{\Delta}_{ij}\in\mathbb{R}^{5}$, but only the first $B=4$ components $\widehat{\boldsymbol{\Delta}}_{ij}\in\mathbb{R}^{4}$ are used for pairs in $\mathcal{M}$, ignoring the ``far'' class beyond $D_{\max}=32$\,\AA. The class probabilities are:
\begin{align}
p_{ijb} = \frac{\exp(\widehat{\Delta}_{ijb})}{\sum_{b'=1}^{4}\exp(\widehat{\Delta}_{ijb'})}.
\end{align}
The loss combines class balancing with distance weighting:
\begin{align}
\mathcal{L}_{\mathrm{geo}} = -\,\frac{1}{|\mathcal{M}|}\sum_{(i,j)\in\mathcal{M}} w_{ij}\, \sum_{b=1}^{4}\alpha_b\,\mathbf{1}[\,b(i,j)=b\,]\, \log p_{ijb},
\end{align}
where $\alpha_b > 0$ are class-balance weights computed from empirical bin frequencies within $\mathcal{M}$ and $w_{ij} > 0$ are distance weights inversely proportional to $d_{ij}$, normalized to unit mean over $\mathcal{M}$.

\subsection{PLM Sensitivity Analysis}\label{appendix:plm_embeddings}
\label{appendix:plm_analysis}

We compared our geometric feature representation against protein language model (PLM) embeddings across EpiFormer and four baselines. PLM embeddings $\mathbf{z}^{\text{plm}}_i \in \mathbb{R}^{d_c}$ are extracted from pre-trained models (ESM-2~\citep{lin2023esm2} for antigens, AntiBERTy for antibodies) and projected via $\mathbf{h}_i^{\text{plm}} = \mathbf{W}_{\text{plm}} \mathbf{z}_i^{\text{plm}}$. Table~\ref{tab:plm_sensitivity} shows that EpiFormer achieves higher performance with geometric features than with PLMs, while all baselines degrade when PLMs are removed. This is consistent with the observation that epitope residues are not evolutionarily conserved~\citep{ponomarenko2007conformational}, making the conservation signals encoded by PLMs uninformative for binding prediction. Additionally, combining high-dimensional PLM features (1,280D) with lower-dimensional geometric features (105D) creates a modality dominance problem~\citep{wu2022characterizing} where the model overfits to the PLM pathway.

% These embeddings are then projected to a low-dimensional subspace through the linear transformation:
% \[
% \mathbf{h}^{\text{plm}}_i = \mathbf{W}_{\text{plm}}\mathbf{z}^{\text{plm}}_i, \; \quad \text{where} \quad \mathbf{W}_{\text{plm}}\in\mathbb{R}^{{d_c}\times d_{\text{plm}}}
% \]

% We compute one vector $\mathbf{z}^{\text{PLM}}_i \in \mathbb{R}^{d_{\text{plm}}}$ for each residue $i$, extracted with a protein-language model (e.g., ESM-2~\citep{lin2023esm2}). 
% The only trainable part touching PLM features is the small bottleneck $\mathbf{W}_{\text{plm}}$ that compresses them.
% \murray{what is $\mathbf{z}_i$, it is never explained.  Also is $\mathbf{h}_i = \mathbf{h}_i^{geo} || \mathbf{h}_i^{plm}$?  i.e., it has the geometric part and PLM part ? it is unclear}

EpiFormer is the only method that improves without PLMs, while methods that rely on PLMs for their primary representations degrade substantially. This suggests that the interleaved cross-attention captures the relational binding signal directly from geometric features, whereas other architectures depend on PLMs to compensate for weaker structural reasoning.

\begin{table}[h!]
\centering
\scriptsize
\caption{PLM sensitivity analysis on the epitope-ratio test set. PLM embeddings are replaced with geometric features. }
\label{tab:plm_sensitivity}
\begin{tabular}{llcccc}
\toprule
\textbf{Method} & \textbf{Input Features} & \textbf{AUC} & \textbf{F1} & \textbf{MCC} & \textbf{$\Delta$ AUC} \\
\midrule
\textbf{EpiFormer} & ESM2-650M + AntiBERTy & 0.889 & 0.433 & 0.404 & --- \\
\textbf{EpiFormer} & \textbf{Geometric} & \textbf{0.924} & \textbf{0.482} & \textbf{0.464} & \textbf{+0.035} \\
\midrule
EquiformerV2 & ESM-2 (default) & 0.815 & 0.241 & 0.243 & --- \\
EquiformerV2 & Geometric & 0.798 & 0.241 & 0.230 & $-0.017$ \\
\midrule
DiscoTope3 & ESM-IF1 (default) & 0.817 & 0.249 & 0.243 & --- \\
DiscoTope3 & Geometric & 0.684 & 0.166 & 0.123 & $-0.133$ \\
\midrule
WALLE & ESM-2 + AntiBERTy & 0.809 & 0.202 & 0.206 & --- \\
WALLE & Geometric & 0.687 & 0.145 & 0.109 & $-0.122$ \\
\midrule
MIPE & ESM + AbLang & 0.827 & 0.337 & 0.356 & --- \\
MIPE & Geometric & 0.739 & 0.148 & 0.184 & $-0.088$ \\
\bottomrule
\end{tabular}
\end{table}

\subsection{Implementation details}\label{appendix:implementation}

The model is trained with an Adam optimizer and a ReduceLROnPlateau learning-rate schedule with decoupled weight decay. The learning rate is selected from the sweep-defined range and fixed at approximately $6.5{e}{-5}$ in the best configuration. A ReduceLROnPlateau scheduler monitors validation performance and decays the learning rate on stagnation, while an early stopping with patience of 10 epochs prevents overfitting and reduces variance in final selection. 
We used SiLU activation functions~\citep{elfwing2018silu} throughout the model because they provide stable gradients via their smooth, non-monotonic curve, which are crucial for training deep graph networks. 
% Experiments with ReLU and GELU resulted in less stable training and were consistent with prior work showing SiLU's advantages in GNN architectures.
The hyperparameter tuning was performed via a Bayesian optimization sweep in Weights \& Biases to minimize the validation loss, and the best hyperparameters were chosen within a predefined search space using bounded uniform and log-uniform distributions. We ran the hyperparameter sweeping experiments separately for the epitope ratio and epitope group split settings.

\begin{itemize}
    % \item The model weight decay was sampled log-uniformly over $[1{e}{-6},\,1e{-5}]$ to prevent overfitting by penalizing large weights.

    % \item The model dropout was sampled log-uniformly over $[0.05,0.5]$ to improve the generalizability of the model, and the best performing configuration used a dropout of 0.132.

    \item The model weight decay was sampled log-uniformly over $[1{e}{-6},\,1{e}{-4}]$ to prevent overfitting by penalizing large weights, with the best configuration using approximately $9.9{e}{-5}$.

    \item The model dropout was sampled log-uniformly over $[0.05,0.5]$ to improve the generalizability of the model, and the best performing configuration for the epitope ratio model used a dropout of 0.132, while the epitope group model used a dropout of 0.053.

    \item The number of layers in the encoder module is treated as a hyperparameter and was chosen from the set $[3,4,5]$ while for the decoder, the number of layers was chosen from the $[2,3,4]$. We experimented with different encoder hidden dimensions and the best configuration of 128 was picked from $[64, 128, 256, 512]$ across different runs. 

    \item We also experimented with different number of attention heads for the encoder and decoder MHCA (2,4,8,16) and picked the best model with 8 attention heads.

    \item A batch size of 8 was chosen from [4,8,16,32] across different runs.

    \item $\alpha_{\text{ag}}$ and $\alpha_{\text{ab}}$ are initialised to $0.05$

\end{itemize}

For the loss coefficients, the best configuration for the epitope-ratio split uses $\lambda_{\mathrm{edge}}=1.0$, $\lambda_{\mathrm{node}}=0.4816$, $\lambda_{\mathrm{geo}}=0.0514$, $\beta_{\mathrm{BCE}}=9.3249$, $\beta_{\mathrm{Dice}}=2.2966$, $\beta_{\mathrm{sparsity}}=0.3068$, $\pi_{\mathrm{epi}}=15.2856$, $\pi_{\mathrm{edge}}=58.7077$, label smoothing $\epsilon=0.1$, and a distance cutoff of 32\,\AA\ for $\mathcal{L}_{\mathrm{geo}}$. 
The best run for the model trained on the epitope group split uses $\lambda_{\mathrm{node}}=0.143$, $\lambda_{\mathrm{edge}}=1.0$,  $\lambda_{\mathrm{geo}}=0.158$, $\beta_{\mathrm{BCE}}=9.16$, $\beta_{\mathrm{Dice}}=1.83$, $\beta_{\mathrm{sparsity}}=0.64$, $\pi_{\mathrm{epi}}=53.18$, $\pi_{\mathrm{edge}}=44.11$, label smoothing $\epsilon=0.1$, and a distance cutoff of 32\,\AA\ for $\mathcal{L}_{\mathrm{geo}}$.

% hparams from previous playful-sweep######

\begin{itemize}

    \item The bipartite edge positive-class weight $\pi_{\mathrm{edge}}$ for the BCE-with-logits interaction loss was sampled log-uniformly over $[30,\,150]$, accommodating variation in pairwise sparsity across complexes.
    
    \item The node objective weight $\lambda_{\mathrm{node}}$ was sampled uniformly over $[0.05,\,0.5]$, exploring the trade-off between residue supervision and the other objectives.

    \item The binary cross-entropy multiplier within the node objective $\beta_{\mathrm{BCE}}$ was drawn uniformly over $[2,\,10]$, spanning weak to strong emphasis on classification error.

    \item The Dice multiplier $\beta_{\mathrm{Dice}}$ was drawn uniformly over $[0.1,\,3.0]$, reflecting its role as a secondary calibrator under class imbalance.

    \item The epitope positive-class weight $\pi_{\mathrm{epi}}$ was sampled log-uniformly over $[10,\,60]$, covering roughly an order of magnitude in imbalance without biasing toward either extreme.

    \item The per-graph epitope count-regularizer weight $\beta_{\mathrm{sparsity}}$ was sampled uniformly over $[0.05,\,1.0]$, enabling calibration of predicted positive counts at the complex level.

    \item The auxiliary distance-classification weight $\lambda_{\mathrm{geo}}$ was sampled uniformly over $[0.05,\,0.3]$, with class balancing across distance bins and distance-aware pair weighting kept enabled and the maximum distance fixed at $32\,\text{\r{A}}$ for all trials.

\end{itemize}

The experiments were performed on an NVIDIA H100 GPU and it took around 35-60 minutes for a single hyperparameter sweeping experiment of around 50 epochs. To ensure full reproducibility of our experiments, we implement random seed management across all computational components, including NumPy (\texttt{numpy.random}), Python (\texttt{random}), PyTorch (\texttt{torch.manual\_seed}), and CUDA operations (\texttt{torch.cuda.manual\_seed\_all}), while additionally controlling worker initialization in data loaders and disabling non-deterministic algorithms (\texttt{torch.backends.cudnn.deterministic=True}).

\subsection{Extended Ablation Studies}\label{appendix:ablations}

We ablate the GNN backbone, pooling strategy, and loss configuration.

\subsubsection{GNNs}

To assess the impact of geometric message passing, we replaced the EGNN-R layers with alternative GNN architectures while keeping all other components fixed (Table~\ref{tab:gnn_backbone}). Even simple backbones like GCN (0.447 F1) outperform all 19 baselines, confirming that the EpiFormer framework, namely its interleaved cross-attention and joint losses, is the primary performance driver. EGNN-R provides an additional $+$0.028 F1 over the best non-equivariant variant (RGCN, .454) through equivariant coordinate updates and multi-relational encoding. We include vanilla EGNN and EGNN-R with frozen coordinates to decompose the contributions, which are $+$0.018 F1 from equivariance (EGNN vs GAT, .470 vs .452), $+$0.007 from multi-relational edges (RGCN vs GCN, .454 vs .447), and $+$0.017 from coordinate updates (EGNN-R vs EGNN-R frozen, .482 vs .465).

% While traditional GNNs like GCN, GIN, and GAT perform competitively but below EGNN-R, which highlights the critical importance of incorporating geometric equivariance for accurate modeling of three-dimensional protein binding interfaces. Figure~\ref{fig:tsne_silhouette} visualizes the learned representations: (a) t-SNE projections show clear separation between epitope and non-epitope residues, and (b) silhouette scores demonstrate progressive improvement in class separability from early encoder layers through the decoder.

\subsubsection{Representation Quality}

Figure~\ref{fig:tsne} visualizes the learned antigen representations via t-SNE, showing that epitope residues form distinct clusters separated from non-epitope residues, with some overlap reflecting the inherent difficulty of the classification task. Figure~\ref{fig:epiformer-model}(a) quantifies this separation, showing that silhouette scores increase progressively from the early encoder layers through the decoder, which confirms that each layer improves class separability. Figure~\ref{fig:epiformer-model}(b) shows the learned gating weights for the default configuration ($\alpha_{\text{init}}{=}0.05$); the antibody gate consistently exceeds the antigen gate across all encoder blocks, consistent with the asymmetry discussed in Section~4.5.

\begin{figure}[h!]
    \centering
    \includegraphics[width=0.48\linewidth]{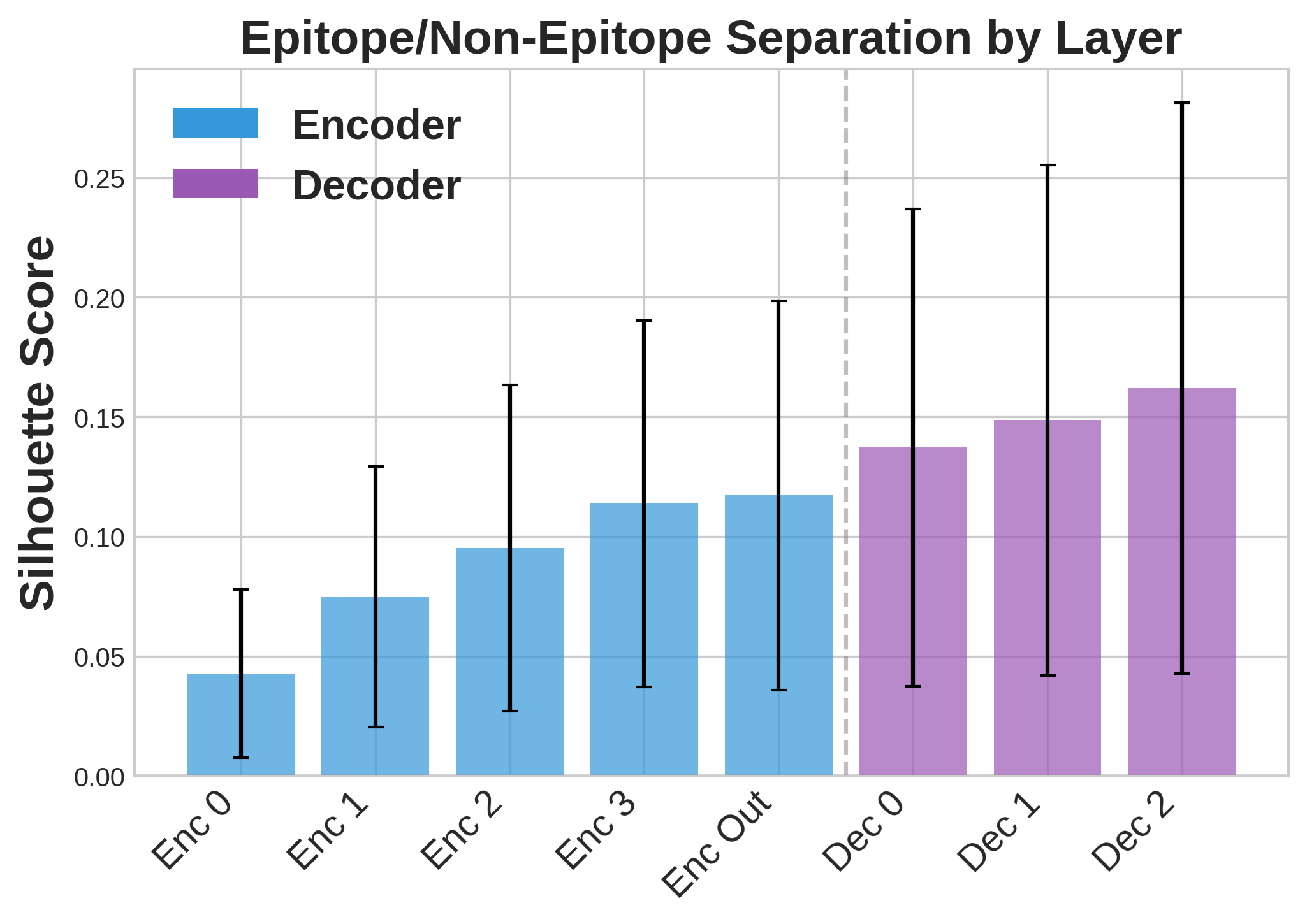} %
    \includegraphics[width=0.48\linewidth]{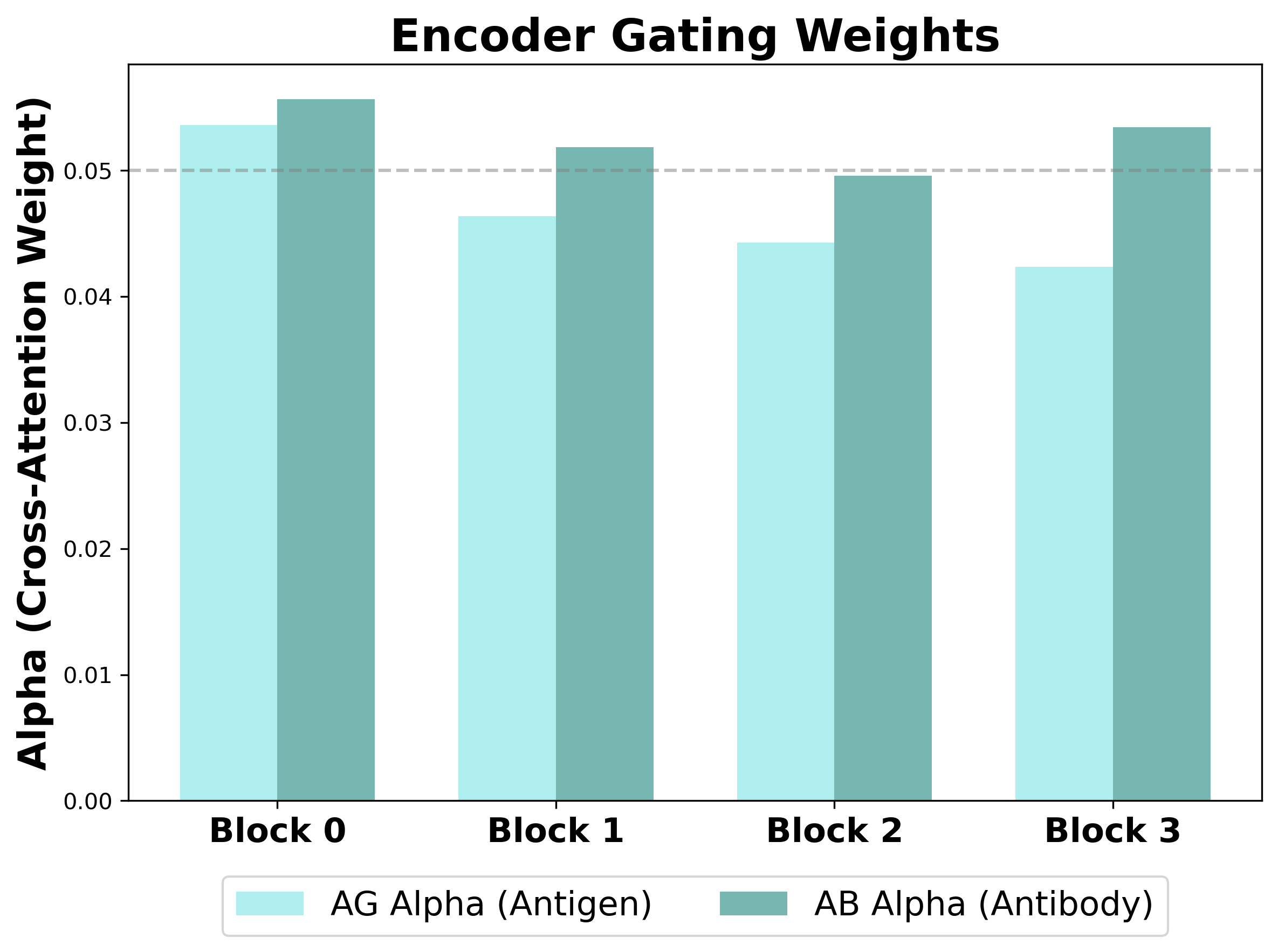}
    \caption{(a) Layer-wise silhouette scores showing progressive class separability from encoder to decoder. (b) Learned gating weights ($\alpha_{\text{ag}}$, $\alpha_{\text{ab}}$) controlling the cross-attention contribution per encoder block (default initialization $\alpha{=}0.05$).}
    \label{fig:epiformer-model}
\end{figure}

We evaluated different strategies for mapping the bipartite interaction matrix ${\hat{\mathcal{E}}}_{\text{bg}}$ to per-residue epitope probabilities (Table~\ref{tab:pooling_ablation}). All pooling strategies converge to a narrow F1 range (0.471--0.483), with mean pooling achieving the best F1 while being the simplest. This suggests that the interleaved encoder produces a well-calibrated interaction matrix largely independent of the downstream aggregation choice.

\begin{figure}[h!]
    \centering
    \includegraphics[width=0.48\linewidth] {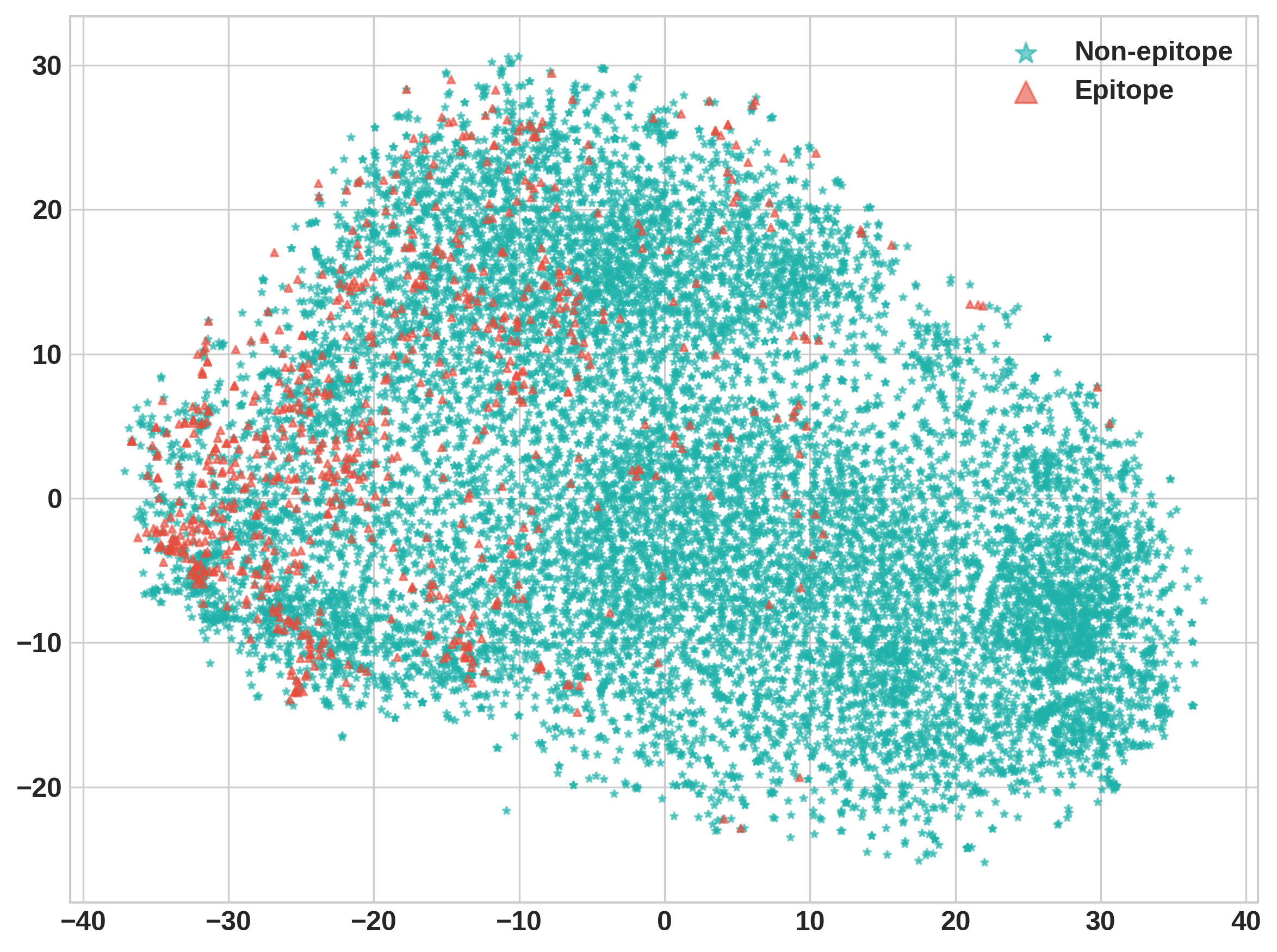} 
    \caption{t-SNE visualization of antigen residue embeddings from the final encoder layer, colored by epitope (red) and non-epitope (teal) labels.}
    \label{fig:tsne}
\end{figure}

\begin{table}[h!]
\centering
\scriptsize
\caption{Pooling strategy comparison on the epitope-ratio split.}
\label{tab:pooling_ablation}
\begin{tabular}{lcccc}
\toprule
\textbf{Pooling} & \textbf{AUC} & \textbf{F1} & \textbf{MCC} \\
\midrule
Max & 0.917 & 0.473 & 0.448 \\
\textbf{Mean} & \textbf{0.921} & \textbf{0.483} & \textbf{0.454} \\
Top-2 & 0.920 & 0.471 & 0.454 \\
Top-3 & 0.923 & 0.481 & 0.458 \\
Softmax Attn. & 0.921 & 0.480 & 0.452 \\
\bottomrule
\end{tabular}
\end{table}

To verify that \emph{EpiFormer}'s gains are not explained by pooling choice, we applied top-$k$ pooling to antibody-aware baselines with explicit interaction matrices (Table~\ref{tab:topk_fairness}). Top-$k$ degrades all four baselines, most severely for MIPE ($-$0.331 F1), while \emph{EpiFormer}'s F1 varies by only 0.011 across strategies.

\begin{table}[h!]
\centering
\small
\caption{Top-$k$ pooling applied to antibody-aware baselines. Original AUC/F1/MCC are the published configurations reported in Table~\ref{tab:baseline_comparison}; $\Delta$ F1 is computed relative to each method's original pooling on the same test split.}
\label{tab:topk_fairness}
\scriptsize
\begin{tabular}{llcccc}
\toprule
\textbf{Method} & \textbf{Pooling} & \textbf{AUC} & \textbf{F1} & \textbf{MCC} & \textbf{$\Delta$ F1} \\
\midrule
WALLE & Original & 0.808 & 0.203 & 0.208 & --- \\
WALLE & + Top-2 & 0.809 & 0.159 & 0.145 & $-0.044$ \\
\midrule
PECAN & Original & 0.740 & 0.229 & 0.182 & --- \\
PECAN & + Top-3 + edge loss & 0.641 & 0.198 & 0.118 & $-0.031$ \\
\midrule
MIPE & Original & 0.827 & 0.337 & 0.356 & --- \\
MIPE & + Top-2 & 0.651 & 0.006 & 0.022 & $-0.331$ \\
\midrule
CheapNet & Original & 0.719 & 0.184 & 0.149 & --- \\
CheapNet & + Top-2 + edge loss & 0.724 & 0.175 & 0.141 & $-0.009$ \\
\midrule
\textbf{EpiFormer} & Mean pooling & 0.924 & 0.482 & 0.464 & --- \\
\textbf{EpiFormer} & Top-2 pooling & 0.920 & 0.471 & 0.454 & $-0.011$ \\
\textbf{EpiFormer} & Top-3 pooling & 0.923 & 0.481 & 0.458 & $-0.001$ \\
\bottomrule
\end{tabular}
\end{table}

\subsubsection{Loss}

We performed ablations to evaluate the contribution of each loss component (primary, auxiliary, and regularizers) to epitope prediction, as shown in Table~\ref{tab:loss_ablation}.

% \textcolor{red}{also add the definition of infonce loss that we used in the ablations and mention why it doesn't help much with reference to the regcl paper of conflict in the optimization with bce}

\paragraph{Contrastive Learning Loss ($\mathcal{L}_{\text{InfoNCE}}$).}
We also performed contrastive learning with the SimCLR InfoNCE (Information Noise Contrastive Estimation) loss~\citep{chen2020infonce} to learn discriminative representations by contrasting positive and negative residue pairs within and across protein chains. The contrastive loss combines intra-chain and inter-chain objectives:
\begin{equation}
    \mathcal{L}_{\text{contrastive}} = \lambda_{\text{intra}} \mathcal{L}_{\text{intra}} + \lambda_{\text{inter}} \mathcal{L}_{\text{inter}},
\end{equation}
where $\lambda_{\text{intra}}$ and $\lambda_{\text{inter}}$ balance the relative importance of within-chain and cross-chain contrastive learning.

\subparagraph{Intra-Chain Contrastive Loss ($\mathcal{L}_{\text{intra}}$).}
The intra-chain loss encourages similar representations for residues with the same label (epitope/non-epitope or paratope/non-paratope) within each protein chain:
\begin{equation}
    \mathcal{L}_{\text{intra}} = \mathcal{L}_{\text{intra}}^{\text{ag}} + \mathcal{L}_{\text{intra}}^{\text{ab}}.
\end{equation}
For each chain (antigen or antibody), the loss is computed as:
\begin{equation}
    \mathcal{L}_{\text{intra}}^{\text{chain}} = -\frac{1}{|\mathcal{P}|} \sum_{i \in \mathcal{P}} \log \frac{\sum_{j \in \mathcal{P}_{i^+}} \exp(\mathbf{h}_i^T \mathbf{h}_j / \tau)}{\sum_{k \in \mathcal{N}_i} \exp(\mathbf{h}_i^T \mathbf{h}_k / \tau)},
\end{equation}
where $\mathcal{P} = \{i : y_i = 1\}$ is the set of positive (binding) residues, $\mathcal{P}_{i^+} = \{j \in \mathcal{P} : j \neq i\}$ are other positive residues sharing the same label as anchor $i$, $\mathcal{N}_i$ includes all negative residues for anchor $i$, $\mathbf{h}_i, \mathbf{h}_j$ are $L_2$-normalized residue embeddings, and $\tau$ is the temperature parameter controlling concentration.

\subparagraph{Inter-Chain Contrastive Loss ($\mathcal{L}_{\text{inter}}$).}
The inter-chain loss promotes alignment between epitope and paratope representations across antigen-antibody pairs:
\begin{equation}
    \mathcal{L}_{\text{inter}} = \mathcal{L}_{\text{ag} \to \text{ab}} + \mathcal{L}_{\text{ab} \to \text{ag}}.
\end{equation}
The bidirectional formulation ensures symmetric learning:
\begin{equation}
    \mathcal{L}_{\text{ag} \to \text{ab}} = -\frac{1}{|\mathcal{P}_{\text{ag}}|} \sum_{i \in \mathcal{P}_{\text{ag}}} \log \frac{\sum_{j \in \mathcal{P}_{\text{ab}}} \exp(\mathbf{h}_i^{\text{ag}T} \mathbf{h}_j^{\text{ab}} / \tau)}{\sum_{k \in \mathcal{N}_{\text{cross}}} \exp(\mathbf{h}_i^{\text{ag}T} \mathbf{h}_k / \tau)},
\end{equation}
where $\mathcal{P}_{\text{ag}}, \mathcal{P}_{\text{ab}}$ are epitope and paratope residue sets, $\mathcal{N}_{\text{cross}}$ includes negative residues from both chains, and the loss pulls epitope embeddings closer to paratope embeddings while pushing them away from non-binding residues. Our experiments show that contrastive learning didn't contribute to improving the classification performance. We attribute this to conflicting optimization objectives between BCE loss and standard InfoNCE loss, a phenomenon demonstrated in a recent work~\citep{ji2024regcl}.

\begin{table}[h!]
\centering
\caption{Loss function ablation on the epitope-ratio split. Each row adds the listed term to the configuration above it; best in \textbf{bold}.}
\label{tab:loss_ablation}
\small
\begin{tabular}{lcccc}
\toprule
\textbf{Loss Config} & \textbf{AUC} & \textbf{F1} & \textbf{MCC} & \textbf{$\Delta$ F1} \\
\midrule
$\mathcal{L}_{\mathrm{BCE}}$ & .916 & .336 & .366 & $-.160$ \\
$+\,\mathcal{L}_{\mathrm{edge}}$ & .916 & .323 & .356 & $-.173$ \\
$+\,\mathcal{L}_{\mathrm{Dice}}$ & .916 & .329 & .361 & $-.167$ \\
$+\,\mathcal{L}_{\mathrm{sparsity}}$ & .917 & .461 & .440 & $-.035$ \\
$\textbf{+}\,\mathcal{L}_{\mathrm{geo}}$ \textbf{(full)} & \textbf{.926} & \textbf{.496} & \textbf{.475} & --- \\
$+\,\mathcal{L}_{\mathrm{InfoNCE}}$ & .920 & .467 & .443 & $-.029$ \\
\bottomrule
\end{tabular}
\end{table}

\subsection{Cross-Dataset Evaluation}\label{appendix:crossdataset}

We evaluated EpiFormer on three external benchmarks with no overlap with AsEP training data, SAbDab~\citep{dunbar2014sabdab} (494 diverse antibody-antigen complexes from the Structural Antibody Database, post-2023), CoV-AbDab~\citep{raybould2021cov} (170 SARS-CoV-2 neutralizing antibody complexes), and ANABAG~\citep{grandguillaume2025anabag} (499 curated complexes). We compare EpiFormer against MIPE, which is the second-best method on AsEP (Table~\ref{tab:baseline_comparison}) and the strongest antibody-aware baseline. Table~\ref{tab:crossdataset_full} reports both zero-shot (no retraining) and leave-one-dataset-out (LODO) fine-tuning settings.

\begin{table}[h!]
\centering
\small
\caption{Cross-dataset generalization on SAbDab~\citep{dunbar2014sabdab}, CoV-AbDab~\citep{raybould2021cov}, and ANABAG. Zero-shot: models trained on AsEP evaluated without retraining. LODO (leave-one-dataset-out): models fine-tuned on the two external datasets not used for testing. Per-dataset best in \textbf{bold}.}
\label{tab:crossdataset_full}
\begin{tabular}{llcccc}
\toprule
& & \multicolumn{2}{c}{\textbf{Zero-shot}} & \multicolumn{2}{c}{\textbf{LODO}} \\
\cmidrule(lr){3-4} \cmidrule(lr){5-6}
\textbf{Dataset} & \textbf{Method} & \textbf{AUC} & \textbf{F1} & \textbf{AUC} & \textbf{F1} \\
\midrule
\multirow{2}{*}{SAbDab (494)} & \textbf{EpiFormer} & \textbf{0.782} & \textbf{0.297} & \textbf{0.858} & \textbf{0.364} \\
& MIPE & 0.614 & 0.140 & 0.746 & 0.190 \\
\midrule
\multirow{2}{*}{CoV-AbDab (170)} & \textbf{EpiFormer} & \textbf{0.819} & \textbf{0.358} & \textbf{0.890} & \textbf{0.359} \\
& MIPE & 0.730 & 0.175 & 0.813 & 0.212 \\
\midrule
\multirow{2}{*}{ANABAG (499)} & \textbf{EpiFormer} & \textbf{0.758} & \textbf{0.301} & 0.837 & 0.367 \\
& MIPE & 0.719 & 0.220 & \textbf{0.772} & \textbf{0.217} \\
\bottomrule
\end{tabular}
\end{table}

% \subsection{Hyperparameter Sensitivity}\label{appendix:hparam_sensitivity}

% \begin{figure}[h!]
%     \centering
%     \includegraphics[width=1\linewidth]{figures/v7_hyperparam_sensitivity_dual.png}
%     \caption{Sensitivity of F1 score to seven hyperparameters of the joint objective for epitope-ratio and epitope-group sweeps. Most parameters ($\lambda_{\text{geo}}$, $\beta_{\text{Dice}}$, $\beta_{\text{sparsity}}$, $\pi_{\text{epi}}$) show flat trends, indicating robustness to their exact values.}
%     \label{fig:hparam-sensitivity}
% \end{figure}

\subsection{Computational Cost}\label{appendix:compute}

Table~\ref{tab:compute} compares computational cost across baselines. Despite its cross-attention decoder, EpiFormer requires only 0.174~GB peak VRAM, which is lower than EquiformerV2, CheapNet, and EquiPocket, because it operates on residue-level rather than atom-level representations. Table~\ref{tab:scaling} shows that EpiFormer scales sub-linearly with antigen size ($1.61\times$ from small to large), compared to MIPE's near-quadratic scaling ($9.98\times$).

\begin{table}[h!]
\centering
\scriptsize
\caption{Computational cost on the epitope-ratio test set (H100 GPU).}
\label{tab:compute}
\begin{tabular}{llccccc}
\toprule
\textbf{Method} & \textbf{Category} & \textbf{Params (M)} & \textbf{Train (s/ep)} & \textbf{Infer.\ (ms)} & \textbf{VRAM (GB)} & \textbf{F1} \\
\midrule
DiffDock-PP & AB-aware Eq. & 0.04 & 25.6 & 15.90 & 0.041 & 0.186 \\
ATProt & AB-aware GNN & 0.06 & 50.9 & 26.62 & 0.050 & 0.246 \\
DiscoTope3 & AG-only MLP & 0.18 & 2.0 & 0.94 & 0.036 & 0.249 \\
WALLE & AB-aware & 0.25 & 22.1 & 20.96 & 0.039 & 0.202 \\
EquiformerV2 & Equivariant & 0.59 & 63.5 & 45.74 & 0.418 & 0.255 \\
CheapNet & AB-aware & 0.84 & 48.4 & 31.92 & 0.545 & 0.178 \\
MIPE & AB-aware & 0.96 & 719.4 & 239.31 & 0.073 & 0.337 \\
EquiPocket & AG-only Eq. & 1.26 & 87.7 & 89.27 & 0.473 & 0.202 \\
EpiGraph & AG-only GNN & 5.78 & 1482.4 & 15.63 & 0.192 & 0.078 \\
\textbf{EpiFormer} & \textbf{AB-aware Eq.} & \textbf{5.54} & \textbf{150.0} & \textbf{118.65} & \textbf{0.174} & \textbf{0.482} \\
\bottomrule
\end{tabular}
\end{table}

\begin{table}[h!]
\centering
\scriptsize
\caption{Inference latency scaling by antigen size (ms/sample, batch size 1).}
\label{tab:scaling}
\begin{tabular}{lcccc}
\toprule
\textbf{Antigen Size} & \textbf{\# Samples} & \textbf{WALLE (ms)} & \textbf{MIPE (ms)} & \textbf{EpiFormer (ms)} \\
\midrule
Small ($<$200 res.) & $\sim$100 & 19.16 & 108.93 & 100.56 \\
Medium (200--500 res.) & $\sim$55 & 22.94 & 221.79 & 123.00 \\
Large ($>$500 res.) & $\sim$15 & 24.10 & 1087.26 & 161.62 \\
\midrule
\textbf{Ratio (Large/Small)} & & \textbf{1.26$\times$} & \textbf{9.98$\times$} & \textbf{1.61$\times$} \\
\bottomrule
\end{tabular}
\end{table}

\subsection{MHCA Alpha Initialization Ablation}\label{appendix:mhca_ablation}

Table~\ref{tab:mhca} shows the effect of the cross-attention gate initialization $\alpha$ on performance. Removing encoder cross-attention entirely ($\alpha{=}0.0$) causes a notable drop in F1. Even when initialized at zero, the model learns non-zero gating values, confirming that it discovers cross-attention utility without explicit supervision. The default initialization ($\alpha{=}0.05$) achieves the best results. Figure~\ref{fig:perfcomp} visualizes these trends across all four metrics.

Figure~\ref{fig:attnmaps} shows encoder cross-attention heatmaps for a representative complex (6xq0\_1P) across the four $\alpha$ variants and all encoder blocks. When $\alpha{=}0.0$ (no cross-attention), attention maps remain diffuse with no structure. With cross-attention enabled, attention progressively concentrates on ground-truth epitope and paratope residues across deeper encoder blocks. Figure~\ref{fig:decoderattn} extends this analysis to the decoder, comparing decoder cross-attention maps and predicted interaction matrices against ground-truth contacts. All cross-attention variants produce focused decoder attention, but the $\alpha{=}0.05$ configuration yields the sharpest predicted interaction matrix.

\begin{table}[h!]
\centering
\small
\caption{MHCA gate initialization ablation (epitope-ratio split).}
\label{tab:mhca}
\begin{tabular}{lcccccc}
\toprule
\textbf{Alpha Init} & \textbf{AUC} & \textbf{AUPRC} & \textbf{F1} & \textbf{MCC} & \textbf{Prec.} & \textbf{$\Delta$ F1} \\
\midrule
0.0 (no MHCA) & 0.906 & 0.464 & 0.454 & 0.424 & 0.354 & $-0.042$ \\
\textbf{0.05 (base)} & \textbf{0.926} & \textbf{0.509} & \textbf{0.496} & \textbf{0.475} & \textbf{0.383} & --- \\
0.2 & 0.918 & 0.485 & 0.469 & 0.449 & 0.351 & $-0.027$ \\
0.4 & 0.919 & 0.490 & 0.484 & 0.457 & 0.383 & $-0.012$ \\
\bottomrule
\end{tabular}
\end{table}

\begin{figure}[h!]
\centering
\includegraphics[width=0.95\textwidth]{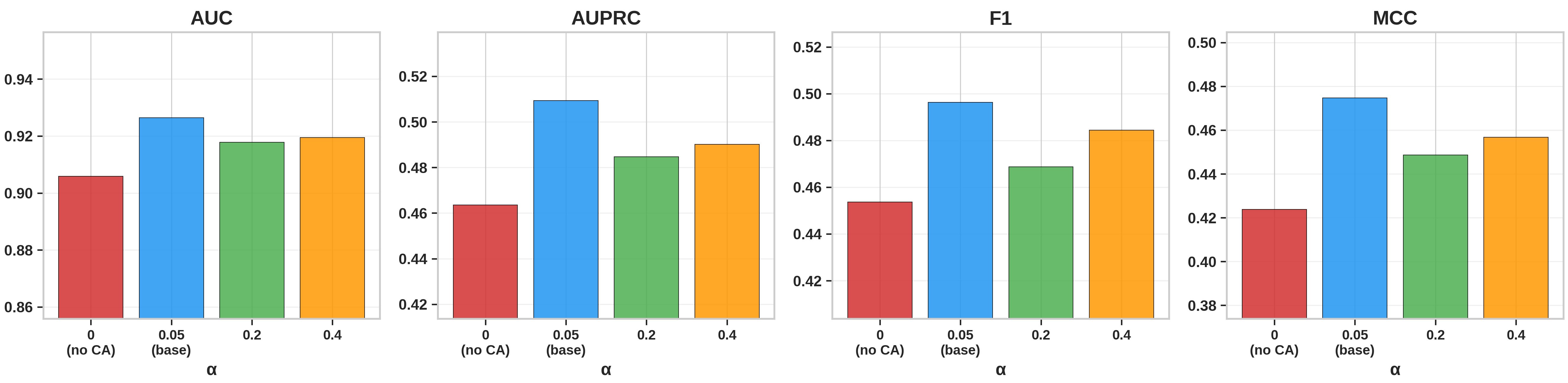}
\caption{Test set performance across four interleaved-MHCA gate initialization values ($\alpha \in \{0.0, 0.05, 0.2, 0.4\}$).}
\label{fig:perfcomp}
\end{figure}

\begin{figure}[h!]
\centering
\includegraphics[width=0.95\textwidth]{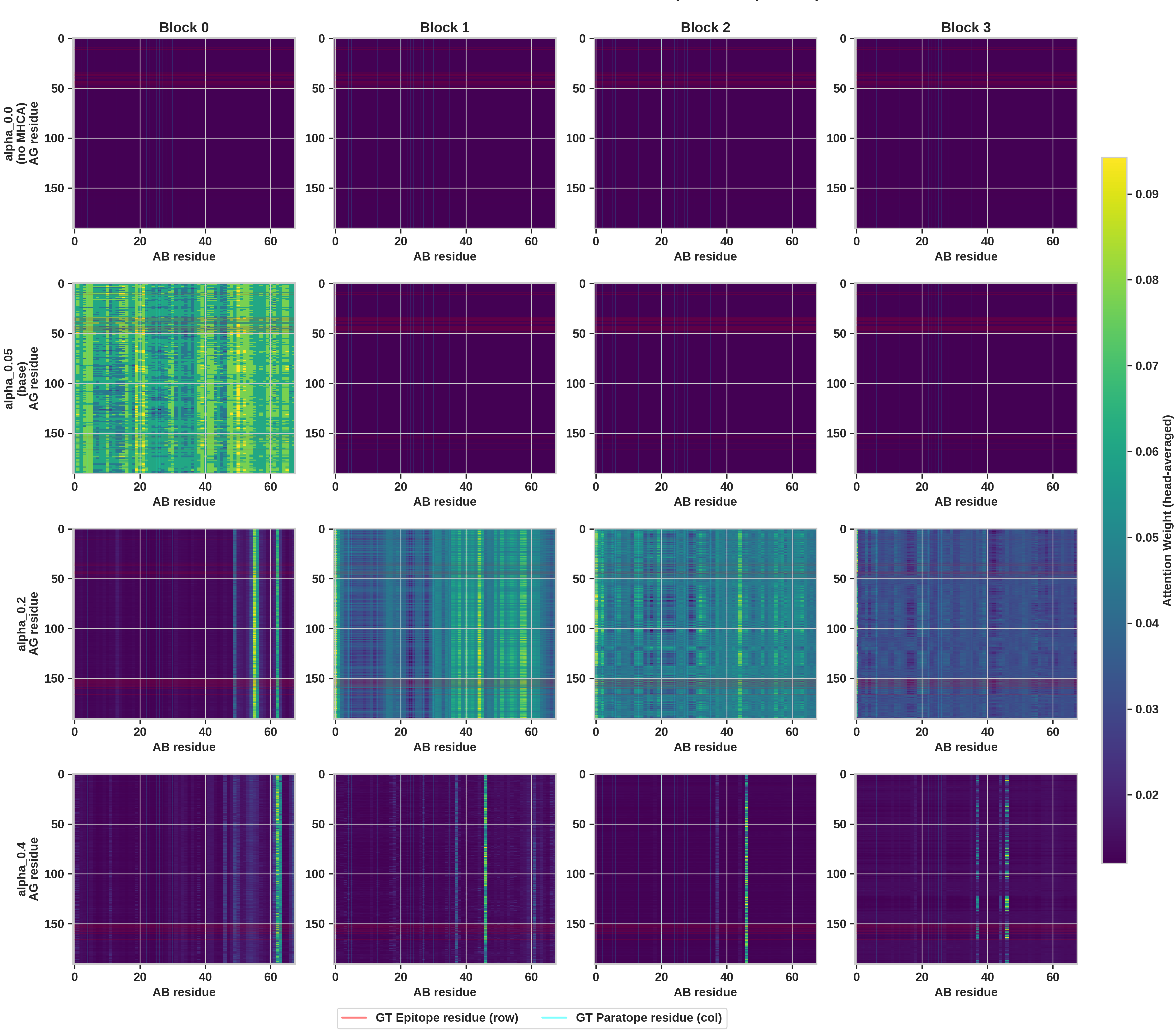}
\caption{Encoder AG$\to$AB cross-attention heatmaps for complex 6xq0\_1P. Rows: $\alpha_{\text{init}} \in \{0.0, 0.05, 0.2, 0.4\}$; columns: encoder blocks 0--3. Ground-truth epitope (red) and paratope (cyan) residues are overlaid.}
\label{fig:attnmaps}
\end{figure}

\begin{figure}[h!]
\centering
\includegraphics[width=0.95\textwidth]{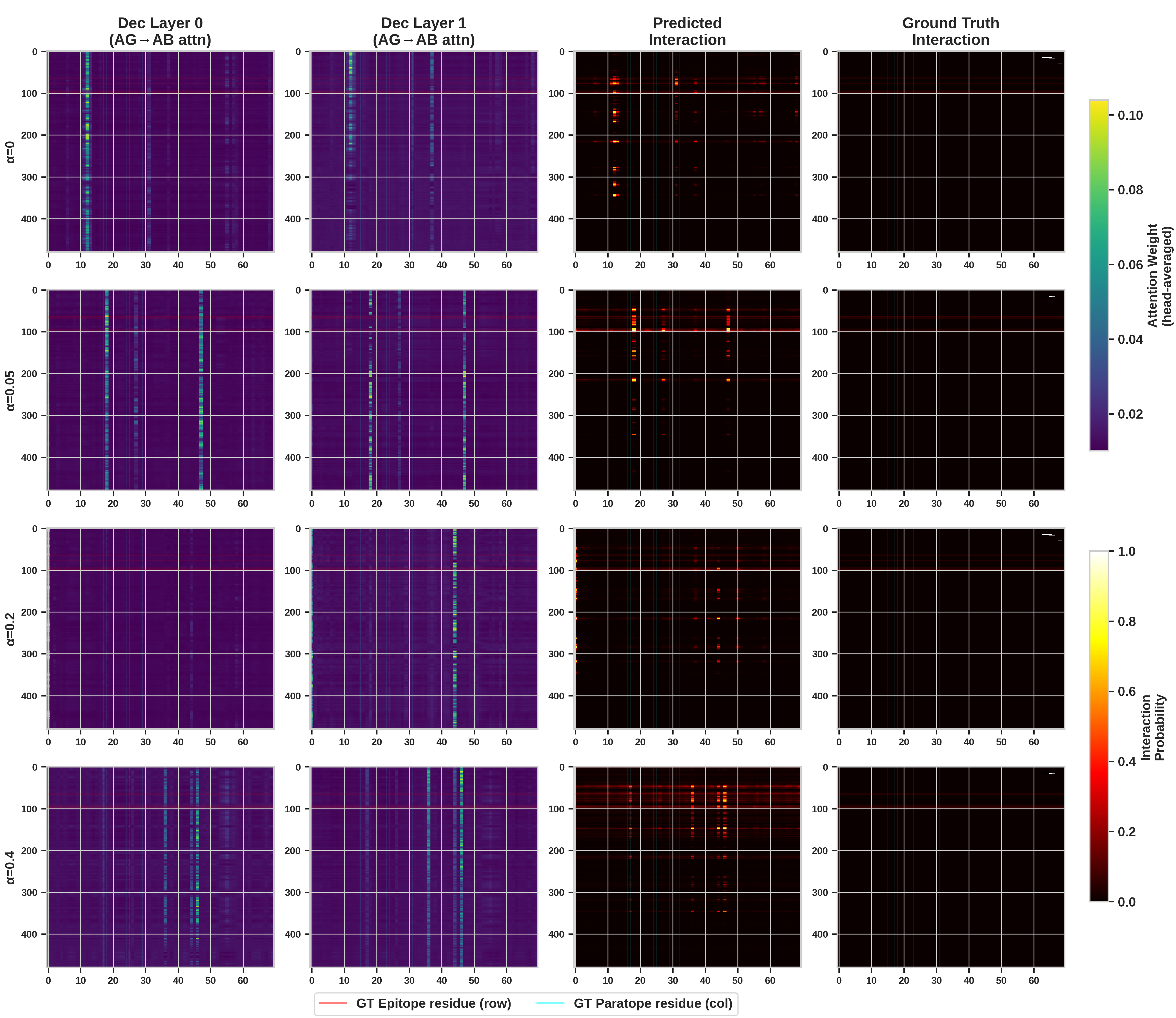}
\caption{Decoder cross-attention and predicted interaction maps vs.\ ground-truth contacts for complex 6xq0\_1P across four $\alpha_{\text{init}}$ variants. Columns: decoder layers 0--1, predicted interaction, ground truth.}
\label{fig:decoderattn}
\end{figure}

\subsection{Coordinate Refinement Analysis}\label{appendix:coord_refinement}

Figure~\ref{fig:paratope} shows per-block coordinate displacement magnitude across the EGNN-R encoder blocks. Displacements decrease across blocks, confirming convergent geometric refinement. Epitope residues undergo consistently larger positional adjustments than non-epitope residues, and paratope residues show a similar pattern on the antibody side. The difference is statistically significant at Block~3 ($p = 1.80 \times 10^{-7}$, Wilcoxon rank-sum test).

\begin{figure}[h!]
\centering
\includegraphics[width=0.95\textwidth]{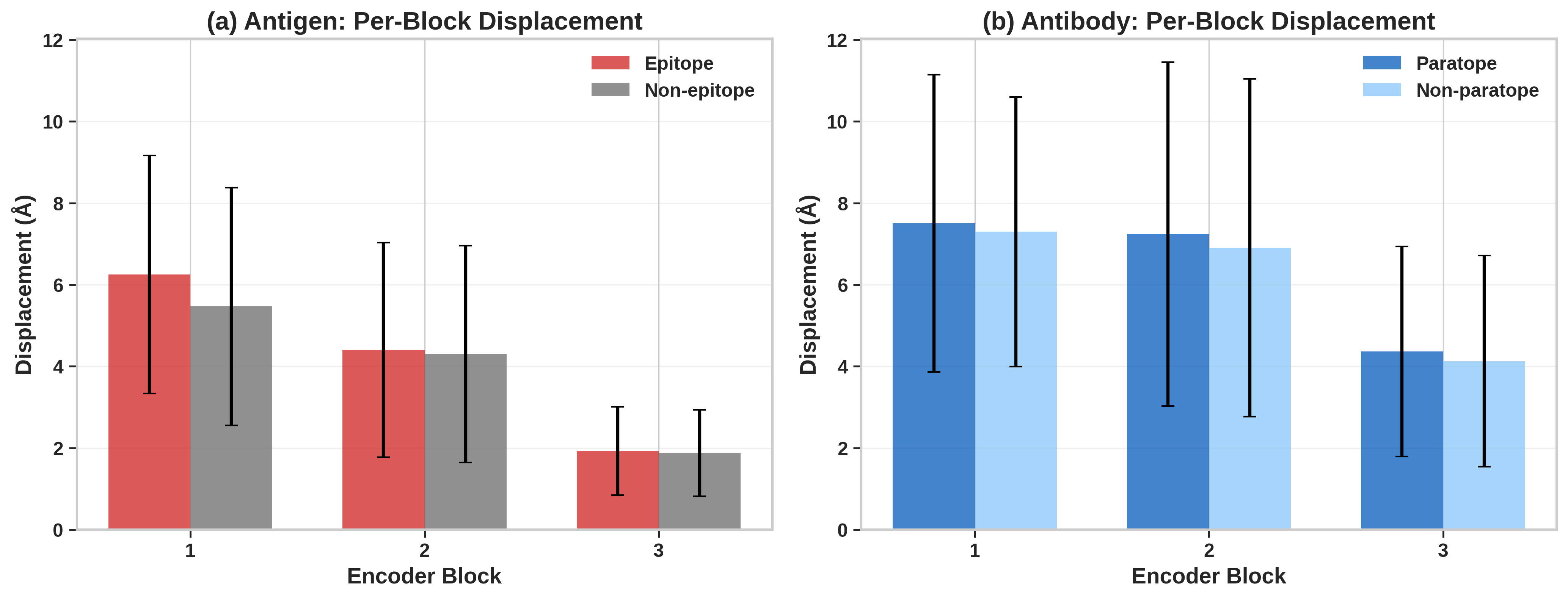}
\caption{Per-block coordinate displacement magnitude across EGNN-R encoder blocks 1--3 (block 0 has no prior coordinates to compare against). (a) Antigen: epitope residues undergo larger displacements than non-epitope residues. (b) Antibody: paratope residues show a similar pattern. Displacements decrease across blocks, indicating convergent refinement.}
\label{fig:paratope}
\end{figure}

\subsection{Resolution Stratification}\label{appendix:resolution}

Table~\ref{tab:resolution} stratifies performance by X-ray crystallographic resolution. EpiFormer maintains strong AUC and F1 across all resolution bins with no catastrophic degradation at low resolution ($>3.0$\,\AA). Figure~\ref{fig:resolution} shows the full per-complex distributions; EpiFormer's advantage over MIPE and WALLE is consistent across all resolution bins and metrics.

\begin{table}[h!]
\centering
\scriptsize
\caption{Performance stratified by X-ray resolution on the epitope-ratio test set (per-complex mean $\pm$ std).}
\label{tab:resolution}
\begin{tabular}{llccccc}
\toprule
\textbf{Resolution} & \textbf{N} & \textbf{Method} & \textbf{AUC} & \textbf{F1} & \textbf{MCC} & \textbf{Rec.} \\
\midrule
\multirow{3}{*}{$<2.0$\,\AA} & 17 & \textbf{EpiFormer} & \textbf{.908$\pm$.071} & \textbf{.622$\pm$.238} & \textbf{.553$\pm$.265} & .721$\pm$.269 \\
& 3 & MIPE & .515$\pm$.212 & .119$\pm$.205 & .011$\pm$.203 & .107$\pm$.185 \\
& 17 & WALLE & .734$\pm$.204 & .372$\pm$.146 & .207$\pm$.209 & \textbf{.964$\pm$.090} \\
\midrule
\multirow{3}{*}{2.0--3.0\,\AA} & 75 & \textbf{EpiFormer} & \textbf{.898$\pm$.080} & \textbf{.484$\pm$.235} & \textbf{.434$\pm$.243} & .553$\pm$.291 \\
& 6 & MIPE & .674$\pm$.221 & .163$\pm$.398 & .145$\pm$.406 & .167$\pm$.408 \\
& 75 & WALLE & .719$\pm$.170 & .238$\pm$.109 & .137$\pm$.119 & \textbf{.915$\pm$.146} \\
\midrule
\multirow{3}{*}{$>3.0$\,\AA} & 78 & \textbf{EpiFormer} & \textbf{.916$\pm$.070} & \textbf{.464$\pm$.206} & \textbf{.437$\pm$.211} & .565$\pm$.271 \\
& 3 & MIPE & .803$\pm$.128 & .000$\pm$.000 & .010$\pm$.011 & .000$\pm$.000 \\
& 78 & WALLE & .748$\pm$.172 & .202$\pm$.096 & .162$\pm$.121 & \textbf{.932$\pm$.167} \\
\bottomrule
\end{tabular}
\end{table}

\begin{figure}[h!]
\centering
\includegraphics[width=0.95\textwidth]{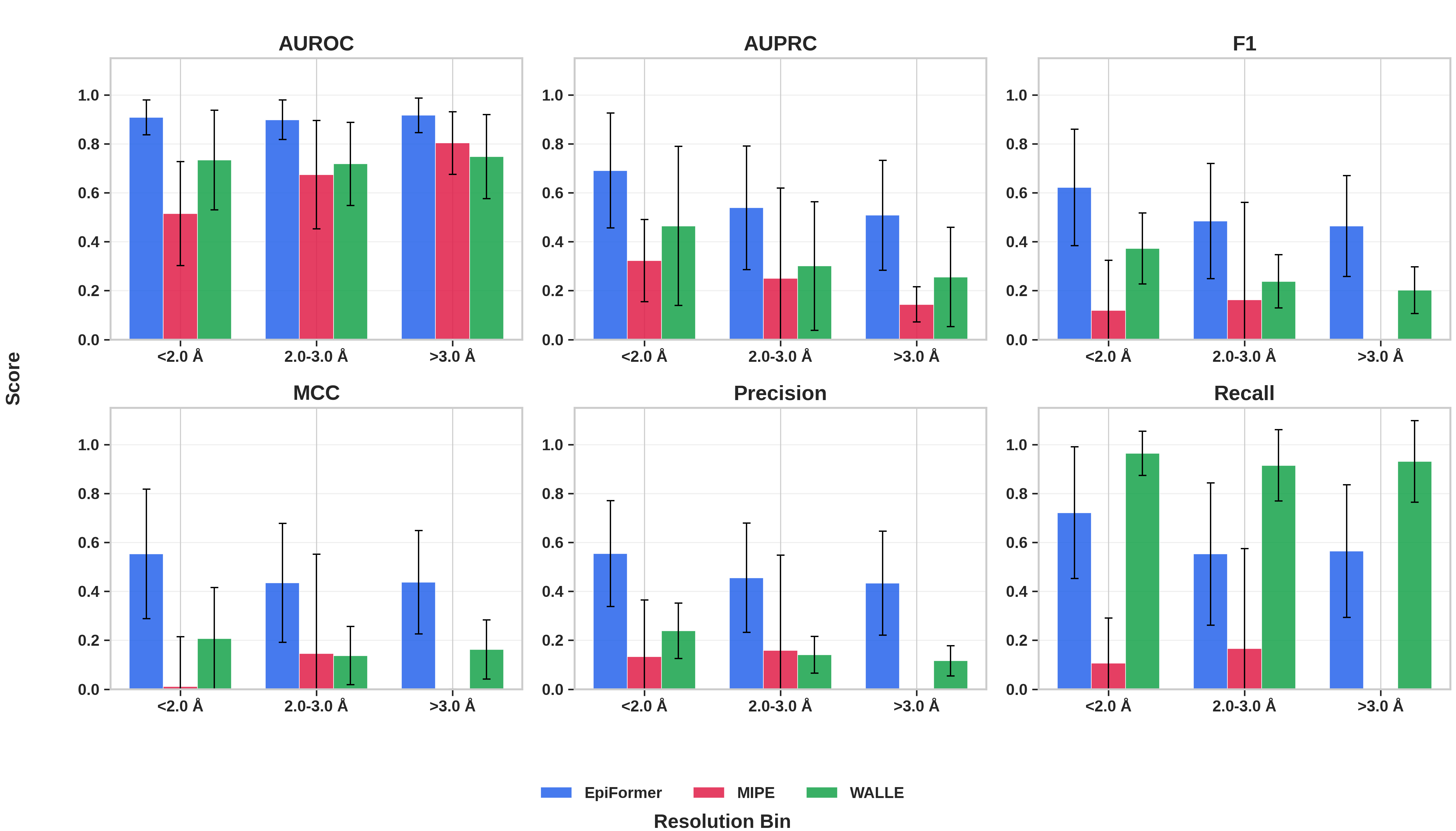}
\caption{Per-complex metric distributions (AUROC, AUPRC, F1, MCC, Precision, Recall) stratified by X-ray resolution for EpiFormer, MIPE, and WALLE.}
\label{fig:resolution}
\end{figure}

\subsection{Antigen Size Stratification}\label{appendix:sizestrat}

Table~\ref{tab:sizestrat} stratifies performance by antigen surface residue count. EpiFormer's AUC improves with antigen size, while F1 declines due to increasing class imbalance (the mean epitope ratio drops from 13.9\% in small antigens to 2.2\% in large ones). MIPE collapses entirely for medium and large antigens. Figure~\ref{fig:sizestrat} shows the full per-complex distributions across all six metrics.

\begin{table}[h!]
\centering
\scriptsize
\caption{Performance stratified by antigen surface residue count on the epitope-ratio test set (per-complex mean $\pm$ std).}
\label{tab:sizestrat}
\begin{tabular}{llcccc}
\toprule
\textbf{Size Bin} & \textbf{N} & \textbf{Epi.\ Ratio} & \textbf{Method} & \textbf{AUC} & \textbf{F1} \\
\midrule
\multirow{3}{*}{$<200$ (small)} & \multirow{3}{*}{90} & \multirow{3}{*}{13.9\%} & \textbf{EpiFormer} & \textbf{.901$\pm$.081} & \textbf{.568$\pm$.202} \\
& & & WALLE & .707$\pm$.180 & .293$\pm$.114 \\
& & & MIPE & .630$\pm$.241 & .222$\pm$.396 \\
\midrule
\multirow{3}{*}{200--500 (med)} & \multirow{3}{*}{62} & \multirow{3}{*}{6.6\%} & \textbf{EpiFormer} & \textbf{.905$\pm$.068} & \textbf{.412$\pm$.206} \\
& & & WALLE & .730$\pm$.156 & .167$\pm$.078 \\
& & & MIPE & .561$\pm$.145 & .000$\pm$.000 \\
\midrule
\multirow{3}{*}{$>500$ (large)} & \multirow{3}{*}{18} & \multirow{3}{*}{2.2\%} & \textbf{EpiFormer} & \textbf{.953$\pm$.047} & \textbf{.356$\pm$.260} \\
& & & WALLE & .876$\pm$.136 & .175$\pm$.090 \\
& & & MIPE & .845$\pm$.109 & .000$\pm$.000 \\
\bottomrule
\end{tabular}
\end{table}

\begin{figure}[h!]
\centering
\includegraphics[width=0.95\textwidth]{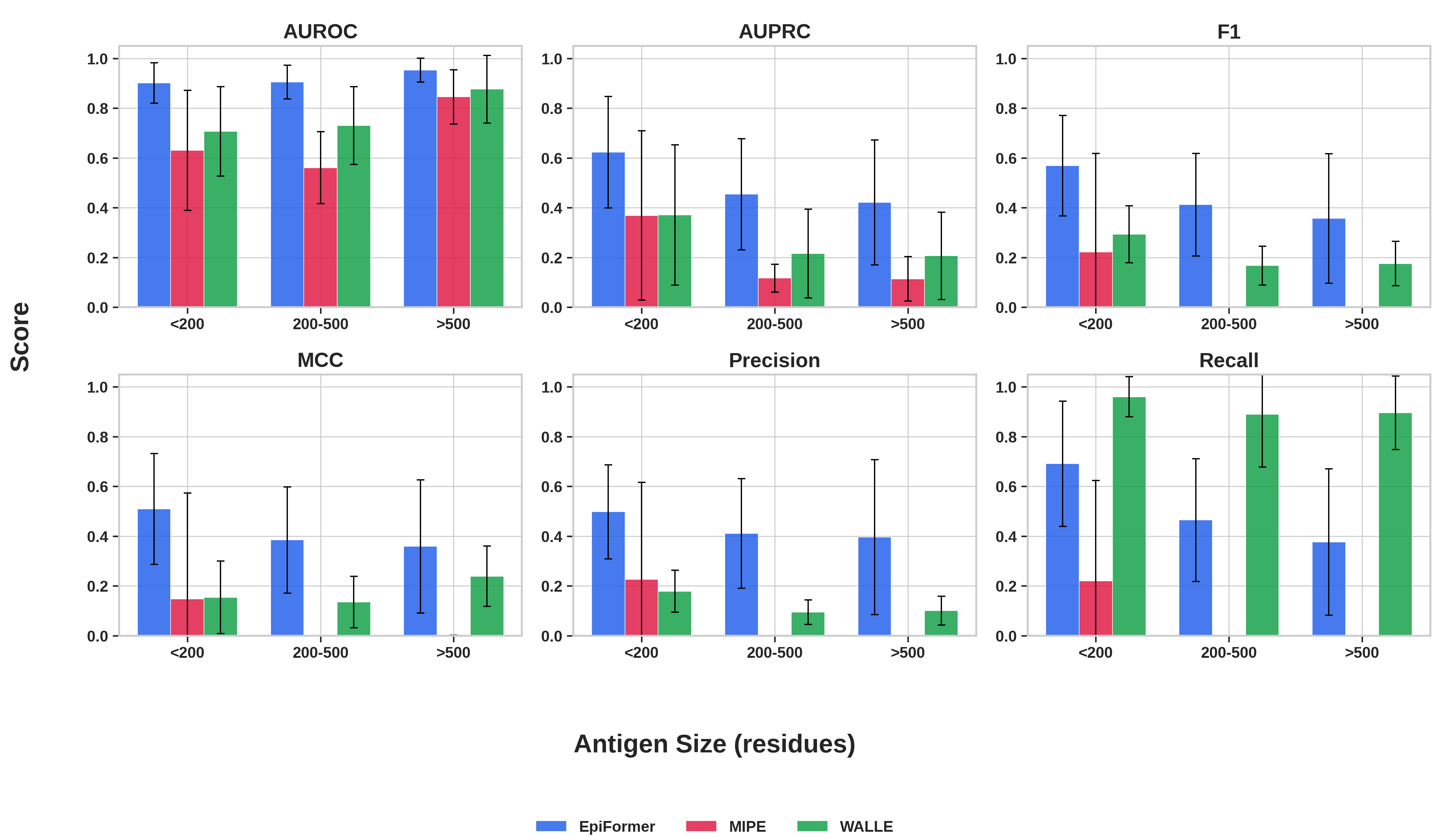}
\caption{Per-complex metric distributions (AUROC, AUPRC, F1, MCC, Precision, Recall) stratified by antigen size for EpiFormer, MIPE, and WALLE.}
\label{fig:sizestrat}
\end{figure}

\subsection{Per-Complex Metric Distributions}\label{appendix:persample}

Figure~\ref{fig:persample} shows per-complex metric distributions on the epitope-ratio test set. EpiFormer achieves higher medians and tighter interquartile ranges than MIPE and WALLE across all six metrics, indicating more consistent predictions. WALLE exhibits high recall but low precision, confirming its tendency to over-predict epitope residues.

\begin{figure}[h!]
\centering
\includegraphics[width=0.95\textwidth]{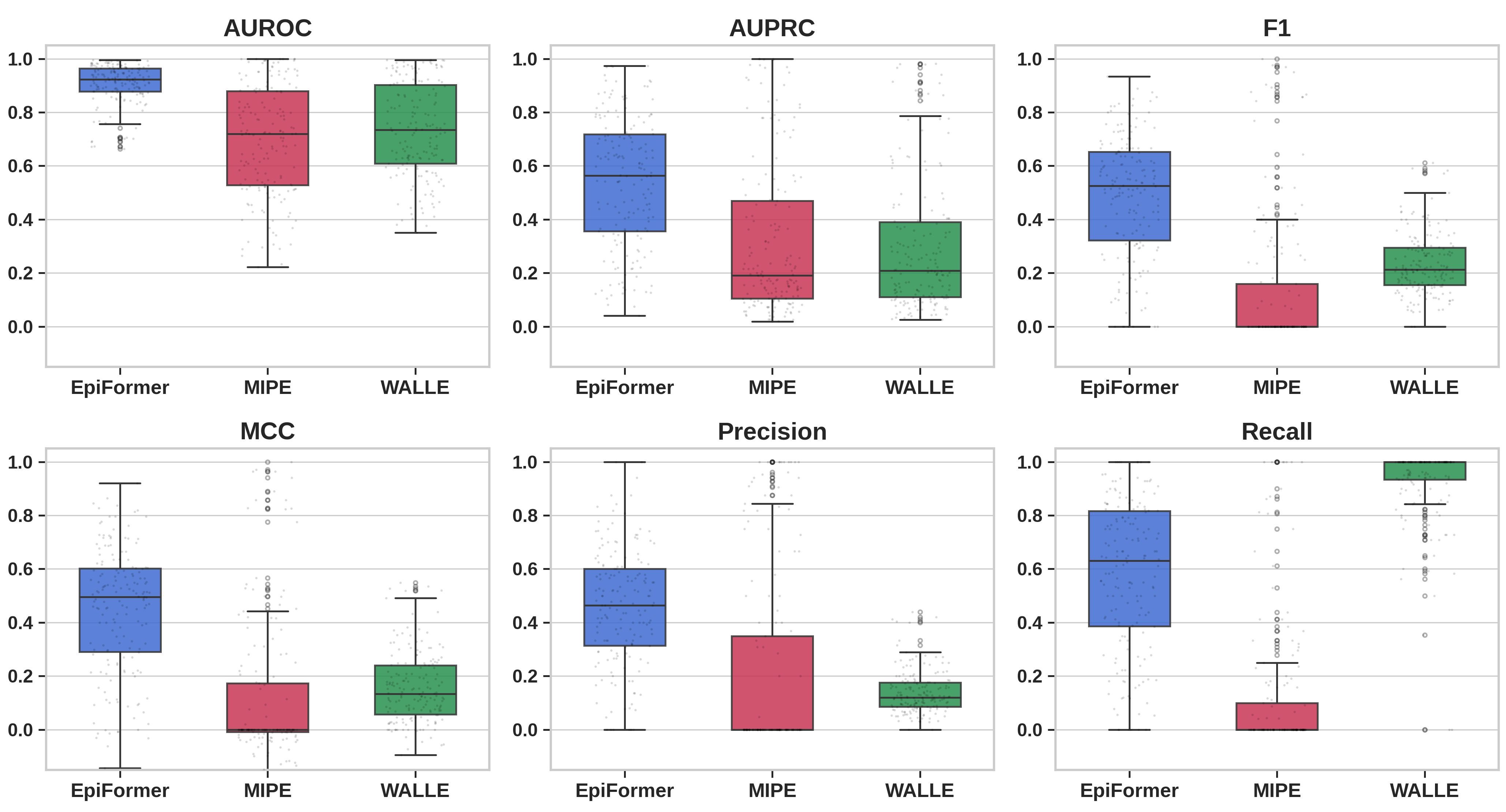}
\caption{Per-complex metric distributions on the epitope-ratio test set (170 complexes) for EpiFormer, MIPE, and WALLE.}
\label{fig:persample}
\end{figure}

\subsection{Distribution Shift Analysis}\label{appendix:distshift}

\begin{figure}[h!]
\centering
\includegraphics[width=0.75\textwidth]{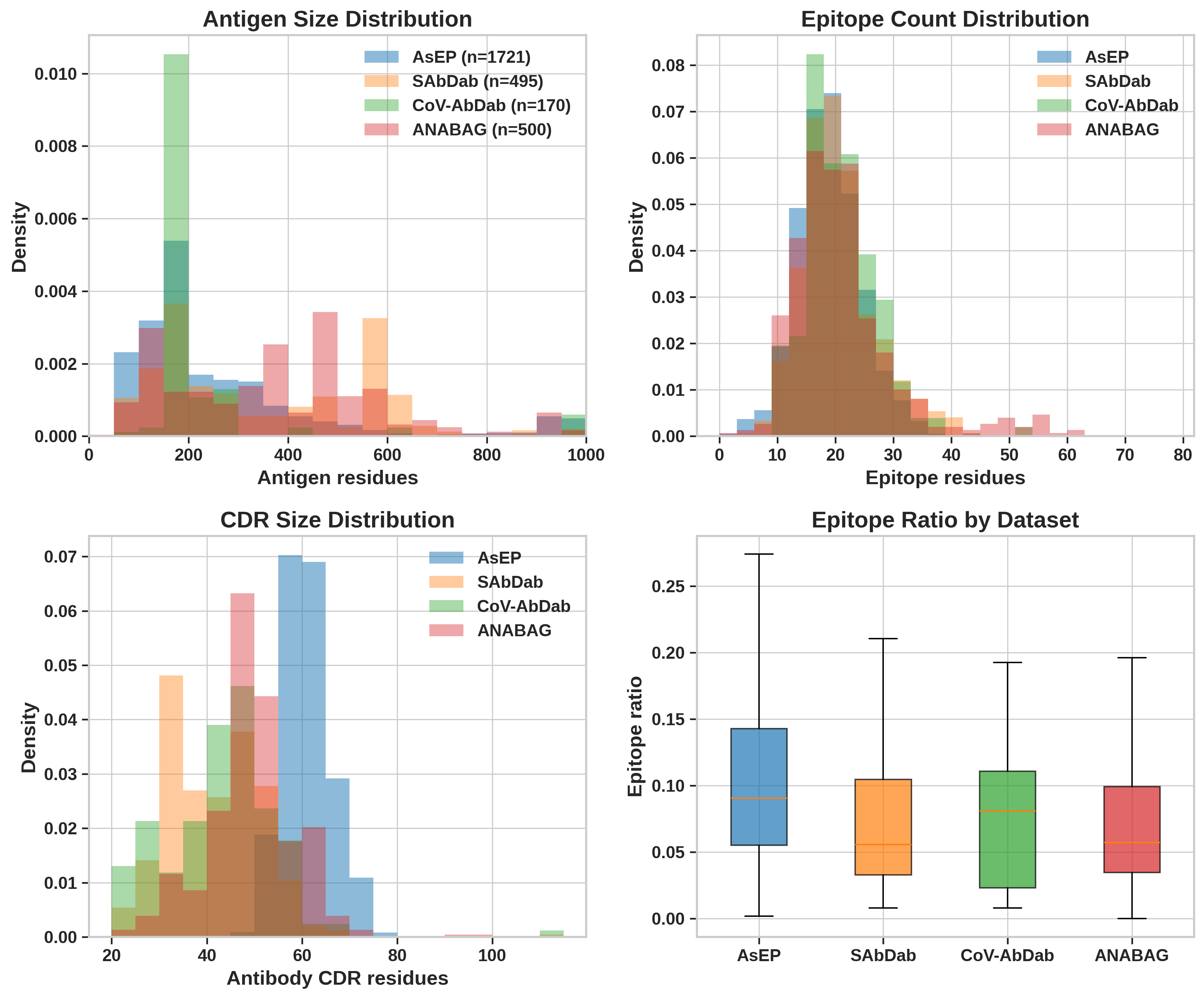}
\caption{Distribution shift between AsEP and three external benchmarks across (a) antigen size, (b) epitope count, (c) CDR size, and (d) epitope ratio. These differences explain the zero-shot calibration gap recovered by LODO fine-tuning.}
\label{fig:distshift}
\end{figure}

\subsection{Limitations}\label{appendix:limitations}

\textbf{Limitations.} While \emph{EpiFormer} maintains strong AUROC across all antigen sizes, metrics like AUPRC show more variation for very large antigens ($>$500 residues), where class imbalance is most severe.
Future work could explore SE(3)-equivariant alternatives~\citep{fuchs2020se3}, self-supervised pretraining~\citep{zhang2022gearnet}, and sequence-only variants using predicted structures.

\newpage

\section*{NeurIPS Paper Checklist}

\begin{enumerate}

\item {\bf Claims}
    \item[] Question: Do the main claims made in the abstract and introduction accurately reflect the paper's contributions and scope?
    \item[] Answer: \answerYes{}
    \item[] Justification: The abstract and introduction state failure modes.
    \item[] Guidelines:
    \begin{itemize}
        \item The answer \answerNA{} means that the abstract and introduction do not include the claims made in the paper.
        \item The abstract and/or introduction should clearly state the claims made, including the contributions made in the paper and important assumptions and limitations. A \answerNo{} or \answerNA{} answer to this question will not be perceived well by the reviewers. 
        \item The claims made should match theoretical and experimental results, and reflect how much the results can be expected to generalize to other settings. 
        \item It is fine to include aspirational goals as motivation as long as it is clear that these goals are not attained by the paper. 
    \end{itemize}

\item {\bf Limitations}
    \item[] Question: Does the paper discuss the limitations of the work performed by the authors?
    \item[] Answer: \answerYes{}
    \item[] Justification: Limitations are discussed in the Appendix.
    \item[] Guidelines:
    \begin{itemize}
        \item The answer \answerNA{} means that the paper has no limitation while the answer \answerNo{} means that the paper has limitations, but those are not discussed in the paper. 
        \item The authors are encouraged to create a separate ``Limitations'' section in their paper.
        \item The paper should point out any strong assumptions and how robust the results are to violations of these assumptions (e.g., independence assumptions, noiseless settings, model well-specification, asymptotic approximations only holding locally). The authors should reflect on how these assumptions might be violated in practice and what the implications would be.
        \item The authors should reflect on the scope of the claims made, e.g., if the approach was only tested on a few datasets or with a few runs. In general, empirical results often depend on implicit assumptions, which should be articulated.
        \item The authors should reflect on the factors that influence the performance of the approach. For example, a facial recognition algorithm may perform poorly when image resolution is low or images are taken in low lighting. Or a speech-to-text system might not be used reliably to provide closed captions for online lectures because it fails to handle technical jargon.
        \item The authors should discuss the computational efficiency of the proposed algorithms and how they scale with dataset size.
        \item If applicable, the authors should discuss possible limitations of their approach to address problems of privacy and fairness.
        \item While the authors might fear that complete honesty about limitations might be used by reviewers as grounds for rejection, a worse outcome might be that reviewers discover limitations that aren't acknowledged in the paper. The authors should use their best judgment and recognize that individual actions in favor of transparency play an important role in developing norms that preserve the integrity of the community. Reviewers will be specifically instructed to not penalize honesty concerning limitations.
    \end{itemize}

\item {\bf Theory assumptions and proofs}
    \item[] Question: For each theoretical result, does the paper provide the full set of assumptions and a complete (and correct) proof?
    \item[] Answer: \answerYes{}
    \item[] Justification: Full proof the equivariance is provided in the Appendix.
    \item[] Guidelines:
    \begin{itemize}
        \item The answer \answerNA{} means that the paper does not include theoretical results. 
        \item All the theorems, formulas, and proofs in the paper should be numbered and cross-referenced.
        \item All assumptions should be clearly stated or referenced in the statement of any theorems.
        \item The proofs can either appear in the main paper or the supplemental material, but if they appear in the supplemental material, the authors are encouraged to provide a short proof sketch to provide intuition. 
        \item Inversely, any informal proof provided in the core of the paper should be complemented by formal proofs provided in appendix or supplemental material.
        \item Theorems and Lemmas that the proof relies upon should be properly referenced. 
    \end{itemize}

    \item {\bf Experimental result reproducibility}
    \item[] Question: Does the paper fully disclose all the information needed to reproduce the main experimental results of the paper to the extent that it affects the main claims and/or conclusions of the paper (regardless of whether the code and data are provided or not)?
    \item[] Answer: \answerYes{}
    \item[] Justification: Experiments section describes the dataset, splits, and evaluation protocol. Appendix provides full architectural details and the algorithm. 
    \item[] Guidelines:
    \begin{itemize}
        \item The answer \answerNA{} means that the paper does not include experiments.
        \item If the paper includes experiments, a \answerNo{} answer to this question will not be perceived well by the reviewers: Making the paper reproducible is important, regardless of whether the code and data are provided or not.
        \item If the contribution is a dataset and\slash or model, the authors should describe the steps taken to make their results reproducible or verifiable. 
        \item Depending on the contribution, reproducibility can be accomplished in various ways. For example, if the contribution is a novel architecture, describing the architecture fully might suffice, or if the contribution is a specific model and empirical evaluation, it may be necessary to either make it possible for others to replicate the model with the same dataset, or provide access to the model. In general. releasing code and data is often one good way to accomplish this, but reproducibility can also be provided via detailed instructions for how to replicate the results, access to a hosted model (e.g., in the case of a large language model), releasing of a model checkpoint, or other means that are appropriate to the research performed.
        \item While NeurIPS does not require releasing code, the conference does require all submissions to provide some reasonable avenue for reproducibility, which may depend on the nature of the contribution. For example
        \begin{enumerate}
            \item If the contribution is primarily a new algorithm, the paper should make it clear how to reproduce that algorithm.
            \item If the contribution is primarily a new model architecture, the paper should describe the architecture clearly and fully.
            \item If the contribution is a new model (e.g., a large language model), then there should either be a way to access this model for reproducing the results or a way to reproduce the model (e.g., with an open-source dataset or instructions for how to construct the dataset).
            \item We recognize that reproducibility may be tricky in some cases, in which case authors are welcome to describe the particular way they provide for reproducibility. In the case of closed-source models, it may be that access to the model is limited in some way (e.g., to registered users), but it should be possible for other researchers to have some path to reproducing or verifying the results.
        \end{enumerate}
    \end{itemize}

\item {\bf Open access to data and code}
    \item[] Question: Does the paper provide open access to the data and code, with sufficient instructions to faithfully reproduce the main experimental results, as described in supplemental material?
    \item[] Answer: \answerYes{}
    \item[] Justification: The source code is provided as supplementary material during the review period for reproducibility. Code and data will be released publicly upon acceptance. 
    \item[] Guidelines:
    \begin{itemize}
        \item The answer \answerNA{} means that paper does not include experiments requiring code.
        \item Please see the NeurIPS code and data submission guidelines (\url{https://neurips.cc/public/guides/CodeSubmissionPolicy}) for more details.
        \item While we encourage the release of code and data, we understand that this might not be possible, so \answerNo{} is an acceptable answer. Papers cannot be rejected simply for not including code, unless this is central to the contribution (e.g., for a new open-source benchmark).
        \item The instructions should contain the exact command and environment needed to run to reproduce the results. See the NeurIPS code and data submission guidelines (\url{https://neurips.cc/public/guides/CodeSubmissionPolicy}) for more details.
        \item The authors should provide instructions on data access and preparation, including how to access the raw data, preprocessed data, intermediate data, and generated data, etc.
        \item The authors should provide scripts to reproduce all experimental results for the new proposed method and baselines. If only a subset of experiments are reproducible, they should state which ones are omitted from the script and why.
        \item At submission time, to preserve anonymity, the authors should release anonymized versions (if applicable).
        \item Providing as much information as possible in supplemental material (appended to the paper) is recommended, but including URLs to data and code is permitted.
    \end{itemize}

\item {\bf Experimental setting/details}
    \item[] Question: Does the paper specify all the training and test details (e.g., data splits, hyperparameters, how they were chosen, type of optimizer) necessary to understand the results?
    \item[] Answer: \answerYes{}
    \item[] Justification: Experiments section describes the benchmark, data splits, and evaluation metrics. 
    \item[] Guidelines:
    \begin{itemize}
        \item The answer \answerNA{} means that the paper does not include experiments.
        \item The experimental setting should be presented in the core of the paper to a level of detail that is necessary to appreciate the results and make sense of them.
        \item The full details can be provided either with the code, in appendix, or as supplemental material.
    \end{itemize}

\item {\bf Experiment statistical significance}
    \item[] Question: Does the paper report error bars suitably and correctly defined or other appropriate information about the statistical significance of the experiments?
    \item[] Answer: \answerYes{}
    \item[] Justification: The main results tables report mean $\pm$ standard deviation across test complexes for all metrics for random seeds. 
    \item[] Guidelines:
    \begin{itemize}
        \item The answer \answerNA{} means that the paper does not include experiments.
        \item The authors should answer \answerYes{} if the results are accompanied by error bars, confidence intervals, or statistical significance tests, at least for the experiments that support the main claims of the paper.
        \item The factors of variability that the error bars are capturing should be clearly stated (for example, train/test split, initialization, random drawing of some parameter, or overall run with given experimental conditions).
        \item The method for calculating the error bars should be explained (closed form formula, call to a library function, bootstrap, etc.)
        \item The assumptions made should be given (e.g., Normally distributed errors).
        \item It should be clear whether the error bar is the standard deviation or the standard error of the mean.
        \item It is OK to report 1-sigma error bars, but one should state it. The authors should preferably report a 2-sigma error bar than state that they have a 96\% CI, if the hypothesis of Normality of errors is not verified.
        \item For asymmetric distributions, the authors should be careful not to show in tables or figures symmetric error bars that would yield results that are out of range (e.g., negative error rates).
        \item If error bars are reported in tables or plots, the authors should explain in the text how they were calculated and reference the corresponding figures or tables in the text.
    \end{itemize}

\item {\bf Experiments compute resources}
    \item[] Question: For each experiment, does the paper provide sufficient information on the computer resources (type of compute workers, memory, time of execution) needed to reproduce the experiments?
    \item[] Answer: \answerYes{}
    \item[] Justification: Appendix reports the GPU type (single NVIDIA H100) and training details. 
    \item[] Guidelines:
    \begin{itemize}
        \item The answer \answerNA{} means that the paper does not include experiments.
        \item The paper should indicate the type of compute workers CPU or GPU, internal cluster, or cloud provider, including relevant memory and storage.
        \item The paper should provide the amount of compute required for each of the individual experimental runs as well as estimate the total compute. 
        \item The paper should disclose whether the full research project required more compute than the experiments reported in the paper (e.g., preliminary or failed experiments that didn't make it into the paper). 
    \end{itemize}
    
\item {\bf Code of ethics}
    \item[] Question: Does the research conducted in the paper conform, in every respect, with the NeurIPS Code of Ethics \url{https://neurips.cc/public/EthicsGuidelines}?
    \item[] Answer: \answerYes{}
    \item[] Justification: The research uses publicly available structural data from the Protein Data Bank and SAbDab, involves no human subjects, and conforms to the NeurIPS Code of Ethics.
    \item[] Guidelines:
    \begin{itemize}
        \item The answer \answerNA{} means that the authors have not reviewed the NeurIPS Code of Ethics.
        \item If the authors answer \answerNo, they should explain the special circumstances that require a deviation from the Code of Ethics.
        \item The authors should make sure to preserve anonymity (e.g., if there is a special consideration due to laws or regulations in their jurisdiction).
    \end{itemize}

\item {\bf Broader impacts}
    \item[] Question: Does the paper discuss both potential positive societal impacts and negative societal impacts of the work performed?
    \item[] Answer: \answerYes{}
    \item[] Justification: The paper acknowledges the positive impact of improved computational antibody design for therapeutic development.
    \item[] Guidelines:
    \begin{itemize}
        \item The answer \answerNA{} means that there is no societal impact of the work performed.
        \item If the authors answer \answerNA{} or \answerNo, they should explain why their work has no societal impact or why the paper does not address societal impact.
        \item Examples of negative societal impacts include potential malicious or unintended uses (e.g., disinformation, generating fake profiles, surveillance), fairness considerations (e.g., deployment of technologies that could make decisions that unfairly impact specific groups), privacy considerations, and security considerations.
        \item The conference expects that many papers will be foundational research and not tied to particular applications, let alone deployments. However, if there is a direct path to any negative applications, the authors should point it out. For example, it is legitimate to point out that an improvement in the quality of generative models could be used to generate Deepfakes for disinformation. On the other hand, it is not needed to point out that a generic algorithm for optimizing neural networks could enable people to train models that generate Deepfakes faster.
        \item The authors should consider possible harms that could arise when the technology is being used as intended and functioning correctly, harms that could arise when the technology is being used as intended but gives incorrect results, and harms following from (intentional or unintentional) misuse of the technology.
        \item If there are negative societal impacts, the authors could also discuss possible mitigation strategies (e.g., gated release of models, providing defenses in addition to attacks, mechanisms for monitoring misuse, mechanisms to monitor how a system learns from feedback over time, improving the efficiency and accessibility of ML).
    \end{itemize}
    
\item {\bf Safeguards}
    \item[] Question: Does the paper describe safeguards that have been put in place for responsible release of data or models that have a high risk for misuse (e.g., pre-trained language models, image generators, or scraped datasets)?
    \item[] Answer: \answerNA{}
    \item[] Justification: The model provides predictions about antigen binding sites that require extensive wet-lab validation before any practical use. 
    \item[] Guidelines:
    \begin{itemize}
        \item The answer \answerNA{} means that the paper poses no such risks.
        \item Released models that have a high risk for misuse or dual-use should be released with necessary safeguards to allow for controlled use of the model, for example by requiring that users adhere to usage guidelines or restrictions to access the model or implementing safety filters. 
        \item Datasets that have been scraped from the Internet could pose safety risks. The authors should describe how they avoided releasing unsafe images.
        \item We recognize that providing effective safeguards is challenging, and many papers do not require this, but we encourage authors to take this into account and make a best faith effort.
    \end{itemize}

\item {\bf Licenses for existing assets}
    \item[] Question: Are the creators or original owners of assets (e.g., code, data, models), used in the paper, properly credited and are the license and terms of use explicitly mentioned and properly respected?
    \item[] Answer: \answerYes{}
    \item[] Justification: All baseline methods and datasets are cited with their original publications. SAbDab, the Protein Data Bank, AsEP dataset, and all baseline implementations are cited.
    \item[] Guidelines:
    \begin{itemize}
        \item The answer \answerNA{} means that the paper does not use existing assets.
        \item The authors should cite the original paper that produced the code package or dataset.
        \item The authors should state which version of the asset is used and, if possible, include a URL.
        \item The name of the license (e.g., CC-BY 4.0) should be included for each asset.
        \item For scraped data from a particular source (e.g., website), the copyright and terms of service of that source should be provided.
        \item If assets are released, the license, copyright information, and terms of use in the package should be provided. For popular datasets, \url{paperswithcode.com/datasets} has curated licenses for some datasets. Their licensing guide can help determine the license of a dataset.
        \item For existing datasets that are re-packaged, both the original license and the license of the derived asset (if it has changed) should be provided.
        \item If this information is not available online, the authors are encouraged to reach out to the asset's creators.
    \end{itemize}

\item {\bf New assets}
    \item[] Question: Are new assets introduced in the paper well documented and is the documentation provided alongside the assets?
    \item[] Answer: \answerNA{}
    \item[] Justification: The paper introduces a new model but does not release new datasets or pre-trained model assets at submission time. Code and model weights will be released upon acceptance.
    \item[] Guidelines:
    \begin{itemize}
        \item The answer \answerNA{} means that the paper does not release new assets.
        \item Researchers should communicate the details of the dataset\slash code\slash model as part of their submissions via structured templates. This includes details about training, license, limitations, etc. 
        \item The paper should discuss whether and how consent was obtained from people whose asset is used.
        \item At submission time, remember to anonymize your assets (if applicable). You can either create an anonymized URL or include an anonymized zip file.
    \end{itemize}

\item {\bf Crowdsourcing and research with human subjects}
    \item[] Question: For crowdsourcing experiments and research with human subjects, does the paper include the full text of instructions given to participants and screenshots, if applicable, as well as details about compensation (if any)?
    \item[] Answer: \answerNA{}
    \item[] Justification: This work does not involve crowdsourcing or research with human subjects.
    \item[] Guidelines:
    \begin{itemize}
        \item The answer \answerNA{} means that the paper does not involve crowdsourcing nor research with human subjects.
        \item Including this information in the supplemental material is fine, but if the main contribution of the paper involves human subjects, then as much detail as possible should be included in the main paper. 
        \item According to the NeurIPS Code of Ethics, workers involved in data collection, curation, or other labor should be paid at least the minimum wage in the country of the data collector. 
    \end{itemize}

\item {\bf Institutional review board (IRB) approvals or equivalent for research with human subjects}
    \item[] Question: Does the paper describe potential risks incurred by study participants, whether such risks were disclosed to the subjects, and whether Institutional Review Board (IRB) approvals (or an equivalent approval/review based on the requirements of your country or institution) were obtained?
    \item[] Answer: \answerNA{}
    \item[] Justification: This work does not involve human subjects.
    \item[] Guidelines:
    \begin{itemize}
        \item The answer \answerNA{} means that the paper does not involve crowdsourcing nor research with human subjects.
        \item Depending on the country in which research is conducted, IRB approval (or equivalent) may be required for any human subjects research. If you obtained IRB approval, you should clearly state this in the paper. 
        \item We recognize that the procedures for this may vary significantly between institutions and locations, and we expect authors to adhere to the NeurIPS Code of Ethics and the guidelines for their institution. 
        \item For initial submissions, do not include any information that would break anonymity (if applicable), such as the institution conducting the review.
    \end{itemize}

\item {\bf Declaration of LLM usage}
    \item[] Question: Does the paper describe the usage of LLMs if it is an important, original, or non-standard component of the core methods in this research? Note that if the LLM is used only for writing, editing, or formatting purposes and does \emph{not} impact the core methodology, scientific rigor, or originality of the research, declaration is not required.
    %this research?
    \item[] Answer: \answerNA{}
    \item[] Justification: No LLMs are used as a core methodological contribution.
    \item[] Guidelines:
    \begin{itemize}
        \item The answer \answerNA{} means that the core method development in this research does not involve LLMs as any important, original, or non-standard components.
        \item Please refer to our LLM policy in the NeurIPS handbook for what should or should not be described.
    \end{itemize}

\end{enumerate}

\end{document}